\documentclass[11pt,a4paper,dvipsnames]{article}

\newcommand{\ignore}[1]{}
\usepackage{fullpage}

\usepackage[utf8]{inputenc}
\usepackage{graphicx}
\usepackage{amsfonts}
\usepackage{complexity}
\usepackage{amsmath,amssymb}
\usepackage{amsthm}
\usepackage{subfigure}
\usepackage[usenames,dvipsnames]{xcolor}
\usepackage{fancyhdr} 
\usepackage{algorithm, algpseudocode, algorithmicx}
\usepackage{thmtools}
\usepackage{tikz}
\usepackage{soul}
\usetikzlibrary{positioning}
\usetikzlibrary{matrix, arrows}

\usepackage{tikz,tikz-cd}
\tikzset{node distance=1.1cm and 1.1cm, on grid, semithick, shorten <=2pt, shorten >=2pt,
state/.style ={circle,black, fill, inner sep=0pt, minimum width = 0.3 em}}
\tikzset{every node/.style={circle, fill=white, inner sep = 0.1 em}}

\newcommand{\bbP}{\mathbb P}

\newcommand{\N}{\mathbb N}
\newcommand{\Z}{\mathbb Z}

\renewcommand{\C}{\mathbb C}

\newcommand{\F}{\mathbb F}
\newcommand{\x}{{\mathbf x}}
\newcommand{\y}{{\mathbf y}}
\newcommand{\z}{{\mathbf z}}
\newcommand{\e}{{\mathbf e}}

\renewcommand{\a}{{\mathbf a}}

\newcommand{\cf}{\text{Coeff}}

\renewcommand{\sp}{\textup{sparsity}}

\newcommand{\calC}{{\mathcal C}}

\newcommand{\ROABP}{\textrm{ROABP}}

\newcommand{\size}{\mathsf{size}}

\usepackage{nccmath}

\newcommand{\SPS}{\Sigma\Pi\Sigma}

\renewcommand{\span}{\textup{sp}}

\newcommand{\SPSP}{\Sigma\Pi\Sigma\Pi}

\newcommand{\VBP}{\mathsf{VBP}}

\newtheorem{theorem}{Theorem}

\newtheorem{lemma}[theorem]{Lemma}
\newtheorem{question}{Open question}
\newtheorem{proposition}[theorem]{Proposition}

\newtheorem{claim}{Claim}
\newtheorem{corollary}[theorem]{Corollary}
\theoremstyle{definition}
\newtheorem{definition}[theorem]{Definition}
\newtheorem{example}[theorem]{Example}
\newtheorem{remark}{Remark}
\newtheorem{observation}{Observation}

\usepackage[hidelinks, pagebackref]{hyperref}
\usepackage[capitalise]{cleveref}

\usepackage{xcolor}
\hypersetup{
    colorlinks,
    linkcolor={blue!80!black},
    citecolor={green!50!black},
    urlcolor={red!40!black}
}

\newcommand{\val}{\textsf{val}}

\newcommand{\bbC}{\mathbb{C}}

\newcommand{\eps}{\varepsilon}
\renewcommand{\epsilon}{\varepsilon}

\newcommand{\image}{\operatorname{image}}

\newcommand{\Kc}{\textup{\textsf{Kc}}}

\newcommand{\WR}{\mathsf{WR}}
\newcommand{\bwr}{\underline{\mathsf{WR}}}
\newcommand{\ELL}{\boldsymbol{\ell}}

\newcommand{\dlog}{\mathsf{dlog}}
\newcommand{\SES}{\Sigma\wedge\Sigma}

\newcommand{\pderiv}[2]{\partial_{#2}\inparen{#1}}
\newcommand{\inparen }[1]{\left(#1\right)} 
\newcommand{\GenSES}{\mathsf{Gen}}

\newcommand{\SPSSES}[1]{ \Sigma^{#1} \inparen{\Pi\Sigma / \Pi\Sigma} \inparen{\SES / \SES} }

\newcommand{\SPSfanin}[2]{\Sigma^{#1}\Pi^{#2}\Sigma}

\newcommand{\rank}{\mathsf{rank}}
\newcommand{\vect}[1]{\mathbf{#1}}

\newcommand{\VPk}[1]{\mathsf{VBP}_{#1}} 
\newcommand{\VPe}{\mathsf{VF}} 
\newcommand{\VF}{\VPe}

\newcommand{\approxbar}[1]{\overline{#1}} 
\newcommand{\FF}{\mathbb{F}} 
\newcommand{\Oh}{\mathcal{O}} 
\newcommand{\coeffset}[3]{{#1}_{(#2,#3)}}

\newcommand{\xa}{\x^{\a}}  
\newcommand{\ya}{\y^{\a}}
\DeclareMathOperator{\Span}{span}  
\newcommand{\abs}[1]{\lvert #1 \rvert}  

\DeclareMathOperator{\spanning}{span}
\DeclareMathOperator{\depending}{depend}

\newcommand{\ARO}{\mathsf{ARO}}
\newcommand{\VW}{\mathsf{VW}}

\renewcommand{\det}{\textup{det}}
\newcommand{\sgn}{\textup{sgn}}
\newcommand{\per}{\textup{per}}
\usepackage{parskip}
\newcommand{\GL}{\operatorname{GL}}
\newcommand{\trace}{\textup{trace}}
\newcommand{\Sym}{\textup{Sym}}
\newcommand{\bP}{\mathbb{P}}
\newcommand{\bC}{\mathbb{C}}
\newcommand{\detdrei}{\det_3}

\newcommand{\End}{\mathrm{End}}
\newcommand\mop[1]{\operatorname{#1}}

\title{Recent Advances in Debordering Methods}

\author{\href{https://sites.google.com/view/pduttashomepage/home}{Pranjal Dutta} \thanks{Nanyang Technological University (NTU), Singapore \& Simons Institute for the Theory of Computing, UC Berkeley, USA. Email: \texttt{duttpranjal@gmail.com}} \and \href{https://qi.rub.de/lysikov}{Vladimir Lysikov} \thanks{Ruhr-Universit\"at Bochum, Bochum, Germany. Email: \texttt{vladimir.lysikov@rub.de}}}

\date{}

\begin{document}

\thispagestyle{empty}



\maketitle






%


\begin{abstract}

Border complexity captures functions that can be approximated by low-complexity ones. {\em Debordering} is the task of proving an upper bound on some non-border complexity measure in terms
of a border complexity measure, thus getting rid of limits. Debordering lies at the heart of foundational complexity theory questions relating Valiant's determinant versus permanent conjecture (1979) and its geometric complexity theory (GCT) variant proposed by Mulmuley and Sohoni (2001). The debordering of matrix multiplication tensors by Bini (1980) played a pivotal role in the development of efficient matrix multiplication algorithms. Consequently, debordering finds
applications in both establishing computational complexity lower bounds and facilitating algorithm
design. Recent years have seen notable progress in debordering various restricted border complexity measures. In this survey, we highlight these advances and discuss techniques underlying them.
\end{abstract}


\bigskip

{\small\textbf{2020 Math.~Subj.~Class.:}
Primary 68Q17; Secondary 68Q15, 14L30
}

\bigskip

{\small\textbf{ACM Computing Classification System:}
Theory of computation $\to$ Algebraic complexity theory; Circuit complexity; Complexity classes; Problems, reductions and completeness
} 

\tableofcontents

\section{Introduction}
A central goal of theoretical computer science is to understand the computational resources required to solve algorithmic problems. Computational complexity theory approaches this by organizing problems into complexity classes -- such as $\P$, $\NP$, and $\#\P$ -- and asking whether natural problems in one class can or cannot be efficiently computed using algorithms from another. The most famous of these questions is the $\mathsf{P}$ vs.\ $\mathsf{NP}$ problem.

In parallel to this Boolean setting, algebraic complexity theory studies the complexity of computing multivariate polynomials using arithmetic operations over a field. Rather than manipulating bits, the focus is on computing polynomials using only additions, multiplications, and constants. This algebraic viewpoint captures a wide range of problems in symbolic computation, algebraic geometry, invariant theory, and even numerical analysis.

\paragraph{Algebraic circuits.}~The fundamental model of computation here is the \emph{algebraic circuit}: a directed acyclic graph whose input nodes (nodes of in-degree zero) are labeled by variables $\{x_1,x_2,\ldots, x_n\}$ or the constants from the underlying field $\F$, the internal nodes labeled by `$+$' (addition gate) and `$\times$' (multiplication gate). Each edge is labeled by some field constant. Semantically, each node of the graph computes a polynomial on the input variables naturally. There is a single node of out-degree zero, called the output node of the circuit. The in-degree of a vertex is called its {\em fan-in} and out-degree
its {\em fan-out}. The {\em size} of an algebraic circuit is the size of the graph, which is the total number of nodes and edges.
The {\em depth} of the circuit is the length of the longest path from the root to a leaf node. Algebraic circuits can be assumed to be layered with alternating
layers of $+$ and $\times$ nodes. A circuit can compute a polynomial with exponentially large degree with respect to its size. For our purpose, we only focus on \emph{low degree} circuits, i.e.~the degree of the polynomial computed by the circuit is polynomially upper-bounded by the circuit size. One interesting polynomial that admits a polynomial-size circuit is the {\em symbolic determinant polynomial}, defined as follows:
\[
\text{det}_n\;:=\;\sum_{\sigma \in S_n}\sgn(\sigma)\prod_{i=1}^n x_{i,\sigma(i)}\;.
\]
\paragraph{Algebraic formulas and ABPs.}~Despite this elegant structural formulation, proving lower bounds against general arithmetic circuits remains a major open challenge. To make progress, researchers study restricted circuit models, such as formulas, ABPs, depth-bounded circuits, or bounded fan-in circuits. These restrictions allow us to understand how structural constraints affect expressiveness and open avenues for proving meaningful lower bounds. 
An {\em algebraic formula} is a circuit with the underlying structure to be a tree.  
On the other hand, every homogeneous degree $d$ polynomial $f$ can be written as a product
\[
f = \begin{pmatrix}\ell_{1,1,1} & \cdots & \ell_{1,n,1}\end{pmatrix}
\begin{pmatrix}
\ell_{1,1,2} &    \cdots   &   \ell_{1,n,2}     \\
  \vdots    & \ddots &   \vdots     \\
  \ell_{n,1,2}    &   \cdots    & \ell_{n,n,2}
\end{pmatrix}
\cdots
\begin{pmatrix}
\ell_{1,1,d-1} &    \cdots   &   \ell_{1,n,d-1}     \\
  \vdots    & \ddots &   \vdots     \\
  \ell_{n,1,d-1}    &   \cdots    & \ell_{n,n,d-1}
\end{pmatrix}
\begin{pmatrix}
\ell_{1,1,d} \\
  \vdots\\
  \ell_{n,1,d}
\end{pmatrix}
\]
of matrices whose entries are homogeneous linear polynomials.
We define $w(f)$ to be the smallest possible such $n$, and call it the \emph{homogeneous branching program width} of $f$.
For an inhomogeneous polynomial, we define $w(f) := \sum_{d\in\N}w(f_d)$ to be the sum of the widths of its homogeneous components. This notion is polynomially equivalent to the {\em determinantal complexity}: Given a polynomial $f$, its determinantal complexity $\textsf{dc}(f)$, is $m$ if it is the smallest size of a matrix $A$ whose entries are affine linear polynomials such that $\det(A)= f$. For more details and properties on these classes, we refer to~\cref{sec:prelimI}, and these beautiful surveys \cite{Shpilka10,mahajan2014algebraic}.

\paragraph{Algebraic complexity classes.}~Introduced by Valiant~\cite{Valiant79}, the class $\mathsf{VP}$ consists of families of polynomials ${f_n}$ over a field $\mathbb{F}$ that can be computed by arithmetic circuits of size and degree bounded by a polynomial in $n$. $\VPe$ consists of families of polynomials ${f_n}$ over a field $\mathbb{F}$ that can be computed by an algebraic formula of size bounded by a polynomial in $n$. $\VBP$ consists such polynomial families whose $w$, or the determinantal complexity $\textsf{dc}$ is polynomially bounded. The class $\mathsf{VNP}$ generalizes this by allowing an exponential sum (in fact {\em hypercube} sum) over polynomially computable polynomials: A polynomial family $(f_n) \in \VNP$, if there exists $(g_r) \in \VP$, such that 
\[
f(\x)\;=\;\sum_{\a \in \{0,1\}^m}\;g_r(\x,\a)\;.
\]
A natural complete problem for $\VNP$ is the {\em symbolic permanent polynomial}:
\[
\;\;\per_m:=\sum_{\sigma \in S_m}\prod_{i=1}^m x_{i,\sigma(i)}\;.
\]
Equivalently, $(f_n) \in \VNP$ if its {\em permanental complexity} $\textsf{pc}(f_n)$: the smallest
size of a matrix $A$ whose entries are affine linear polynomials such that $f_n=\per_m(A)$, is polynomially bounded. The class $\VNP$ plays the role of $\mathsf{NP}$ in this algebraic setting. 

The problem of separating algebraic complexity classes has been a central theme of this study. It is known that $\VF \subseteq \VBP \subseteq \VP \subseteq \VNP$~\cite{Valiant79, Toda92}.
The conjectures $\VF \neq \VNP$, $\VBP \neq \VNP$, $\VP \neq \VNP$, are known as \emph{Valiant's conjectures}. Especially $\VNP \not \subseteq \VBP$ is known as the \emph{determinant vs permanent} problem, which asks to prove the following statement: $\textsf{dc}(\per_m)$ is {\em not polynomially} bounded. Whereas, the $\mathsf{VP}$ vs.\ $\mathsf{VNP}$ problem asks to prove the following statement: $\size(\per_m)$ is {\em not polynomially} bounded. These questions naturally mirror the Boolean $\mathsf{P}$ vs.\ $\mathsf{NP}$ question. It is known that the nonuniform $\P \ne \NP$ conjecture implies $\VP \ne \VNP$, assuming Generalized Riemann Hypothesis (GRH)~\cite{burgisser2000cook}.


\subsection{Geometric Complexity Theory}
Over the years, impressive progress has been made towards resolving Valiant's conjectures, however, the existing tools have not been able to resolve this conclusively. Mulmuley and Sohoni strengthened the conjecture~\cite{MS01} by allowing the permanent to be approximated arbitrarily closely coefficientwise instead of being computed exactly. The hope in the Geometric Complexity Theory (GCT) program is to use available tools from algebraic geometry and representation theory, and possibly settle the question once and for all. 

The Mulmuley–Sohoni conjecture can be stated in terms of group orbit closures as $\ell^{n-m}\per_m \not \in \overline{\GL_{n^2}\det_{n}}$, if $n=\poly(m)$; here $\GL_{n^2}:=\GL(\C^{n\times n})$ acts on the space
of homogeneous degree $n$ polynomials in $n^2$ variables by (invertible) linear transformations of the variables\footnote{For a homogeneous polynomial $p$ and $g\in\GL_{n^2}$ define the homogeneous polynomial $gp$ via $(gp)(\x) := p(g^t \x)$. The orbit is defined as $\GL_{n^2}p := \{gp \mid g \in \GL_{n^2}\}$.}, $\ell$ is some homogeneous linear polynomial (one can assume $\ell:=x_{1,1}$),
and the closure can be taken equivalently in the Zariski or the Euclidean topology, see e.g.~\cite[AI.7.2 Folgerung]{Kra85}. %
The polynomial $\ell^{m-n}\per_n$ is called the `padded permanent', and the phenomenon of multiplying with a power of a linear form is called \emph{padding}.
Note here that the action of $\GL_{n^2}$ replaces variables by \emph{homogeneous} linear polynomials. One could formulate this setup without padding, but then the reductive group $\GL_{n^2}$ would have to be replaced by the general affine group (see e.g.~\cite{MS21}), which is \emph{not} a reductive group.
For reductive groups, every representation decomposes into a direct sum of irreducible representations.
This is important for the representation-theoretic attack proposed in \cite{MS01,MS08}, hence the padding is introduced in those papers.
The idea is that $\ell^{n-m}\per_m \in \overline{\GL_{n^2}\det_{n}}$ if and only if
$\overline{\GL_{n^2}\ell^{n-m}\per_m} \subseteq \overline{\GL_{n^2}\det_{n}}$.
Such an inclusion induces a $\GL_{n^2}$-equivariant surjection between the coordinate rings and between their homogeneous degree $\delta$ components, see e.g.~\cite{BLMW11}:
$
\C[\overline{\GL_{n^2}\det_{n}}]_\delta \;\twoheadrightarrow \;\C[\overline{\GL_{n^2}\ell^{n-m}\per_m}]_\delta.
$

Outside $\VP$ vs.~$\VNP$ implication, GCT has deep connections with computational invariant theory \cite{forbes2013explicit,mulmuley2012gct,garg2016deterministic, burgisser2018alternating,DBLP:journals/cc/IvanyosQS17}, algebraic natural proofs \cite{grochow2017towards,DBLP:conf/soda/BlaserILPS21,chatterjee2020existence,kumar2020if,berg2024algebraic}, lower bounds \cite{burgisser2013explicit, grochow2015unifying, LO15}, optimization \cite{allen2018operator,burgisser2019towards} and many more. We refer to~\cite[Sec.~9]{BLMW11} and \cite{mulmuley2012gct,Mul12} for expository references.

\subsection{Border Complexity}
The complexity notions mentioned above, such as formula size, circuit size, width $w$, and the permanental complexity, have an associated \emph{border complexity} variant: A polynomial has border complexity $\leq k$ if it is the limit of polynomials of complexity at most $k$. Here, the limit is taken in the Euclidean topology on the coefficient vector space, see e.g.~\cite{IS22}. This notion is also equivalent to taking the Zariski closure; we will discuss it in detail in~\cref{sec:bordercomplexity}.
Border complexity measures are usually indicated by an underlined symbol: e.g., $\underline{w}$ is the border homogeneous algebraic branching program width.
Clearly $\underline{w}(p)\leq w(p)$ for all polynomials $p$.


Border complexity is an old area of study in algebraic geometry. In theoretical computer
science it was introduced in \cite{BCRL79,Bin80} in the context of fast matrix
multiplication, and later studied in~\cite{CW90,LO15}. In algebraic complexity theory, border complexity was first discussed independently in~\cite{MS01,Bur04}. 

\paragraph{Connection to matrix multiplication.}~It turns out that understanding border Waring rank would lead to designing faster matrix multiplication algorithms. For a homogeneous polynomial $f$, its \emph{Waring rank} $\WR(f)$ is defined as the smallest number $k$, such that $f = \sum_{i=1}^k \ell_i^d$, where $\ell_i$ are homogeneous linear polynomials over $\C$. Its {\em border Waring rank} $\bwr(f)$ is defined as the smallest $k$ such that $f = \lim_{\eps \to 0} \sum_{i=1}^k \ell_i^d$, where $\ell_i$ are homogeneous linear polynomials (in $\x$) over $\C(\eps)$. For more details, see~\cref{sec:waring}. 

On the other hand, the matrix multiplication exponent is defined as 
\[
\omega = \inf\{ \tau : \text{ two $n \times n$ matrices can be multiplied using $O(n^\tau)$ scalar multiplications}\}.
\]
This fundamental constant can be defined in terms of the tensor rank and the tensor border rank of the matrix multiplication tensor \cite{Strassen:Gauss_elimination_not_optimal,BCRL79}. Let~$X_n := (x_{ij})_{i,j = 1 \cdots n}$ be a matrix of variables. Then, $\trace(X_n^3)$, the trace polynomial of the matrix $X_n^3$, is a homogeneous degree 3 polynomial in $n^2$ variables. The results of \cite{CHIL18} show the following.
\[
\omega \;=\; \lim_{n \to \infty} \log_n  \bwr( \trace(X_n^3))\;.\] 

\paragraph{Advantage using border complexity.}~
Let $f \in \F[x]$ be a degree $d$ univariate polynomial. Observe that even if $r$ is very small compared to $d$, interpolating the coefficient of $x^r$
in $f(x)$ {\em requires} $d + 1$
many evaluations. Interestingly, in the border, the situation is quite different, since we can express the coefficient of $x^r$ in $f(x)$ as the limit of a
sum of just $r + 1$ evaluations of $f$! The following lemma states this formally; and the proof works even when $f$ is a formal power series.

\medskip
\begin{lemma}[Border Interpolation]
Let $R$ be a commutative ring that contains a field $\F$ of at least $r + 1$ elements, and let $\alpha_0, \cdots, \alpha_r$ be distinct elements in $\F$. F. Then, there exists fields elements $\beta_0, \cdots, \beta_r$ such that
for any $\sum_{i \ge 0} f_i x^i =: f(x) \in R[[x]]$, we have
\[
f_r \;=\; \lim_{\eps \to 0}\; \frac{1}{\eps^r}\,\left(\sum_{i=0}^r \beta_i \cdot f(\eps \alpha_i)\right)\;.
\]
\end{lemma}
\begin{proof}[Proof sketch]
Let $g(x) := \sum_{i=0}^r f_i x^i$, and $h := f-g$. Interpolation on $\epsilon \cdot \alpha_i$ shows that there exist constants $\beta_0, \cdots, \beta_r$ such that $\eps^r f_r = \sum_{i=0}^r \beta_i \cdot g(\eps \alpha_i)$. Since, $\lim_{\eps \to 0} \frac{1}{\eps^r} h(\eps \alpha_i) = 0$, the conclusion follows.
\end{proof}

\paragraph{Other connections and importance.}~A central question in the GCT program is whether the class $\VP$ is {\em closed under approximation}, that is, whether $\approxbar{\VP} = \VP$~\cite{Mul12, DBLP:conf/icalp/GrochowMQ16}. Resolving this question -- either by proving or disproving it -- would have significant consequences in both algebraic complexity and algebraic geometry. If $\VP = \overline{\VP}$, then any proof separating $\VP$ from $\VNP$ would automatically imply that $\VNP \not\subseteq \overline{\VP}$, as conjectured in\cite{mulmuley2012gct}. On the other hand, if this closure fails, then any approach to separating $\VP$ from $\VNP$ must first separate the permanent from members of $\overline{\VP} \setminus \VP$, a task that currently appears well beyond reach. Moreover, a long line of depth reduction results~\cite{vsbr83, Agrawal08, Koi12, Tav13, Gupta16} and bootstrapping phenomena~\cite{AGS19, KST19, GKSS19, And20} show that debordering restricted models, such as the border of depth-3 or depth-4 circuits, is equally intriguing and interesting.

Debordering results are also closely tied to the {\em flip principle} in GCT~\cite{mulmuley2010geometric, mulmuley2012gct}, which emphasizes understanding upper bounds first, to leverage that understanding to eventually prove lower bounds. Even for restricted models, such as depth-3 circuits or small-width ABPs, establishing debordering results can have significant implications, especially for problems like derandomizing Polynomial Identity Testing (PIT). Derandomizing PIT is equivalent to constructing small, explicit hitting sets for $\VP$. This has far-reaching implications across computational mathematics, including but not limited to, graph algorithms~\cite{Lov79, MVV87, FGT19}, polynomial factoring~\cite{KSS14, 10.1145/3510359, bhargav2025primer,bhattacharjee2025closurefactorizationresultfurstenberg}, cryptography~\cite{Agrawal04}, and foundational results in hardness-vs-randomness \cite{HS80, nw94, Agrawal05, KI03, DSY09, DST21}. 
We refer readers to~\cite{Shpilka10, Sax14, kumar2019hardness,10.1145/3674159.3674165} for excellent surveys on these topics.

Going beyond $\VP$, PIT for the border class $\overline{\VP}$ is particularly significant due to its deep connections with algebraic geometry, as observed by Mulmuley~\cite{Mul12}. For example, the computational version of Noether's Normalization Lemma (NNL) -- a foundational result in algebraic geometry -- can be reduced to constructing explicit hitting sets for $\approxbar{\VP}$\cite{Mul12, forbes2013explicit}. In fact, certain formulations of NNL derandomization are known to be equivalent to proving explicit lower bounds\cite{Mul12, Mukh16}. The construction of a robust and explicit hitting set for $\VP$~\cite{FS18, GSS19} remains an important open goal in its own right. Naturally, strong debordering results would serve as a crucial tool in constructing such hitting sets for a wide range of algebraic models.

Debordering has found applications in the context of polynomial factoring as well, see~\cite{burgisser2000completeness,Bur04,10.1145/3510359,bhargav2024learning}.

\subsection{Purpose of This Survey}
This survey aims to provide a compact yet comprehensive overview of the most significant results related to debordering in algebraic complexity. We organize the known literature into two main categories: 

\begin{enumerate}
    \item {\bf Algebraic Characterizations of Border Complexity} (\cref{sec:bordercomplexity}):
Here, we present and sketch proofs of various characterizations of border complexity that offer a more algebraic viewpoint—serving as alternatives to the standard topological definition.
\item  {\bf Debordering Results} (\cref{sec:debordering}):
This section focuses on the known structural weaknesses of various circuit models and outlines the corresponding debordering results, organized by key techniques and frameworks.
\end{enumerate}
Within these two themes, we attempt to cover the major developments while highlighting common proof strategies and underlying frameworks. To make the survey accessible, we include the necessary background in the preliminaries. The exposition is written to require minimal prior knowledge, making it approachable to readers with a basic understanding of algebraic complexity. 

In several instances, we provide more detailed proof sketches than usual -- particularly for results such as the $\varepsilon$-degree of approximation discussed in~\cref{sec:bordercomplexity}, and the debordering of the border of depth-3 circuits with fan-in 2 discussed in~\cref{sec:debordering}. These proofs have long intrigued readers due to their subtle and intricate nature, but on many separate past occasions, we were told that the existing presentations are seemingly not-too-helpful in conveying the underlying ideas. We therefore felt it useful to provide reasonably detailed arguments for clarity and completeness.

\section{Preliminaries} \label{sec:prelimI}

\paragraph{Notation:} 
\begin{itemize}
    \item For a positive integer $k$, $[k]$ denotes the set of positive integers $\{1,2,\ldots,k\}$.
    \item For a finite set $S$, $|S|$ denotes the cardinality of the set $S$.
    \item We use boldface letters such as $\x$ to refer to an order tuple of variables such as $(x_1, \cdots, x_n)$. The size of the tuple would usually be clear from context
    \item We use bold-face letters (such as $\F, \C$) to denote fields. We use $\F[x]$ to denote the polynomial ring, $\F[x^{\pm 1}]$ to denote the ring of Laurent polynomials, $\F[[x]]$ to denote the ring of formal power series, and $\F((x))$ to denote the ring of formal Laurent series with respect to the variable $x$ with coefficients from the field $\F$. Throughout, we will work with $\F=\C$, the field of complex numbers, unless specified otherwise. 
    \item $\C[\x]_d$ denotes the set of homogeneous degree-$d$ polynomials, while $\C[\x]_{\le d}$ denotes the set of polynomials of degree at most $d$. In particular, $\C[\x]_1$ contains all the homogeneous linear forms $a_1x_1 + \cdots + a_n x_n$, where $a_i \in \C$. 
    \item  For an $\a=(a_1,a_2,\ldots,a_n)\in\Z_{\ge 0}^n$, $\x^\a$ denotes the monomial $\prod_{i=1}^nx_i^{a_i}$.
    \item Let~$A(\x)$ be a polynomial over~$\F$ in~$n$ variables. A polynomial $A(\x)$ is said to have \emph{individual degree} $d$,  if the degree of each variable is bounded by $d$ for each monomial in $A(\x)$. When~$A(\x)$ has individual degree~$d$, then the exponent~$\a$ of any monomial~$\xa$ of~$A(\x)$
is in the set
\[
 \M = \{0,1, \dots, d\}^n \,.
\]
\item By~$\cf_{\xa}(A) \in \F$ we denote the coefficient of the monomial~$\xa$ in~$A(\x)$.
Hence, we can write
\[
 A(\x) = \sum_{\a \in \M} \cf_{\xa}(A)\, \xa \,.
\]
The {\em sparsity\/} of polynomial~$A(\x)$ is the number of nonzero coefficients~$\cf_{\xa}(A)$.
\item {\bf Coefficient space.}~We also consider {\em matrix polynomials\/}
where the coefficients~$\cf_{\xa}(A)$ are $w \times w$ matrices, for some~$w$.
In an abstract setting,
these are polynomials over a $w^2$-dimensional $\F$-algebra of matrices~$\F^{w \times w}$.
The {\em coefficient space\/} is then defined as the span of all coefficients
of~$A$,
i.e., 
$\Span_{\F}\{\cf_{\xa}(A) \mid \a \in \M \}$,

Consider a partition of the variables~$\x$ into two parts~$\y$ and~$\z$, with $\abs{\y}=k$.
A polynomial~$A(\x)$ can be viewed as a polynomial in  variables~$\y$, 
where the coefficients are polynomials in~$\F[\z]$.
For monomial~$\ya$,
let us denote the coefficient of~$\ya$ in~$A(\x)$ by $\coeffset{A}{\y}{\a} \in \F[\z]$.
For example, in the polynomial $A(\x) = x_1 + x_1x_2 + {x_1}^2$, we have
$\coeffset{A}{x_1}{1} = 1 + x_2$, whereas
$\cf_{x_1}(A) = 1$.
Observe that $\cf_{\ya}(A)$ is the constant term in $\coeffset{A}{\y}{\a}$.

Thus, $A(\x)$ can be written as 
\begin{equation}\label{eq:C_ya}
A(\x) = \sum_{\a \in \{0,1, \dots,d\}^k} \coeffset{A}{\y}{\a} \, \ya \,.
\end{equation}
The coefficient $\coeffset{A}{\y}{\a}$ is also sometimes expressed 
in the literature as a partial derivative~$\frac{\partial A}{\partial \ya} $ 
evaluated at $\y = {\boldsymbol 0}$
(and multiplied by an appropriate constant), see~\cite[Section 6]{Forbes13-greybox-hs}.
\item For a set of polynomials~$\mathcal{P}$,
we define their $\F$-$\Span$ as
\[
\Span_{\F} \mathcal{P} = \left\{\sum_{A \in \mathcal{P}}
\alpha_A A \mid \alpha_A \in \F \text{ for all } A \in \mathcal{P}\right\}.
\]
The set of polynomials~$\mathcal{P}$ is said to be $\F$-{\em linearly independent\/}
if $\sum_{A \in \mathcal{P}} \alpha_A A = 0$ holds only for $\alpha_A = 0$, 
for all~$A \in \mathcal{P}$.
The \emph{dimension\/}~$\dim_\F \mathcal{P}$ of~$\mathcal{P}$ is
the cardinality of the largest $\F$-linearly independent subset of~$\mathcal{P}$.
\item {\bf Valuation.}~For any $g \in \C[\epsilon^{\pm 1}][\x]$, one can define $\val_{\epsilon}(g)$ as the minimum exponent of $\epsilon$ appearing in $g$. Clearly, $\lim_{\epsilon \to 0} g$ exists if and only if $\val_{\epsilon}(g) \ge 0$. We also assume that $\val_{\epsilon}(0) = +\infty$. 

\item {\bf Equivalence Relation.}~We introduce an equivalence relation of \emph{approximate equality} on~$\C[\eps^{\pm 1}][\x]$: given two polynomials $f_1,f_2$ whose coefficients depend rationally on $\eps$, we write $f_1 \simeq f_2$, iff $\lim_{\epsilon \to 0} f_1$ and $\lim_{\epsilon \to 0} f_2$ are both finite and they coincide. We often use this notation with either $f_1$ or $f_2$ not depending on $\eps$: if, for instance, $f_1$ does not depend on $\eps$, then $f_1 \simeq f_2$ means that $f_2 = f_1  + O(\eps)$.

\end{itemize}  

\subsection{Arithmetic branching programs}
\label{sec:abp}
An {\em arithmetic branching program\/} (ABP) is a directed graph 
with $\ell+1$ layers of vertices $V_0, \ldots, V_{\ell}$.
The layers $V_0$ and $V_\ell$ each contain only one vertex, 
the {\em start node\/}~$v_{0}$ and 
the {\em end node\/}~$v_{\ell}$, respectively.
The edges are 
only going from the vertices in the layer $V_{i-1}$ to the vertices in the layer $V_i$, 
for any $i \in [\ell]$.
All the edges in the graph have weights from~$\F[\x]$,
for some field~$\F$. 
The {\em length\/} of an ABP is the length of a longest path in the ABP, i.e.~$\ell$.
An ABP has {\em width\/}~$w$, 
if $\abs{V_i} \leq w$ for all $1 \leq i \leq \ell-1$.

For an edge~$e$, let us denote its weight by~$W(e)$.
For a path~$p$,
its weight~$W(p)$ is defined to be the product of weights of all the edges
in it,
\[ W(p) = \prod_{e \in p} W(e).\]
The {\em polynomial $A(\x)$ computed by the ABP\/}
is the sum of the weights of all the paths from $v_{0}$ to $v_{\ell}$,
\[
A(\x) = \sum_{p \text{ path } v_{0} \leadsto v_{\ell}} W(p).
\]

Let the set of nodes in $V_i$ be $\{v_{i,j} \mid j \in [w]\}$.
The branching program can alternately be represented by a matrix product
$\prod_{i=1}^{\ell} D_i  $,
where $D_1 \in \F[\x]^{1\times w}$, 
$D_i \in \F[\x]^{w \times w}$ for $2 \leq i \leq \ell-1$,
and $D_{\ell} \in \F[\x]^{w\times 1}$ 
such that 
\begin{eqnarray*}
D_1(j) &=& W(v_0,v_{1,j}),\; \text{ for } 1 \leq j \leq w,\\
D_i(j, k) &=& W(v_{i-1,j},v_{i,k}),\;  \text{ for } 1 \leq j,k \leq w \text{ and } 2 \leq i \leq n-1,\\
D_{\ell}(k) &=& W(v_{\ell-1,k},v_{\ell}),\;  \text{ for } 1 \leq k \leq w.
\end{eqnarray*}
Here we use the convention that $W(u,v) = 0$ if $(u,v)$ is not an edge in the ABP.


\subsection{Read-once Oblivious Arithmetic Branching Programs}
\label{subsec:roabpCharacterization}

An ABP is called a {\em read-once oblivious ABP (ROABP)}
if the edge weights in every layer are univariate polynomials in the same variable,
and every variable occurs in at most one layer.
Hence,
the length of an ROABP is~$n$, the number of variables.
The entries in the matrix~$D_{i}$ defined above come from~$\F[x_{\pi(i)}]$,
for all $i \in [n]$, where $\pi$ is a permutation on the set~$[n]$.
The order $(x_{\pi(1)}, x_{\pi(2)}, \dots, x_{\pi(n)})$ is said 
to be the \emph{variable order\/} of the ROABP.

We will view~$D_{i}$ as a polynomial in the variable~$x_{\pi(i)}$,
whose coefficients are $w$-dimensional vectors or matrices. The read-once property gives us an easy way to express the coefficients of the
polynomial~$A(\x)$ computed by an ROABP; namely for a polynomial $A(\x) = D_1(x_{\pi(1)}) D_2(x_{\pi(2)}) \cdots D_n(x_{\pi(n)})$ computed by an {\rm ROABP},
we have 
\begin{equation}
 \textrm{\cf}_{\xa}(A) \;=\; \prod_{i=1}^n \cf_{x_{\pi(i)}^{a_{\pi(i)}}}(D_i)~~ \in \F \,.
 \label{eq:coeff}
\end{equation}

We also consider matrix polynomials computed by an ROABP.
A matrix polynomial $A(\x) \in F^{w \times w}[\x]$ is said to be computed by an ROABP if
$A = D_1 D_2 \cdots D_n$, where $D_i \in F^{w \times w}[x_{\pi(i)}]$
for $i = 1, 2, \dots, n$ and some permutation~$\pi$ on~$[n]$.
Similarly, a vector polynomial $A(\x) \in F^{1 \times w}[\x]$
is said to be computed by an ROABP if 
$A = D_1 D_2 \cdots D_n$, where $D_1 \in F^{1 \times w}[x_{\pi(1)}]$ 
and $D_i \in F^{w \times w}[x_{\pi(i)}]$ for $i =  2, \dots, n$.
Usually, we will assume that an ROABP computes a polynomial in $\F[\x]$,
unless mentioned otherwise.

We state the definition of characterizing dependencies, which defines an ROABP layer by layer.

\medskip
\begin{definition}\label{def:dependencies}
Let $A(\x)$ be polynomial of individual degree~$d$, with variable-order $(x_{\pi(1)},  \cdots, x_{\pi(n)})$.
For any $0 \leq k \leq n$ and $\y_k = x_{\pi(1)},  \cdots, x_{\pi(k)}$, 
let
$$\dim_\F \{ \coeffset{A}{\y_k}{\a} \mid \a \in \{0,1,\dots, d\}^k \} \;\leq\; w ,$$
for some~$w$.

For $0 \leq k \leq n$,
we define the {\em spanning sets\/}
$\spanning_k(A)$ and the {\em dependency sets\/} $\depending_k(A)$ as subsets of $\{ 0,1, \dots, d \}^k$ as follows.

For $k = 0$, let
$\depending_0(A) = \emptyset$ and
$\spanning_0(A) = \{ (\,) \}$,
where $(\,)$ is the empty tuple. 
For $k > 0$, let
\begin{itemize}
\item 
$\depending_k(A) = \{ (\a , j) \mid \a \in \spanning_{k-1}(A) \text{ and } 0 \leq j \leq d \}$,
i.e.\
$\depending_k(A)$ contains all possible extensions of the tuples in $\spanning_{k-1}(A)$.
\item 
$\spanning_k(A) \subseteq \depending_k(A)$
is any set of size $\le w$,
such that 
for any $\vect b \in \depending_k(A)$, the polynomial~$\coeffset{A}{\y_k}{\vect b}$
is in the span of $\{\coeffset{A}{\y_k}{\a} \mid \a \in \spanning_k(A)\}$.
\end{itemize}
The linear dependencies of the polynomials in $\{\coeffset{A}{\y_k}{\a} \mid \a \in \depending_k(A)\}$
over $\{\coeffset{A}{\y_k}{\a} \mid \a \in \spanning_k(A)\}$ are the 
{\em characterizing set of dependencies}.
\end{definition}
The spanning set $\spanning_k(A)$ is not unique.

Nisan~\cite{nisan1991lower} gave an exact
width characterization for ROABPs (Nisan considers the model of noncommutative ABPs, but all statements can be translated to the ROABP setting). We follow the presentation of~\cite{Gurjar15} for this characterization.

\medskip
\begin{lemma}[{\cite{nisan1991lower}}] \label{lem:ROABPdim}
Let $A(\x)$ be polynomial of individual degree~$d$, computed by an ROABP of
width $w$, with variable-order $(x_{\pi(1)},  \cdots, x_{\pi(n)})$. For $k \in [n]$, let $\y = x_{\pi(1)},  \cdots, x_{\pi(k)}$, be the prefix of length $k$ and $\z$ be the suffix of length $n-k$. Then,
\[\dim_\F \{ \coeffset{A}{\y}{\a} \mid \a \in \{0,1,\dots, d\}^k \} \leq w\;.\]
Conversely, let $A(\x)$ be a polynomial of individual degree $d$, such that for any $k \in [n]$ and
$\y_k = (x_{\pi(1)}, \cdots, x_{\pi(k)})$, we have $\dim_\F \{ \coeffset{A}{\y}{\a} \mid \a \in \{0,1,\dots, d\}^k \} \leq w$. Then, there exists an ROABP of width $w$ for $A(\x)$ in
the variable order $ (x_{\pi(1)}, \cdots, x_{\pi(n)})$.
\end{lemma}

A polynomial $A \in \F[\x]$ is computable by an {\em any-order ROABP} ($\ARO$) of size $w$, if for all possible permutations of variables there exists an ROABP of size at most $w$ in that variable order. It is easy to check that for an $\ARO$, \cref{lem:ROABPdim} holds wrt any variable-order.

One can also capture the space by the coefficient matrix (also known as the partial derivative matrix) where the rows are indexed by monomials $p_i$ from $\y$, columns are indexed by monomials $q_j$ from $\z = \x \backslash \y$ and $(i,j)$-th entry of the matrix is $\cf_{p_i \cdot q_j}(A)$.  We refer the reader to~\cite{saptharishi2015survey} for details on this matrix and its connection with coefficient polynomials.


\subsection{Some Technical Lemmas}
\paragraph{Newton Identities.}~ Let $e_k(x_1,\hdots,x_n)$ denotes the $k$-th {\bf elementary symmetric} polynomial, defined by
\[
e_k(x_1,\hdots,x_n)\;:=\;\sum_{1 \le j_1 < j_2 < \cdots <j_{k} \le n} x_{j_1}\cdots x_{j_k}\;;
\]
Recall that by definition $e_0 = 1$. It is easy to observe that
\[
\prod_{i = 1}^n (1+x_i) = \sum_{j =0}^n e_j(\x)\;.
\]

{\em Newton identities} are a central tool in this section; they relate the elementary symmetric polynomials and the {\em power sum} polynomial, defined as $p_k(\x):=x_1^k+\cdots+x_n^k$. 

\medskip
\begin{proposition}[Newton Identities, see e.g. \cite{Macdon:SymmetricFunctions}, Section I.2] \label{prop:Newt-Id}
Let $n,k$ be integers with $n \ge k \ge 1$. Then
\[ 
k\cdot e_k(x_1,\hdots,x_n)\;=\;\sum_{i \in [k]}\,(-1)^{i-1}e_{k-i}(x_1,\hdots,x_n) \cdot p_i(x_1,\hdots,x_n)\;.
\]
\end{proposition}

{\noindent{\bf Power series and $\dlog$}.}~~One of the key benefits of the power series ring comes from the \emph{inverse} identity: $(1 - x)^{-1} = \sum_{i \ge 0}\, x^i$. This will be used widely in many proof sketches.

The logarithmic derivative operator $\dlog_{\,z}(f) \;:=\; (\partial_{z}f)/f$ is another key tool which {\em linearizes} the product gate, since 
\begin{equation}\label{eq:dlog-equation}
 \dlog_{y}(f \cdot g) \,=\, \partial_y(fg)/(fg) \,=\, (f \cdot \partial_y g \,+\, g \cdot \partial_y f)/(fg)\,=\,\dlog_{y}(f) + \dlog_y(g) \;. 
\end{equation}
This operator enables us to use power-series expansion, and converts the $\prod$-gate to $\wedge$.

\medskip
Let~$\ell \in \C[\x]$ be a linear polynomial such that the constant term is nonzero. For simplicity, suppose $\ell := 1 + \tilde{\ell}$, where $\tilde{\ell}$ is a homogeneous linear polynomial. Further, let~$\Phi: \x \mapsto z\x$. Note that $\Phi(\ell) = 1+ z \cdot \tilde{\ell}$. Therefore, by simple power series expansion as mentioned above, $\dlog_{z}(\Phi(\ell))$ becomes:
\begin{equation}\label{eq:dlog-linear-mod}
 \dlog_{z}(\Phi(\ell))  = \frac{\tilde{\ell}}{1 + z \cdot \tilde{\ell}} \,=\, \sum_{i \ge 0}\,(-1)^{i} z^{i} \cdot \tilde{\ell}^{i+1}\;.   
\end{equation}
In later proofs, we will generally work with transformations of the form $\x \mapsto z\x + \a$. While this alters the constant term (specifically the coefficient of $1$), we will nonetheless obtain a power series expansion that remains a sum of powers of linear forms.

One crucial fact that we will use throughout is the following. Let~$h \in \F[z]$, for a field $\F$ and suppose $\val_{z}(h) = 0$. Then $1/h$ is a power series in $z$, i.e.~$1/h \in \F[[z]]$. To give an explicit example,~let~$h:=z+\epsilon$; trivially 
\[\val_z(h)=0\;,\;\;\;\;\text{and}\;\;\;\; \frac{1}{z+\eps} \;=\; \sum_{i=0}^{d-1}\; (-1)^i \frac{z^i}{\epsilon^{i+1}}~\bmod~z^d\;.
\]

\medskip
\begin{proposition}[Valiant's criterion~\cite{Valiant79,burgisser2000completeness}]
\label{valiant-criterion}
Let function $\phi : \{0,1\}^{*} \rightarrow \N$ be in $\#\P/\poly$. Then, the family of polynomials defined by $ f_n(\x) := \sum_{\e \in \{0,1\}^n} \phi(\e)\cdot \x^{\e}$, is in $\VNP$.
\end{proposition}

\section{Border complexity and its algebraic characterization}
\label{sec:bordercomplexity}

Let us give a formal definition of our main object of study --- {\em border complexity}.

\medskip
\begin{definition}[Border complexity]
    Let $\Gamma$ be a complexity measure on polynomials.
    The corresponding \emph{border complexity} $\underline{\Gamma}(f)$ of a polynomial $f \in \bbC[x_1, \dots, x_n]$ is defined as the minimal~$s$ such that $f$ lies in the closure of the set $\mathcal{C}(s, n)$ of polynomials with complexity $s$, that is, if there exists a sequence of polynomials $f_k \in \bbC[\x]$ such that $\Gamma(f_k) \leq s$ and $f = \lim_{k \to \infty} f_k$.
\end{definition}

While this definition explains border complexity conceptually, it is not very convenient to work with.
For a certain class of complexity measures there is an equivalent algebraic definition.
We call these complexity measures \emph{parameterizable}.
A formal definition of a parameterizable complexity measure is stated in~\cref{def:parameterizable} below.
Intuitively, a parameterizable complexity measure $\Gamma$ is defined in terms of the minimal size of a ``device'', ``circuit'' or ``expression'' with parameters in~$\bbC$ computing polynomials in $\bbC[x_1, \dots, x_n]$, and can be extended to polynomials with coefficients in the algebra of Laurent polynomials $\bbC[\eps^{\pm 1}]$ (or any other algebra over $\bbC$) by changing the allowed parameter space.

\medskip
\begin{definition}
    Let $\Gamma$ be a parameterizable complexity measure.
    The \emph{border complexity} $\underline{\Gamma}(f)$ of a polynomial $f \in \bbC[x_1, \dots, x_n]$ is the minimal $s$ such that there exists $\tilde{f} \in \bbC[\eps][\x]$ such that $\tilde{f}|_{\eps = 0} = f$ and $\Gamma_{\bbC[\eps^{\pm 1}]}(\tilde{f}) = s$, where $\Gamma_{\bbC[\eps^{\pm 1}]}$ denotes the complexity measured over~$\bbC[\eps^{\pm 1}]$.
\end{definition}

Most of this section is devoted to the proof of this algebraic characterization of border complexity (see~\cref{thm:equivalent-definitions} below).
It was obtained by Alder~\cite{Alder-thesis} (see also~\cite[§20.6]{bcs97}) for tensor rank and circuit complexity, but applies more generally to all parameterizable complexity measures.
One direction of the proof is simple: given~$\tilde{f} \in \bbC[\eps][x_1, \dots, x_n]$ computed with complexity $s$ over $\bbC[\eps^{\pm 1}]$, we can form a sequence $f_n = \tilde{f}|_{\eps = 1/n}$ of polynomials of complexity $s$ converging to $f$.
The other direction is much more complicated and involves ideas from algebraic geometry.
These ideas are fundamental in deformation theory, where they relate geometric and formal viewpoints on deformations~\cite{fialowski1990comparison}.
The basic idea can be traced back to Hilbert, and was first applied in the context of algebraic complexity theory by Alder.
Our presentation mostly follows the proof presented in~\cite{lehmkuhl1989order} for the tensor rank complexity measure, which includes bounds on the degree and order of the Laurent polynomials involved.

\subsection{Algebro-geometric Preliminaries}

We first review some facts from algebraic geometry and prove several statements about the closure of the image of a polynomial map.
In this section, the terms ``open'' and ``closed'' refer to Zariski topology unless stated otherwise.

\paragraph{Affine and projective varieties.}~Affine and projective varieties are the spaces of solutions for a system of polynomial equation, studied in classical algebraic geometry.
We list some basic definitions in a very concrete setting where we only work with varieties embedded in an affine space $\bbC^n$ or a projective space $\bbP^n$.
We refer to algebraic geometry textbooks (e.\,g.~\cite{shafarevich1994basic}) for more information.

An \emph{affine variety} in $\bbC^n$ is the set of all common zeros of a finite set of polynomials $F_1, \dots, F_k \in \bbC[x_1, \dots, x_n]$.
The \emph{ideal} $I_\mathcal{X}$ of an affine variety $\mathcal{X} \subset \bbC^n$ consists of all polynomials $F \in \bbC[x_1, \dots, x_n]$ vanishing on $\mathcal{X}$.
The \emph{coordinate ring} $O_\mathcal{X}$ of $\mathcal{X}$ is defined as $\bbC[x_1, \dots, x_n]/I_\mathcal{X}$.
Elements of the coordinate ring can be identified with \emph{regular functions} on $\mathcal{X}$, that is, functions on $\mathcal{X}$ given by restrictions of polynomials on $\bbC^n$.
A \emph{morphism} between varieties $\mathcal{X} \subset \bbC^m$ and $\mathcal{Y} \subset \bbC^n$ is a map $\varphi \colon \mathcal{X} \to \mathcal{Y}$ given by restriction of a polynomial map from $\bbC^m$ to $\bbC^n$.
A \emph{closed subvariety} of an affine variety $\mathcal{X} \subset \bbC^n$ is a subset which is itself an affine variety.
The space $\bbC^n$ considered as an affine variety is called the \emph{affine space}.
The \emph{Zariski topology} on an affine variety is the topology in which closed sets are exactly closed subvarieties.
Note that in the Zariski topology every nonempty open set is dense.
A \emph{rational function} on an affine variety $\mathcal{X}$ is a quotient of two regular functions. It is defined on the open subset of $\mathcal{X}$ defined by the nonvanishing of the denominator.

Consider the equivalence relations on $\bbC^{n + 1} \setminus \{0\}$ in which two tuples are equivalent if and only if they are proportional to each other.
The set of all equivalence classes is called the \emph{$n$-dimensional projective space} $\bbP^n$.
The point of $\bbP^n$ corresponding to the tuple $x = (x_0, x_1, \dots, x_n)$ is denoted by $[x] = (x_0 : x_1 : \dots : x_n)$, and the elements $x_i$ are called \emph{homogeneous coordinates of $[x]$}.
If $F \in \bbC[x_0, \dots, x_n]$ is a homogeneous polynomial which vanishes on one tuple in the equivalence class $[x]$, then it vanishes on the whole equivalence class.
A \emph{projective variety} in $\bbP^n$ is the set of all points on which a finite set of homogeneous polynomials vanishes.
The Zariski topology on a projective variety is defined in the same way as on an affine one.
A \emph{rational map} between projective varieties $\mathcal{X} \subset \bbP^m$ and $\mathcal{Y} \subset \bbP^n$ is a map $\varphi \colon \mathcal{U} \to \mathcal{Y}$ defined on an open set $\mathcal{U} \subset \mathcal{X}$ such that every point of $\mathcal{U}$ has a neighborhood on which 
\[
\varphi\, :[x] \;\mapsto\; (f_0(x) : f_1(x) : \dots : f_n(x))\;,
\]
where $f_i$ are homogeneous polynomials of the same degree.
If a rational map is defined on the whole $\mathcal{X}$, it is a \emph{morphism} between varieties.
Morphisms of projective varieties have a very important property.

\medskip
\begin{theorem}[{\cite[Theorem 1.10]{shafarevich1994basic}}]
    Morphisms of projective varieties are closed, that is, they map closed sets to closed sets.
\end{theorem}

An affine or projective variety defined by one nonconstant polynomial (homogeneous in the projective case) is called a \emph{hypersurface}.
A hypersurface defined by an affine linear form (linear in the projective case) is called a \emph{hyperplane}.

\medskip
\begin{lemma}\label{lem:Zariski-open-Euclidean-dense}
    A nonempty Zariski open subset of $\bbC^n$ is dense in Euclidean topology.
\end{lemma}
\begin{proof}
    It is enough to prove that every Zariski closed subset does not contain an open ball around any of its points.
    To see this, note that a Zariski closed subset $\mathcal{X}$ is contained in a hypersurface defined by some polynomial $F \in I_\mathcal{X}$, and the hypersurface does not contain a ball around any of its point because if a polynomial vanishes on an open ball, then it is a zero polynomial (because polynomials are analytic functions).
\end{proof}

The affine space $\bbC^n$ can be embedded into $\bbP^n$ as an open subset $U_0$ consisting of points on which the homogeneous coordinate $x_0$ is nonzero; these points have the form $(1 : x_1 : \dots : x_n)$ and are identified with $(x_1, \dots, x_n) \in \bbC^n$).
An intersection of a projective variety with $U_0$ is an affine variety.
Each open subset $U_i = \{ [x] \in \bbP^n \colon x_i \neq 0 \}$ forms an affine space.
These subsets are called \emph{standard affine patches} of $\bbP^n$. The standard affine patches form a covering of $\bbP^n$, that is, every point of $\bbP^n$ lies in some standard affine patch.

A variety is \emph{irreducible} if it cannot be presented as a union of nontrivial closed subvarieties.
Every variety is a finite union of irreducible closed subvarieties called its \emph{irreducible components}.

Dimension is a fundamental property of a variety.
It has many definitions which come from different points of view on varieties, one of which is the following.

\medskip
\begin{definition}[Dimension of a variety]
    The \emph{dimension} $\dim \mathcal{X}$ of an irreducible variety $\mathcal{X}$ is the maximal length $m$ of a decreasing chain $\mathcal{X} = \mathcal{X}_0 \supset \mathcal{X}_1 \supset \dots \supset \mathcal{X}_n$ of nonempty irreducible closed varieties.
    The dimension of an arbitrary variety is the maximum among the dimensions of its irreducible components.
    We say that a variety $\mathcal{X}$ is \emph{of pure dimension $n$} if all its irreducible components have dimension $n$.
    The \emph{codimension} of an irreducible variety $\mathcal{X} \subset \bbC^m$ is $m - \dim X$.
\end{definition}

A variety of dimension $0$ is a finite set of points. An irreducible variety of dimension $1$ is called a \emph{curve}.

\paragraph{Smooth and singular points.}~In algebraic geometry, the properties of a variety that are local at one of its points are studied through the ideal corresponding to this point.
In particular, the properties of this ideal determine if the variety is smooth at the point, and it can also be used to construct the complete local ring --- an algebraic object containing a fine description of the shape of the singularity if the point is singular.

Consider a point $x$ on an affine variety $\mathcal{X}$.
The evaluation map $\mathrm{ev}_x \colon \bbC[\mathcal{X}] \to \bbC$ sends each regular function $f$ on $\mathcal{X}$ to its value $f(x)$.
This map is a ring homomorphism, and its kernel is the ideal of the point $x$ considered as a subvariety, that is, the ideal $\mathfrak{m}_x$ consisting of functions vanishing on $x$.  
Since $\bbC[\mathcal{X}] / \mathfrak{m}_x \cong \bbC$ is a field, the ideal $\mathfrak{m}_x$ is maximal.

\medskip
\begin{definition}[Tangent and cotangent space]
The vector space $\mathfrak{m}_x/\mathfrak{m}_x^2$ is called the \emph{cotangent space} $T^*_x\mathcal{X}$ of $\mathcal{X}$ at $x$, and its dual -- the \emph{tangent space} $T_x\mathcal{X}$.
\end{definition}

\medskip
\begin{definition}[Smooth variety]
A point $x \in \mathcal{X}$ is called \emph{smooth} if $\dim T^*_x\mathcal{X} = \dim \mathcal{X}$, and \emph{singular} otherwise.
    A variety $\mathcal{X}$ is called \emph{smooth} if it has no singular points.
\end{definition}

The behaviour of a rational function $F$ on $\mathcal{X}$ at $x$ is described by its images in the quotients $\bbC[\mathcal{X}]/\mathfrak{m}_x^p$ for different $p$.
If $x$ is a smooth point of $\mathcal{X}$, these quotients contain the same information as the successive Taylor approximations of $F$ at $x$.
Collecting all these approximations we obtain an object that plays the role of the Taylor series --- the {\em germ} of $F$ at the point $x$.
The ring where germs of regular functions naturally live is an analog of the ring of formal Taylor series.
It is a special case of a categorical construction called a \emph{projective limit}, but we will not use this terminology.

\medskip
\begin{definition}[Complete local ring]
    Consider a point $x$ on an affine variety $\mathcal{X}$.
    The \emph{complete local ring} $\hat{O}_{x, \mathcal{X}}$ is defined as the set of sequences $(F_0, F_1, \dots)$ with $F_p \in \bbC[\mathcal{X}]/\mathfrak{m}_x^{p+1}$ which are \emph{compatible} in the sense that for $p \leq q$ we have that $F_p$ coincides with the image of $F_q$ under the natural projection $\bbC[\mathcal{X}]/\mathfrak{m}_x^{q+1} \to \bbC[\mathcal{X}]/\mathfrak{m}_x^{p+1}$.
    The \emph{germ of a regular function $F \in \bbC[\mathcal{X}]$ at point $x$} is an element of $(F_0, F_1, \dots, ) \in \hat{O}_{x, \mathcal{X}}$ where $F_p$ is the coset of $F$ in $\bbC[\mathcal{X}]/\mathfrak{m}_x^{p+1}$.
    The \emph{value~$\tilde{F}(x)$ of an element $\tilde{F} \in \hat{O}_{x, \mathcal{X}}$ at $x$} is the image of $\tilde{F}_0$ under the identification $\bbC[\mathcal{X}]/\mathfrak{m}_x \cong \bbC$.
\end{definition}

The complete local ring $\hat{O}_{x, \mathcal{X}}$ contains information about the shape of the variety $\mathcal{X}$ near the point $x$.
In particular, a smooth point on an $n$-dimensional variety has the complete local ring isomorphic to $\bbC[[x_1, \dots, x_n]]$.
We present a proof for the case of curves.

\medskip
\begin{lemma}
    If $x$ is a smooth point on an affine curve $\mathcal{E}$, then its complete local ring $\hat{O}_{x, \mathcal{E}}$ is isomorphic to $\bbC[[\eps]]$.
\end{lemma}
\begin{proof}
    We have $\dim \mathfrak{m}_x/\mathfrak{m}_x^2 = 1$.
    Let $e$ be the element of $\bbC[\mathcal{E}]$ corresponding to a nonzero element of this vector space.
    For every element $f$ of $\mathfrak{m}_x$ it holds that $f \equiv \alpha e \pmod{\mathfrak{m}_x^2}$, and by taking a product of $k$ such elements we obtain that every element of $\mathfrak{m}_x^k$ is a multiple of $e^k$ modulo $\mathfrak{m}_x^{k + 1}$.
    We then compute $\bbC[\mathcal{E}]/\mathfrak{m}_x^{k + 1} = \bbC[e]/\left<e^{k + 1}\right>$ and $\hat{O}_{x, \mathcal{E}} = \bbC[[\eps]]$ with $\eps$ being the germ of $e$.
\end{proof}

All the notions defined in this paragraph for affine varieties can also be defined for projective varieties by looking at an affine patch containing the point of interest.

\paragraph{Degree of projective varieties.}~Degree is a fundamental invariant of a projective variety which governs the size of its intersection with other varieties.
We first need a basic fact about dimensions is that the intersection of a variety with a hypersurface typically cuts the dimension by $1$.

\medskip
\begin{lemma}[{\cite[Theorem~1.23]{shafarevich1994basic}}]
    If $\mathcal{X}$ is an irreducible variety of dimension $n$ and $\mathcal{H}$ is a hypersurface not containing $\mathcal{X}$, then all irreducible components of $\mathcal{X} \cap \mathcal{H}$ have dimension $n - 1$.
\end{lemma}

A hypersurface defined by a homogeneous polynomial $F$ does not contain $\mathcal{X}$ if and only if $F$ does not vanish on $\mathcal{X}$. Most hypersurfaces {\em do not} contain $\mathcal{X}$, in the following sense.

\medskip
\begin{lemma}\label{lem:generic-hypersurface}
    The set of homogeneous polynomials $F \in \bbC[x_0, \dots, x_n]_d$ such $F$ does not vanish on any irreducible component of a nonempty affine variety $\mathcal{X} \subset \bbP^n$ is open.
\end{lemma}
\begin{proof}
    For a fixed point $x$ the condition $F(x) = 0$ is a linear equation on the coefficients of the polynomial $F$.
    By considering all points of an irreducible variety $\mathcal{Y}$ we obtain that the set of polynomials vanishing on $\mathcal{Y}$ is a linear subspace in $\bbC[x_0, \dots, x_n]_d$, and the set of polynomials vanishing on an arbitary variety $\mathcal{X}$ is a union of finitely many linear subspaces corresponding to irreducible components of $\mathcal{X}$, and its complement is an open subset.
\end{proof}

We say that a general hypersurface of degree $d$ does not contain any irreducible component of $\mathcal{X}$.
Repeated application of the previous facts gives us the following corollary.

\medskip
\begin{corollary}
    Let $\mathcal{X} \subset \bbP^m$ be a projective variety of dimension $n$.
    A general linear subspace of codimension $n$ intersects $\mathcal{X}$ in a finite set of points, that is,
    there exists an open subset of tuples $(L_1, \dots, L_n)$ of linear forms such that the intersection $\mathcal{X} \cap \mathcal{L}$ with a projective linear subspace $\mathcal{L} = \{ [x] \in  \colon L_1(x) = L_2(x) = \dots = L_n(x) = 0 \}$ is finite.
\end{corollary}

This fact allows us to define the degree of a variety.

\medskip
\begin{definition}[Degree of a variety]
    If $\mathcal{X}$ is a projective variety of pure dimension $n$, then its degree $\deg \mathcal{X}$ is defined as the maximal number of points in the intersection of $\mathcal{X}$ with a codimension $n$ projective linear subspace.
\end{definition}

For example, a hypersurface $\mathcal{H}$ defined by a homogeneous polynomial $F$ of degree $d$ has degree $d$, because number of intersection points of $\mathcal{H}$ with a projective line is the number of zeros of a degree $d$ polynomial on this line obtained by restriction of $F$.
It is not hard to see that the general number of zeros is $d$.
We can also relate the degrees of a variety and its intersection with a hyperplane.

\medskip
\begin{lemma}\label{lem:degree-hyperplane}
    If $\mathcal{X}$ is a projective variety of pure dimension and $\mathcal{H}$ is a hypersurface not containing any irreducible component of $\mathcal{X}$, then $\deg (\mathcal{X} \cap \mathcal{H}) \leq \deg \mathcal{X}$.
\end{lemma}
\begin{proof}An intersection of $\mathcal{X} \cap \mathcal{H}$ with a codimension $n - 1$ linear subspace is an intersection of $\mathcal{X}$ with codimension $n$ linear subspace, and the inequality follows from the definition of degree.
\end{proof}

The degree is a much more intricate invariant than the dimension.
In a general case, the number of points in the intersection is maximal, but even in the case when it is not, there is a way to count points with multiplicity so that the total number is correct.
This is the statement of Bezout's theorem, which we will need only for smooth curves.

\medskip
\begin{theorem}[Bézout's theorem for smooth curves, {\cite[Ch.3, §2.2]{shafarevich1994basic}}]
    If $\mathcal{E}$ is a smooth projective curve and $\mathcal{H}$ is a hypersurface defined by a squarefree polynomial $F$ not containing $\mathcal{E}$, then 
    \[
    \sum_{p \in \mathcal{E} \cap \mathcal{H}} \mathrm{mult}_p(\mathcal{H},\mathcal{E}) \;=\; \deg \mathcal{E} \cdot \deg \mathcal{H}\;,
    \]
    where $\mathrm{mult}_p(\mathcal{H}, \mathcal{E}) = \dim \hat{O}_{p, \mathcal{E}}/\left<\tilde{F}\right>$, and $\tilde{F}$ is the germ of $F$ in $\hat{O}_{p, \mathcal{E}}$.
\end{theorem}

The Bézout's theorem can be generalized to the intersection of pure-dimensional varieties, and for smooth varieties the multiplicity can be defined similarly, but in general, the definition of multiplicity is much more complicated.
One can also state this theorem for the intersections with dimension more than $0$, which is the start of a large area called the {\em intersection theory}.

We can also define the degree of a pure-dimensional affine variety $\mathcal{X} \subset \bbC^n$ in the same way as for projective varieties.
This is in fact equal to the degree of the closure of $\mathcal{X}$ in $\bbP^n$.
The lack of points at infinity makes many statements less precise.
However, the following statement is easier to state for affine varieties.

\medskip
\begin{lemma}[{\cite[Theorem 8.32]{bcs97}}]
    If $\mathcal{X} \subset \bbC^{m + n}$ and $\pi \colon \bbC^{m + n} \to \bbC^m$ is the projection onto the first $m$ coordinates.
    Then $\deg \overline{\pi(\mathcal{X})} \,\leq\, \deg \mathcal{X}$.
\end{lemma}

\paragraph{Constructible sets.}~Constructible sets in an affine or projective variety can be defined as follows.

\medskip
\begin{definition}[Constructible sets]
    Let $\mathcal{X}$ be an affine or projective variety.
    A subset of $\mathcal{X}$ is \emph{locally closed} if it is an open subset of a closed subvariety of $\mathcal{X}$, that is, an intersection of a closed subset with an open subset in $\mathcal{X}$.
    A subset of $\mathcal{X}$ is \emph{constructible} if it is a union of finitely many locally closed subsets.
\end{definition}

Alternatively, constructible sets are elements of the Boolean algebra generated by all open (or closed) sets.
This means that membership in a constructible set is defined by a logical formula involving polynomial equalities and inequalities.
We record some simple consequences of the definition.

\medskip
\begin{lemma}\label{lem:constructible-contains-open}
    If $\mathcal{C}$ is a constructible set, and $\mathcal{D}$ is an irreducible component of $\overline{\mathcal{C}}$, then $\mathcal{C}$ contains a Zariski open subset of $\mathcal{D}$.
\end{lemma}
\begin{proof}
    Let $\mathcal{C} = \bigcup_{i = 1}^n \mathcal{X}_i \cap \mathcal{U}_i$ where $\mathcal{X}_i$ are closed and $\mathcal{U}_i$ are open.
    Without loss of generality $\mathcal{X}_i$ are irreducible (otherwise replace $\mathcal{X}_i \cap \mathcal{U}_i$ by a union of irreducible components of $\mathcal{X}_i$ intersected with the same $\mathcal{U}_i$).
    From general topology, $\overline{\mathcal{C}} = \bigcup_{i = 1}^n \overline{\mathcal{X}_i \cap \mathcal{U}_i}$.
    For each $i$, we either have $\mathcal{X}_i \cap \mathcal{U}_i = \varnothing$, or $\mathcal{X}_i \cap \mathcal{U}_i$ is a nontrivial Zariski open of $\mathcal{X}_i$, and thus $\overline{\mathcal{X}_i \cap \mathcal{U}_i} = \mathcal{X}_i$.
    Therefore, every irreducible component $\mathcal{D}$ of $\overline{\mathcal{C}}$ is equal to one of $\mathcal{X}_i$, and $\mathcal{X}_i \cap \mathcal{U}_i \subset \mathcal{C}$ is the required open subset of $\mathcal{D}$.
\end{proof}

\medskip
\begin{corollary}
    If $\mathcal{C}$ is a constructible set, then its Euclidean closure coincides with the Zariski closure.
\end{corollary}
\begin{proof}
Follows from the previous Lemma and~\cref{lem:Zariski-open-Euclidean-dense}.
\end{proof}

Constructible sets are important for us because of the Chevalley's theorem, which implies that the image of a polynomial map is a constructible set.
The Chevalley's theorem holds in high generality for finitely presented morphisms of schemes~\cite[054K]{stacks-project}. We only state it for the situation we need.

\medskip
\begin{theorem}[Chevalley's theorem on constructible sets] \label{thm:chevalley}
    If $\varphi \colon \bbC^m \to \bbC^n$ is a polynomial map, then for every constructible set $\mathcal{C} \subset \bbC^m$ its image $\varphi(\mathcal{C}) \subset \bbC^n$ is constructible.
\end{theorem}

In this generality, the Chevalley's theorem is equivalent to Tarski-Seidenberg theorem from logic.

\medskip
\begin{theorem}[Tarski--Seidenberg theorem~\cite{seidenberg1954new}]
    First-order theory of an algebraically closed field admits quantifier elimination.
\end{theorem}

\paragraph{Approximating curves.}~A natural way to study higher-dimensional varieties is by intersecting them with hyperplanes; the following lemma ensures that such intersections can still retain geometric relevance to any given point in the variety. 

\medskip
\begin{lemma}
    Let $\mathcal{X}$ be an irreducible affine or projective variety of dimension at least $2$.
    If $\mathcal{U}$ is a nonempty open in $\mathcal{X}$ and $x \in \mathcal{X}$,
    then there exists a hyperplane $\mathcal{H}$ such that
    $x \in \mathcal{H}$,
    the intersection $\mathcal{Y} = \mathcal{X} \cap \mathcal{H}$ is proper,
    and the irreducible component of $\mathcal{Y}$ containing $x$ intersects with $\mathcal{U}$.
\end{lemma}
\begin{proof}
    Let $\mathcal{Z} = \mathcal{X} \setminus \mathcal{U}$.
    Note that $\mathcal{Z}$ is a closed subvariety of $\mathcal{X}$ containing $x$, so $\dim \mathcal{Z} \leq \dim \mathcal{X} - 1$.

    Similarly to~\cref{lem:generic-hypersurface}, one can prove that a general hyperplane $\mathcal{H}$ containing $x$ does not contain any of the irreducible components of $\mathcal{Z}$ except $\{x\}$ if it is an irreducible component.
    It follows that all irreducible components of the intersection $\mathcal{H} \cap \mathcal{Z}$ have dimension at most $\dim \mathcal{Z} - 1$ or $0$.
    In both cases the dimension does not exceed $\dim \mathcal{X} - 2$.
    
    Let $\mathcal{Y}$ be the irreducible component of $\mathcal{H} \cap \mathcal{X}$ containing the point $x$.
    Since $\dim \mathcal{Y} = \dim \mathcal{X} - 1$, it is not contained in $\mathcal{H} \cap \mathcal{Z}$, which has dimension at most $\dim \mathcal{X} - 2$.
    Therefore, the intersection of $\mathcal{Y}$ with $\mathcal{U}$ is nonempty.
\end{proof}

\begin{theorem}\label{thm:curve-existence}
    Let $\mathcal{X}$ be an irreducible affine or projective variety.
    If $\mathcal{U}$ is a nonempty open in $\mathcal{X}$ and $x \in \mathcal{X}$,
    then there exists a curve $\mathcal{E} \subset \mathcal{X}$ such that $x \in \mathcal{E}$ and $\mathcal{E}$ intersects with $\mathcal{U}$.
    Moreover, $\deg \mathcal{E} \leq \deg \mathcal{X}$.
\end{theorem}
\begin{proof}
    The proof is by induction on the dimension of $\mathcal{X}$.
    If $\dim \mathcal{X} = 1$, then $\mathcal{X}$ is the required curve.
    Otherwise, we choose a hyperplane $\mathcal{H}$ using the previous Lemma.
    Let $\mathcal{X}' \subset \mathcal{X}$ be the irreducible component of $\mathcal{X} \cap \mathcal{H}$ containing $x$.
    Since the intersection is proper, $\dim \mathcal{X}' = \dim \mathcal{X} - 1$.
    Since $\mathcal{X}'$ intersects $\mathcal{U}$, the intersection $\mathcal{U}' = \mathcal{X} \cap \mathcal{U}'$ is a nonempty open in $\mathcal{X}'$.
    Moreover, $\deg \mathcal{X}' \leq \deg \mathcal{X}$ by~\cref{lem:degree-hyperplane}.
    By induction hypothesis, we can find a curve in $\mathcal{X}'$ that contains $x$ and intersects with $\mathcal{U}'$.
    This is the required curve.
\end{proof}

\paragraph{Resolution of singularities on curves.}~Resolution of singularities -- the fact that every variety with singular points can be obtained as an image of a smooth variety --- is an important result in algebraic geometry in characteristic $0$.
The general result is very complicated, but the case of curves is classical and well understood.
We will state the resolution for curves in the following form.

\medskip
\begin{theorem}
    For every projective curve $\mathcal{E} \subset \bbP^n$ there exists a smooth projective curve~$\mathcal{D} \subset \bbP^{m}$ and a morphism $\sigma \colon \mathcal{D} \to \mathcal{E}$ such that $\sigma(\mathcal{D}) = \mathcal{E}$ and every smooth point of $\mathcal{E}$ has only one preimage under $\sigma$.
\end{theorem}

We may assume that the resolution $\mathcal{D}$ is contained in $\bbP^n \times \bbP^m$ and $\sigma$ is the projection onto the first component by changing $\mathcal{D}$ to $\{(\sigma(x), x) \mid x \in \mathcal{D}\}$.

The proof of this theorem most often presented in the textbooks goes through a formal argument involving normalization of a curve, which recovers the required smooth curve via the integral closure of the coordinate ring of $\mathcal{E}$.
There are other proofs, one of which recovers the smooth resolution via successive projections of a curve from a singular point.
To the interested reader, we recommend the book of János Kollàr~\cite{kollar2009lectures}, which contains many different proofs of the resolution of singularities for curves.

We need this result to describe the shape of a curve near one of its points using formal series.

\medskip
\begin{theorem}\label{thm:formal-series-existence}
    Consider a point $\bar{x}$ on an affine curve $\mathcal{E} \subset \bbC^n$.
    There exists a tuple of formal series $\tilde{x} = (\tilde{x}_1, \dots, \tilde{x}_n) \in \bbC[[\eps]]^n$ such that $\mathrm{Coeff}_{\eps^0}(\tilde{x}_i) = \bar{x}_i$, for every $F \in I_\mathcal{E}$ we have $F(\tilde{x}) = 0$, and $\tilde{x}_i \neq \bar{x}_i$ unless $\mathcal{E}$ is contained in the hyperplane given by $x_i = \bar{x}_i$.
    Moreover, the valuation $\val_\eps(\tilde{x}_i - \bar{x}_i) \leq \deg \mathcal{E}$.
\end{theorem}
\begin{proof}
    Let $\overline{\mathcal{E}}$ be the closure of $\mathcal{E}$ in the projective space $\bbP^n$ and consider a resolution of singularities for $\mathcal{E}$, given by the projection of a smooth curve $\mathcal{D} \colon \bbP^n \times \bbP^m$ onto the first factor.

    Let $y \in \mathcal{D}$ be one of the preimages of $\mathcal{D}$.
    Restrict to the affine patch $\bbC^n \times \bbC^m$ containing $\bar{y}$.
    Consider the coordinate functions $\xi_1, \dots, \xi_n$ defined as $\xi_i(x) = x_i$.
    It is clear that $\xi_i$ satisfy the following conditions: we have $\xi_i(y) = y_i = \bar{x}_i$, $\xi_i$ is nonconstant on $\mathcal{D}$ unless $\mathcal{E}$ and, therefore, $\mathcal{D}$, is contained in the hyperplane with constant $i$-th coordinate, and if $F \in I_\mathcal{E}$, then $F$ also vanishes on $\mathcal{D}$, so $F(\xi_1, \dots, \xi_n) = 0$ in $\bbC[\mathcal{E}]$.

    Let $\tilde{x}_i$ be the germ of $\xi_i$ in $\hat{O}_{y, \mathcal{D}} \cong \bbC[[\eps]]$. The previously listed properties of $\xi_i$ imply that $\tilde{x}_i$ satisfy the conditions of the theorem.

    To get the degree bound, note that $\val_\eps(\tilde{x}_i)$ is the multiplicity of $y$ in the intersection of $\mathcal{D}$ with the hyperplane $x_i = 0$.
\end{proof}

By restricting to an affine patch containing the point of interest, we obtain the analog of the previous statement for the projective curves.

\paragraph{Closure of the image of a polynomial map.}~Now we state and prove an equivalent statement about the closure of the image of a polynomial map.

\medskip
\begin{theorem}[\cite{lehmkuhl1989order}]\label{thm:image-closure}
    Let $\varphi \colon \bbC^m \to \bbC^n$ be a polynomial map.
    Let $\mathcal{G} \in \bbC^m \times \bbC^n$ be the graph of $\varphi$, that is, the affine variety $\{(t, x) \mid \varphi(t) = x\}$, and $\deg \mathcal{G} = D$.
    The following statements are equivalent:
    \begin{enumerate}
        \item $x \in \overline{\image(\varphi)}$;
        \item there exists an algebraic curve $\mathcal{D} \subset \bbC^m$ such that $x \in \overline{\varphi(\mathcal{D})}$; moreover, $\deg \mathcal{D} \leq D$;
        \item there exists a tuple of formal series $\tilde{u} \in \bbC((\eps))^m$ such that $\varphi(\tilde{u}) = x + \eps y$ for some $y \in \bbC[[\eps]]^n$; moreover, $\val_\eps(\tilde{u}) \geq -D$
        \item there exists a tuple of Laurent polynomials $\tilde{v} \in \bbC[\eps^{\pm 1}]^m$ such that $\varphi(\tilde{v}) = x + \eps y$ for some $y \in \bbC[\eps]^n$; moreover, $\val_\eps(\tilde{v}) \geq - D$.
    \end{enumerate}
\end{theorem}
\begin{proof}
    Let $\pi_1 \colon \bbC^m \times \bbC^n \to \bbC^m$ and $\pi_2 \colon \bbC^m \times \bbC^n \to \bbC^n$ be projections onto the first and second factors respectively.
    Note that $\pi_1$ gives an isomorphism between $\mathcal{G}$ and $\bbC^m$, with the inverse given by $\pi_1^{-1}(t) = (t, \varphi(t))$.
    Moreover, $\image(\varphi)$ is exactly $\pi_2(\mathcal{G})$.
    Embed $\bbC^m \times \bbC^n$ into the product of projective spaces $\bbP^m \times \bbP^n$, and extend the projections $\pi_1$ and $\pi_2$ accordingly,
    and further consider the closure $\overline{\mathcal{G}} \subset \bbP^m \times \bbP^n$.
    Since $\overline{\mathcal{G}}$ is a projective variety, $\pi_2(\overline{\mathcal{G}}) \subset \bbP^n$ is closed,
    and since it contains $\image(\varphi)$, it also contains its closure.

\medskip
    {\bf (1) $\Rightarrow$ (2):}
    First, apply~\cref{thm:curve-existence} to $\overline{\image(\varphi)}$ with the point $x \in \overline{\image(\varphi)}$ and an open subset contained in $\image(\varphi)$, which exists by~\cref{lem:constructible-contains-open} to obtain the curve $\mathcal{E} \subset \overline{\image(\varphi)}$.
    Let $\mathcal{Y} = \pi^{-2}(\mathcal{F})$. It is a closed subvariety of $\overline{\mathcal{G}}$.
    The image of every irreducible component of $\mathcal{Y}$ under $\pi_2$ is an irreducible closed subvariety of $\mathcal{E}$, that is, either $\mathcal{E}$ itself or a point.
    
    Let $\mathcal{Y}'$ be one of the irreducible components such that $\pi_2(\mathcal{Y}') = \mathcal{E}$.
    Let $\mathcal{T} = \pi_1(\mathcal{Y}')$.
    Since $\mathcal{E}$ contains points in $\image(\varphi)$, $\mathcal{T}$ intersects with the affine chart $\bbC^m$.
    Let $s$ be a point in $\mathcal{T} \cap \bbC^m$.
    Apply~\cref{thm:curve-existence} again to $\mathcal{T}$ with the point $s$ and the open set $(\mathcal{T} \cap \{t \in \bbC^m \colon \varphi(t) \neq \varphi(s) \}$ to get a curve $\overline{\mathcal{D}}$, and let $\mathcal{D} = \overline{\mathcal{D}} \cap \bbC^m$.
    By construction of $\mathcal{D}$, the closure $\overline{\varphi(\mathcal{D})}$ contains at least two points.
    Since it is an irreducible closed subset of $\mathcal{E}$, it coincides with the whole $\mathcal{E}$, hence it contains $x$.
    
    To prove the degree bound, let $\mathcal{F} = \pi_1^{-1}(\mathcal{D})$.
    From the construction, it follows that $\mathcal{F}$ is an intersection of $\mathcal{G}$ with a linear subspace, so $\deg \mathcal{F} \leq D$.
    Moreover, $\mathcal{D} = \pi_1(\mathcal{F})$, so $\deg \mathcal{D} \leq \deg \mathcal{F} \leq D$.

\medskip
    {\bf (2) $\Rightarrow$ (3):}
    Let $\mathcal{D}$ be the curve in (2).
    Let $\mathcal{F} = \pi_1^{-1}(\mathcal{D})$ and consider the closure $\overline{\mathcal{F}} \subset \overline{\mathcal{G}}$.
    The image $\pi_2(\overline{\mathcal{F}})$ is a closed subset of $\overline{\mathcal{G}}$ containing $\image(\varphi)$, so it also contains $x$.
    Consider the point $(t, x) \in \overline{\mathcal{F}} \subset \bbP^m \times \bbP^n$.
    Let $t = (t_0 : \dots : t_m)$, $x = (1 : x_1 : \dots : x_n)$ and choose $k$ such that $t_k \neq 0$.
    Restrict to the affine patch $U_k \times \bbC^n$ and
    apply~\cref{thm:formal-series-existence} to obtain formal series $\tilde{t}_i$, $\tilde{x}_i$ such that the following holds: 
    \[\mathrm{Coeff}_{\eps^0} (\tilde{t}_i) \;=\; t_i\,,\; \mathrm{Coeff}_{\eps^0} (\tilde{x}_i) \;=\; x_i\,,\;\text{and}\;(\tilde{t}, \tilde{x})\;\text{satisfies all equations of}\;\overline{\mathcal{G}}\;.
    \]
    In particular, the projectivized version of the equation $\varphi(u) = x$ holds for $\tilde{u}, \tilde{x}$, which means that $\varphi(\frac{\tilde{t}_1}{\tilde{t}_0}, \dots, \frac{\tilde{t}_n}{\tilde{t}_0}) = \tilde{x}$ where we take $\tilde{t}_k = 1$, so $\tilde{u}_i = \frac{\tilde{t}_i}{\tilde{t}_0} \in \bbC((\eps))$ form the required tuple. 
    Since $\deg \mathcal{F} \leq \deg \mathcal{G}$,
    By~\cref{thm:formal-series-existence} we have $\val_\eps(\tilde{t}_0) \leq D$ and therefore $\val(\tilde{u}_i) \geq -D$.

\medskip
    {\bf (3) $\Rightarrow$ (4):}
    Let $\tilde{u}_i = \sum_{k = -d}^\infty u_{ik} \eps^k$, with $-d$ being the most negative among the powers of $\eps$ appearing in $a$ (or $0$ if there are no negative powers).
    Suppose the polynomials $\varphi_j$ defining coordinates of the polynomial map $\varphi$ have degree at most $q$.
    Note that the terms of $\tilde{u}_i$ with powers of $\eps$ higher than $dq$ contribute only to the positive power of $\eps$
    in every monomial in $\tilde{u}$ of degree $p \leq q$, since the monomial can be expressed as $$\tilde{u}_{i_1} \dots \tilde{u}_{i_p} \;=\; \sum_{k = -dp}^\infty \sum_{k_1 + \dots + k_p = k} u_{i_1 k_1} \dots u_{i_p k_p} \eps^k,$$
    and $k_1 + \dots + k_p \geq dq - (p - 1) \cdot d$, when one of $k_i$ is greater than $dq$.
    If we define the Laurent polynomials $\tilde{v}_i$ by truncating the series $\tilde{u}_i$ at degree $dq$, that is, $\tilde{v}_i = \sum_{k = -d}^{qd} \tilde{u}_{ik} \eps^k$,
    the expressions $\varphi(\tilde{v})$ and $\varphi(\tilde{u})$ only differ in the positive powers of $\eps$, hence the required condition holds.

\medskip
    {\bf (2) $\Rightarrow$ (1):}
    Since $\varphi(\mathcal{D}) \subset \image(\varphi)$, we have $x \in \overline{\varphi(\mathcal{D})} \subset \overline{\image(\varphi)}$.

\medskip
    {\bf (4) $\Rightarrow$ (1):}
    Note that for every $\eps \neq 0$ the point $\varphi(\tilde{v}(\eps)) = x + \eps y$ lies in $\image(\varphi)$.
    It follows that $x = \lim_{\eps \to 0} \varphi(\tilde{v}(\eps))$ lies in the Euclidean closure of $\image(\varphi)$, and therefore in its Zariski closure.
\end{proof}

\subsection{Equivalent Definitions of Border Complexity}
\label{sec:equivalent}
In this section, we apply the algebro-geometric results developed earlier to show that the border complexity of a broad class of complexity measures can be characterized via one-parameter families of polynomials. These families can be interpreted either geometrically, as algebraic curves, or algebraically, as polynomials whose coefficients are univariate rational functions.

To define a complexity measure on polynomials, we typically describe a set of expressions or abstract devices (such as algebraic circuits) which compute polynomials and explain how to measure the {\em size} of a given expression.
The complexity of a polynomial is then defined as the size of the {\em minimal expression} for this polynomial.
To compute all polynomials over $\mathbb{C}$, the expressions must incorporate values from $\mathbb{C}$ as labels or parameters of some form.
In most cases, the coefficients of the computed polynomial depend algebraically on these parameters.
We formalize this construction in the following definition.

\medskip
\begin{definition}\label{def:parameterizable}
    Let $\Gamma$ be a complexity measure on polynomials.
    We define a \emph{$\Gamma$-expression} of size $s$ in $n$ variables as a polynomial $\varphi(\vect u,\x) \in \bbC[u_1, \dots, u_p][x_1, \dots, x_n]$
    such that $\Gamma(\varphi(\bar{u}, \x)) \leq s$ for every~$\bar{u} \in \bbC^p$.
    We say that $\Gamma$ is a \emph{parameterizable complexity measure} if for every $s, n \in \mathbb{N}$
    there exist a finite set $\Phi(s, n)$ of $\Gamma$-expressions such that every polynomial $f \in \bbC[\x]$ with $\Gamma(f) \leq s$ can be represented as $f = \varphi(\bar{u}, \x)$ for some $\varphi \in \Phi(s, n)$ and some vector $\bar{u}$.
\end{definition}

\medskip
\begin{example}[Circuit complexity]\label{ex:circuit-parameterizable}
    Circuit complexity is parameterizable.
    Define a \emph{circuit template} of size $s$ in the same way as a circuit of size $s$, but in every context where a constant from $\bbC$ appears in a label, use parameter-variables $u_i$ with $1 \leq i \leq s$ instead.
    Clearly, every circuit template $C$ of size $s$ with~$n$ input variables computes a polynomial $\varphi(\vect u,\x)$, and if we replace parameter-variables $u_i$ by constants $\bar{u}_i$, we obtain a circuit of size $s$ computing the polynomial $f = \varphi(\bar{u}; \x)$.
    In other words, every circuit template defines a circuit complexity expression.
    There is only a finite number of circuit templates of size at most $s$, and each circuit is obtained from some circuit template by replacing parameter-variables with constants.
    Thus, every polynomial $f$ with circuit complexity at most $s$ is covered by one of the circuit templates.
\end{example}

\medskip
\begin{example}[Determinantal complexity]
    Determinantal complexity is parameterizable, defined by the determinantal expressions $\det (u_{ij0} + \sum_{k = 1}^n u_{ijk} x_k)$ (with size of the expression being the size of the determinant).
\end{example}

\medskip
\begin{example}
    Waring rank of a polynomial is technically not a parameterizable complexity measure, but since it only applies to homogeneous polynomials, we only need to consider sequences of homogeneous polynomials of the same degree in the definition of border complexity.
    Restricted to degree $d$ homogeneous polynomials, Waring rank is a parametrizable complexity measure: every degree $d$ polynomial of rank at most $s$ can be obtained via substitution from the expression
    \[
    \varphi(\bar{u},\x) \;=\; \sum_{k = 1}^s \left(\sum_{i = 1}^n u_{ki} x_i\right)^d\;.
    \]        
\end{example}

\medskip
\begin{example}
    An important example of a \emph{non-parameterizable} complexity measure is the top fanin of a constant depth circuit, because the gates in the middle layers can have unbounded fanin and cannot be covered by a finite number of expressions.
    This will not be a significant problem, as we will see below that even for $\Sigma\Pi\Sigma$ circuits the border top fanin of every polynomial is $2$, and the construction does not require unbounded fanin gates; see~\cref{sec:depth3}.
\end{example}

A parameterizable complexity measure $\Gamma$ can also be used to measure complexity of polynomials in~$A[\x]$ where $A$ is an arbitrary algebra over~$\bbC$ in the same way it is used for polynomials over $\bbC$:
in every $\Gamma$-expression $\varphi \in \bbC[\vect u,\x]$ we can substitute a parameter vector $\tilde{u} \in A^p$ for $\vect u$, obtaining a polynomial in $A[\x]$. We define a complexity $\Gamma_A(\tilde{f})$ of a polynomial $\tilde{f} \in A[\x]$, as follows.

\medskip
\begin{definition}
For $f \in A[\x]$, the {\em $\Gamma$-complexity of $f$ over $A$}, denoted $\Gamma_A(\tilde{f})$, is the minimal number $s$ such that $\tilde{f} = \varphi(\tilde{u}; \x)$ for some $\Gamma$-expression $\varphi \in \bbC[\vect u,\x]$ of size $s$ and some $\tilde{u} \in A^p$.    
\end{definition}
In the following theorem, we prove an equivalence. and show that it suffices to focus on the computation of polynomials with coefficients in $\bbC((\eps))$ and $\bbC[\eps^{\pm 1}]$. 

\medskip
\begin{theorem}[Equivalence theorem]\label{thm:equivalent-definitions}
    Let $\Gamma$ be a parameterizable complexity measure
    and let $\underline{\Gamma}$ be the corresponding border complexity measure.
    Denote by $\mathcal{C}(s, n)$ the set of all polynomials in $x_1, \dots, x_n$ with complexity at most $s$.
    Let $f \in \bbC[x_1, \dots, x_n]$ be a polynomial.
    The following statements are equivalent:
    \begin{enumerate}
        \item $\underline{\Gamma}(f) \leq s$;
        \item $f$ is contained in the Zariski closure of $\mathcal{C}(s, n)$;
        \item there exists $\tilde{f} \in \bbC[[\eps]][\x]$ such that $\Gamma_{\bbC((\eps))}(\tilde{f}) \leq s$ and $\mathrm{Coeff}_{\eps^0}(\tilde{f}) = f$;
        \item there exists $\tilde{f} \in \bbC[\eps][\x]$ such that $\Gamma_{\bbC[\eps^{\pm 1}]}(\tilde{f}) \leq s$ and $\tilde{f}(0) = f$;
        \item there exists $\tilde{f} \in \bbC[\eps][\x]$ such that $\Gamma(\tilde{f}(\eps)) \leq s$ for all $\eps \neq 0$ and $\tilde{f}(0) = f$;
        \item there exists a curve $\mathcal{E} \subset \overline{\mathcal{C}(s, n)}$ such that $f \in \mathcal{E}$ and only a finite number of points of $\mathcal{E}$ lie outside of $\mathcal{C}(s, n)$;
    \end{enumerate}
\end{theorem}
\begin{proof}
    Let $\varphi \in \bbC[\vect u,\x]$ be a $\Gamma$-expression of degree $d$ with respect variables $x_1, \dots, x_n$.
    It gives rise to the following polynomial map:
\[
\hat{\varphi}\; \colon\; \bbC^p \;\to\;\bbC[\x]_{\leq d}\,,~~~~~~~~\bar{u} \in \bbC^p \;\mapsto\; \varphi(\bar{u}, \x)\;.
\]
    Every element of $\mathcal{C}(s, n)$ can be obtained from some $\Gamma$-expression from a finite set $\Phi(s, n)$,
    which means that 
    \[
    \mathcal{C}(s, n) \;=\; \bigcup_{\varphi \in \Phi(s, d)} \image(\hat{\varphi}) \;\subset\; \bbC[\x]_{\leq d}\;,
    \]
    where $d$ is the maximal among the $x$-degrees of $\Gamma$-expressions in $\Phi(s, n)$.
    By Chevalley's theorem (\cref{thm:chevalley}), every image in this union is constructible, so $\mathcal{C}(s, n)$ is also constructible and its Zariski closure coincides with the Euclidean closure, therefore,
    $\underline{\Gamma}(f) \leq s$ if and only if $f$ lies in the Zariski closure of $\mathcal{C}(s, n)$.

    Since $\overline{\mathcal{C}(s, n)} = \bigcup_{\varphi \in \Phi(s, n)} \overline{\image(\hat{\varphi})}$, every $f \in \mathcal{C}(s, n)$, lies in $\overline{\image(\hat{\varphi})}$ for some $\varphi \in \Phi(s, n)$.
    We may also choose $\varphi$ in such a way that $\overline{\image(\hat{\varphi})}$ is maximal, so it is an irreducible component of $\overline{\mathcal{C}(s, n)}$.
    The equivalences (2) $\Leftrightarrow$ (3) $\Leftrightarrow$ (4) now follow from~\cref{thm:image-closure}.

    If $\Gamma_{\bbC[\eps^{\pm 1}]}(\tilde{f}) \leq s$, then there is an expression $\varphi \in \Phi(s, n)$ such that $\tilde{f} = \hat{\varphi}(\tilde{u})$ for some $\tilde{u} \in \bbC[\eps^{\pm 1}]$.
    Substituting any nonzero value of $\eps$, we get that $\tilde{f}(\eps) = \hat{\varphi}(\tilde{u}(\eps))$, so $\Gamma(\tilde{f}_\eps) \leq s$.
    This proves the implication (4) $\Rightarrow$ (5).
    The implication (5) $\Rightarrow$ (1) follows directly from the definition of border complexity.

    To prove (2) $\Rightarrow$ (6), apply~\cref{thm:curve-existence} with $\mathcal{X} = \overline{\image(\hat{\varphi})}$, the point $f \in \mathcal{X}$, and a nonempty open set contained in $\image(\hat{\varphi})$, which exists by~\cref{lem:constructible-contains-open}.
    Since an nonempty open subset of the resulting curve $\mathcal{E}$ is contained in $\image(\hat{\varphi}) \subset \mathcal{C}(s, n)$, the part outside of $\mathcal{C}(s, n)$ lies in a nontrivial closed set, hence finite.

    The implication (6) $\Rightarrow$ (2) is trivial, as we have $f \in \mathcal{E} \subset \overline{\mathcal{C}(s, n)}$.
\end{proof}

\subsection{Debordering via Interpolation}

Let $\Gamma$ be a parameterizable complexity measure
and let $\underline{\Gamma}$ be the corresponding border complexity measure.
The equivalence from~\cref{thm:equivalent-definitions} can be used to represent polynomials $f$ with $\underline{\Gamma}(f) = s$ as linear combinations of polynomials with complexity at most $s$.
This, in turn, can be used to bound the non-border complexity of the original polynomial.
The resulting bound is in general too high, since it involves certain hard-to-control degree parameters.

The construction is usually done algebraically using the polynomial interpolation on expressions with $\eps$, but it also has a clear geometric representation, which we present first.

\medskip
\begin{definition}
    Let $\Gamma$ be a parameterizable complexity measure
    and let $\underline{\Gamma}$ be the corresponding border complexity measure.
    An \emph{approximating curve} for a polynomial $f$ with $\underline{\Gamma}(f) = s$ is a curve satisfying condition (6) of~\cref{thm:equivalent-definitions}.
\end{definition}

\medskip
\begin{lemma}
    Let $\mathcal{E}$ be an affine curve of degree $d$. Then $\dim \span~\mathcal{E} \leq d + 1$.
    Moreover, for every nonempty open $\mathcal{U} \subset \mathcal{E}$ there exist $d + 1$ points in $\mathcal{U}$ spanning $\mathcal{E}$.
\end{lemma}
\begin{proof}
    Consider $\mathcal{E}$ as a curve in $\span~\mathcal{E}$.
    Let $\mathcal{H}$ be a hyperplane not containing points of $\mathcal{E}$ outside of the open set $\mathcal{U}$
    By the Bézout's thorem, the intersection of $\mathcal{E}$ with $\mathcal{H}$ contains at most $d$ points.
    Since $\mathcal{H}$ lies in $\span~\mathcal{E}$, it is spanned by the points in the intersection,
    and the whole space $\span~\mathcal{E}$ is spanned by the hyperplane $\mathcal{H}$ and one point of $\mathcal{E}$ outside of $\mathcal{H}$, which can be taken from the open set $\mathcal{U}$.
\end{proof}

\begin{corollary}\label{thm:interpolation-geometric}
    If $\underline{\Gamma}(f) = s$ with an approximating curve $\mathcal{E}$ of degree $\deg \mathcal{E} = e$, then there exist $e + 1$ polynomials $f_1, \dots, f_{e + 1}$ with $\Gamma(f_i) \leq s$ and $e + 1$ coefficients $\alpha_1, \dots, \alpha_{e + 1}$ such that $f = \sum_{k = 1}^{e + 1} \alpha_i f_i$.
\end{corollary}

If the complexity measure $\Gamma$ is such that the complexity of linear combinations can be obtained from the complexity of the polynomials in the combination (for example, if $\Gamma$ is circuit complexity or a rank-type measure), then this statement can be used to obtain debordering results for $\Gamma$.

\medskip
\begin{definition}
    The \emph{geometric error-degree} of a polynomial $f \in \bbC[\x]$ with $\underline{\Gamma}(f) = s$ is the minimal degree of an approximating curve for $f$.
\end{definition}

From the proof of~\cref{thm:equivalent-definitions} we see that the geometric error-degree of a polynomial $f \in \bbC[\x]$ with border complexity $\underline{\Gamma}(f) = s$ is bounded by the maximal among the degrees of irreducible components of $\overline{\mathcal{C}(s, n)}$.
In general, this is only bound on the error-degree available.

As an example, we prove a debordering result for circuit complexity similar to~\cite{Bur04correction},
which we denote by $L(f)$.
We will need the following degree bound for images of polynomial maps.

\medskip
\begin{lemma}[{\cite[Theorem~8.48]{bcs97}}]\label{lem:degree-upper-bound}
    If $\varphi \colon \bbC^m \to \bbC^n$ is a polynomial map with $\deg \varphi_i \leq d$, then $\deg \overline{\image(\varphi)} \leq d^m$.
\end{lemma}

\medskip
\begin{theorem}\label{thm:deborder-circuits}
    If $\underline{L}(f) \leq s$, then $L(f) \leq (3\cdot 2^{s^2} + 2) s$
\end{theorem}
\begin{proof}
    Consider the parameterization by circuit templates from~\cref{ex:circuit-parameterizable}.
    Every polynomial expression $\varphi$ of size $s$ defined by a circuit template computes a polynomial of degree $2^s$ in at most $s$ parameters, so the degree of the corresponding irreducible component is bounded by $e = (2^s)^s$.
    By~\cref{thm:interpolation-geometric}, there exists $(e + 1)$ many polynomials $f_1, \dots f_{e + 1}$ with $L(f_i) \leq s$ such that $f$ is a linear combination of $f_i$.
    This linear combination can be computed from $f_i$ using $(e + 1)$ constant gates, $(e + 1)$ multiplications and $e$ additions, so 
    \[
    L(f) \;\leq\; (3e + 2) s\;\leq\; (3 \cdot 2^{s^2} + 2) \cdot  s\;.
    \]
\end{proof}
The proof of~\cref{thm:interpolation-geometric} does not give explicitly the coefficients of the linear combinations.
They can be made more explicit by using a more algebraic version of the statement.

\medskip
\begin{definition}[Error-degree]
    The \emph{error-degree} of a polynomial $f \in \bbC[\x]$ with $\underline{\Gamma}(f) = s$ is the minimum degree $e$ such that there exists a polynomial $\tilde{f} = f + \eps f_1 + \dots + \eps^e f_e$, where $\Gamma(f(\eps)) \leq s$ for all $\eps \neq 0$.
\end{definition}

\medskip
\begin{theorem}\label{thm:interpolation-algebraic}
    If $\underline{\Gamma}(f) = s$ with error-degree $e$, then there exists $e + 1$ polynomials $f_1, \dots, f_{e + 1}$ with $\Gamma(f_i) \leq s$ and $e + 1$ coefficients $\alpha_1, \dots, \alpha_{e + 1}$ such that $f = \sum_{k = 1}^{e + 1} \alpha_i f_i$.
\end{theorem}
\begin{proof}[Proof sketch]
    The statement is true by polynomial interpolation.
    Take $f_i = \tilde{f}(i)$ and $\alpha_i$ to be the interpolation coefficients required to recover the constant term $f$ from $f_i$.
\end{proof}

The bounds for the error-degree are also typically {\em exponential}.
The following bound can be derived from the bound on valuation in~\cref{thm:image-closure}.

\medskip
\begin{theorem}
    Let $\Phi(s, n)$ be the finite set of $\Gamma$-expressions of size $s$ covering all polynomials with complexity at most $s$.
    For an expression $\varphi \in \bbC[\vect u,\x]$ let $e(\varphi) = q^{p + 2}$ where $q$ is the degree of $\varphi$ with respect to the parameter-variables $u_i$.
    Then the error-degree of every polynomial $f$ with $\underline{\Gamma}(f) = s$ is at most $\max \{ e(\varphi) \mid \varphi \in \Phi(s, n)\}$.
\end{theorem}
\begin{proof}
    Let $\varphi \in \Phi(s, n)$ and $\tilde{u} \in \bbC[\eps^{\pm 1}]$ be such that $\hat{\varphi}(\tilde{u}) = \tilde{f}$ with $\tilde{f}(0) = f$, where $\hat{\varphi}$ is defined as in the proof~\cref{thm:equivalent-definitions}.
    From the proof of~\cref{thm:equivalent-definitions} it follows that $\val_\eps(\tilde{u}) \geq -D$ and $\deg_\eps \tilde{u} \leq Dq$ where $q$ is the degree of $\varphi$ with respect to $u$ and $D$ is the degree of the graph of $\hat{\varphi}$.
    Applying the map $\hat{\varphi}$, we see that $\deg_\eps \tilde{f} \leq Dq^2$.
    From~\cref{lem:degree-upper-bound} we have $D \leq q^{s}$ where $s$ is the number of the parameter-variables $u_i$.
    The result follows.
\end{proof}

As we see, the bounds we get from the general construction are usually exponential.
On the other hand, case-by-case analysis of small cases show that at least for small values of reasonable complexity measures the error-degrees can be made small.
The bounds can likely be improved in general, but not with readily available methods.

\medskip
\begin{question}
    Prove better bounds for the error-degrees of specific polynomials with respect to circuit complexity or other common complexity measures.
\end{question}

The interpolation technique from~\cref{thm:interpolation-geometric} and~\cref{thm:interpolation-algebraic} can be used in cases when the expressions in $\eps$ involved are restricted. 
For example, Grochow, Mulmuley, and Qiao~\cite{DBLP:conf/icalp/GrochowMQ16} introduce the notion of $p$-definable degenerations, which are expressions of the form $g(A(\eps) x + a_0(\eps))$ where the entries of the matrix $A(\eps)$ and $a_0(\eps)$ have coefficients computable in terms of the binary representations of matrix indices and the power of $\eps$, so that the interpolation performed in~\cref{thm:interpolation-algebraic} can be implemented by a combinatorial hypercube sum as in the definition of $\VNP$.
They use this to deborder a subset of $\overline{\VNP}$ defined in terms of these $p$-definable degenerations. 


\subsection{A Presentable Version of Border
Complexity}
Recall that a class $\mathcal{C}$ is said to be border-closed if $\overline{\mathcal{C}} = \mathcal{C}$. While one might intuitively expect a class and its closure to be closely related, this is far from evident. The standard definition of approximation allows the use of arbitrary polynomials in $\epsilon$, potentially of unbounded complexity, as coefficients—making the notion inherently non-constructive and existential in nature. As a result, even fundamental questions remain open; for instance, it is not known whether $\overline{\VP} \subseteq \VNP$.

To address this, Bhargav, Dwivedi, and Saxena~\cite{bhargav2024learning} recently introduced a natural constructive refinement of approximation, termed {\em presentability}. The corresponding presentable class, denoted $\overline{\VP}_{\epsilon}$, captures the essence of approximation while enforcing additional structure that makes the definition more explicit and algorithmically meaningful. It is defined as follows:

\medskip
\begin{definition}[Presentable $\approxbar{\VP}$,{\cite[Definition~4.10]{bhargav2024learning}}]
$(f_n) \in \approxbar{\VP}_{\eps}$ over $\F$, if there
is an approximating polynomial $g_n \in \F[\eps][\x]$ satisfying
\[
g_n(\x,\eps) \;=\; \eps^M \cdot f_n(\x) + \eps^{M+1} \cdot S_n(\x,\eps)
\]
for some $S_n \in \F[\eps][\x]$, and $M \in \N$. Moreover, $\size_{\F}(g)$ and $\deg_{\x}(g)$ are bounded by $\poly(n)$.
\end{definition}
Note that, an additional condition in the definition of $\approxbar{\VP}_{\eps}$ is that all the polynomials in $\eps$ used as `constants'
in the approximating circuit $g_n(\x,\eps)$, have polynomial-size circuits themselves.

An earlier approach to making border classes more explicit was proposed through the notion of `degenerations'~\cite{DBLP:conf/icalp/GrochowMQ16}, specifically via what the authors call {\em $p$-definable one-parameter degenerations}. In this framework, the coefficients of the $\epsilon$-polynomials used in the approximation are required to be generated by arithmetic circuits in $\VP$. While this offers a more structured view of $\overline{\VP}$, the resulting subclass is still quite restricted.

In contrast, the presentable border class, introduced in~\cite{bhargav2024learning}, offers a more natural and general refinement of the standard notion of approximation. It allows approximation via structured and efficiently describable families of polynomials, yet does not constrain them to arise solely through p-definable degenerations. As a result, the presentable class $\overline{\VP}_{\epsilon}$ is incomparable with the class obtained via p-definable degenerations, and cannot be realized as a degeneration of $\VP$ in the sense of~\cite{DBLP:conf/icalp/GrochowMQ16}.

This framework naturally extends beyond $\VP$ and can be used to define the presentable border of $\VNP$ as well. In particular, one can define $\overline{\VNP}_{\epsilon}$ over any field $\F$ in a similar fashion.

\medskip
\begin{definition}[Presentable $\approxbar{\VNP}$,{\cite[Definition~1.2]{bhargav2024learning}}]
$(f_n) \in \approxbar{\VNP}_{\eps}$ over $\F$, if there
is an approximating polynomial $g_n \in \F[\eps][\x]$ satisfying
\[
g_n(\x,\eps) \;=\; \eps^M \cdot f_n(\x) + \eps^{M+1} \cdot S_n(\x,\eps)
\]
for some error polynomial~$S_n \in \F[\eps][\x]$, and order $M \in \N$. Moreover, there exists a verifier polynomial $h \in \F[x_1,\cdots,x_n,y_1,\cdots,y_m,\eps]$, with $m, \deg_{\x,\y}(h)$ and $\size_{\F}(h)$ all bounded by $\poly(n)$, satisfying a
hypercube-sum expression:
\[
\sum_{\a \in \{0,1\}^m} h(\x,\a,\eps) \;=\;g(\x,\eps)\;.
\]
\end{definition}
In~\cite{bhargav2024learning}, the authors showed an efficient debordering of the presentable $\approxbar{\VNP}$.

\medskip
\begin{theorem}[Presentable is Explicit,~{\cite[Theorem 1]{bhargav2024learning}}] \label{thm:vnpeps}
Over a finite field, $\approxbar{\VNP}_{\eps} = \VNP$.
\end{theorem}

\begin{proof}[Proof sketch of~{\cref{thm:vnpeps}}]
Although interpolation may initially seem unhelpful in this context, the core of the proof in fact hinges on a clever use of interpolation. By definition, we have the inclusion $\VNP \subseteq \overline{\VNP}_{\epsilon}$. To establish the reverse inclusion, we appeal to~\cref{valiant-criterion}, which asserts that any polynomial of low degree whose coefficients are efficiently computable in the Boolean setting also lies in $\VNP$ in the algebraic setting. We carry out the argument over a finite field $\F_q$, where $q = p^t$ for some prime $p$.

Let $f = \sum_{\e} c_{\e}\x^{\e} \in \approxbar{\VNP}_{\eps}$. We would like to show that $\phi: \e  \mapsto c_{\e}$, is computable in $\#\P/\poly$. By definition, there exist $g(\x,\eps), h(\x,\y,\eps)$ and $S(\x,\eps)$ such that 
\[
\sum_{\a \in \{0,1\}^m} h(\x,\a,\eps) \;=\; g(\x,\eps) \;=\;\eps^M \cdot f + \eps^{M+1} \cdot S\;.\]
To extract the coefficient of $\epsilon^M \x^{\e}$ from a polynomial $g$, one can strategically choose interpolation points to be roots of unity of large (specifically, exponential) multiplicative order. This careful choice allows the recovery of the desired coefficient $c_{\e}$ as a hypercube sum over an algebraic circuit of polynomial-size, albeit with exponential degree. While the full proof proceeds via a delicate inductive argument, we outline the core idea in the base case, which considers univariate polynomials of exponential degree. 

\paragraph{Unfolding univariate interpolation.}~Let $G = \sum_{e} C_e y^e$, of degree $D = 2^s$, such that $G = \sum_{\a} H(y,\a)$, such that $\size(H) \le s$. We would like to express $c_{e}$ as an exponential sum of structured small circuits. Note that, by simple Vandermonde inverse, we have 
\[C_{e}\;=\;\sum_{i=0}^{k-1}\,\frac{\omega^{-ei}}{k} \cdot G(\omega^i)\;=\;\sum_{\a}\sum_{i=0}^{k-1}\,\frac{\omega^{-ei}}{k} \cdot H(\omega^i,\a)\;,
\]
where $2^s = D < k < \Theta(k)$, and $\omega$ is a root of unity of order $k$. A direct circuit computing the inner sum in the expression above would, in general, have exponential size in the parameter $s$. However, an elegant workaround is to express this sum as a hypercube sum, by cleverly encoding the powers of $\omega$ using the binary representation of the exponent. Let $r=\lceil \log k \rceil$, and define the polynomial
\[
\hat{H}(z,z_1, \cdots, z_r) \;:=\; \prod_{i=1}^r (z_i \cdot z_i \cdot z^{2^{i-1}} + 1 -z_i)\;.
\]
Now, for any integer $i \in {0, \ldots, k-1}$, let $\vect i = (i_1, \ldots, i_r)$ denote its binary representation. It is easy to verify that $\hat{H}(\omega,\vect i) = \omega^i$. Therefore, one can rewrite the coefficient $C_e$ as follows:
\[
C_e \;=\; \sum_{\a}\sum_{\vect i} \frac{1}{k} \cdot \hat{H}(\hat{H}(\omega^{-1},e),\vect i) \cdot H(\hat{H}(\omega,\vect i), \a)\;.
\]
The inner sum described above can be interpreted as the evaluation of a circuit $t_{\e}$ at the input $(\a, \vect i)$, where $\a$ encodes auxiliary parameters and $\vect i$ ranges over the binary hypercube. Crucially, this circuit has size $\size(t_{\e}) = O(s)$, where $s$ is the size of the original representation. This implies that the entire sum can be expressed as a hypercube sum over evaluations of a small circuit.

The construction naturally extends inductively to the multivariate setting. Furthermore, it is well known that a hypercube sum over a polynomial-size algebraic circuit corresponds to a function in $\#\P/\poly$. Combining these observations with Valiant's criterion~\cref{valiant-criterion}, we conclude that the coefficient $c_{\e}$ is computable in $\#\P/\poly$, and hence $f \in \VNP$.
\end{proof}

\section{Debordering Results}
\label{sec:debordering}
This section describes debordering results for various restricted models of computation. Due to various structural results in algebraic circuits \cite{Agrawal08,Gupta16}, it is known that a strong debordering result of restricted classes like depth-3 and depth-4 will lead to significant progress
in understanding the difference between Valiant’s determinant vs permanent conjecture~\cite{Valiant79}, and Mulmuley and Sohoni’s variation which uses border determinantal
complexity~\cite{MS01}. Thus, restricted classes not only provide various challenges to generate new techniques but they can also be seen as stepping stones toward the general problem. 

Due to discussions in \cref{sec:equivalent}, we will work over $\C[\eps^{\pm 1}]$ from now on.
\subsection{Debordering ROABPs}
It turns out that for a single ROABP,
border does not add any power, i.e. $\approxbar{\ROABP}(w,n,d) = \ROABP(w,n,d)$, where $\ROABP(w,n,d)$ denotes the class of $n$-variate polynomials computed by width-$w$ ROABPs of individual degree at most $d$; we omit $w,n,d$ for the simplicity of notation.

\medskip
\begin{lemma}[\cite{Forbes16}]\label{lem:aro}
A polynomial $f \in \F[\x]$ in the border class of width $w$ ROABPs can also be computed by an ROABP of width at most $w$. The same holds for $\ARO$s.
\end{lemma}
\begin{proof}
Let $g = f + \eps \cdot S$, where $g$ can be computed by an ROABP of width $w$ over $\F[\eps^{\pm 1}]$. We need to show that $f$ can also be computed by an ROABP of width $\le w$, over $\F$. Let the unknown variable order of $g$ be $(y_1, \cdots, y_n)$. By applying Nisan’s characterization~\cref{lem:ROABPdim} on the polynomial $g$, we know that for all $k \in [n]$, the partial derivative matrix for each layer $M_k$ has rank at most $w$ over $\F[\eps^{\pm 1}]$. This means that
the determinant of any $(w+1) \times (w+1)$ minor of $M_k$ is identically zero. Observe that the entries of $M_k$ are coefficients of monomials of $g$ which are in $\F[\eps][\x]$. Thus, the determinant polynomial will remain zero even under the limit $\eps \to 0$. Hence, for $f \simeq g$, each matrix $M_k$  also has rank at most $w$ over $\F$. Therefore, by~\cref{lem:ROABPdim}, $f$ also has an ROABP of width at most $w$. Since the variable order did not change in the proof, it also holds for $\ARO$s.
\end{proof}
Although a single ROABP is closed under the border, it is unclear if the class consisting of sum of a constant number of ROABPs is equal to its border class.

\medskip
\begin{question}
Characterize the border of the sum of two ROABPs (possibly of different variable order) of width at most $w$.
\end{question}
\subsection{Debordering Depth-2 Circuits} \label{sec:depth-2}

A depth-$2$ circuit with the top gate being `$+$' gate denoted~$\Sigma\Pi$, often referred as {\em sparse polynomials}, computes a polynomial of the form
\begin{equation}
f(\x)\;=\;\sum_{i=1}^s\,c_{\e_i}\x^{\e_i},~\text{where}~c_{\e_i} \in \C\;.
\end{equation}
We use $d$ to denote the total degree bound for $f(\x)$. We use~$\Sigma\Pi(s,n,d)$ to denote the set of all such depth-$2$ circuits. 

On the other hand a depth-$2$ circuit with the top gate being `$\times$' gate denoted~$\Pi\Sigma$, often referred as {\em product of linear polynomials}, computes a polynomial of the form
\begin{equation}
f(\x)\;=\;\prod_{i=1}^d\,\ell_{i},~\text{where}~\ell_{i}~\text{are linear polynomials}\;.
\end{equation}
It is clear that $\deg(f) =d$. We use~$\Pi\Sigma(n,d)$ to denote the set of all such depth-$2$ circuits computing $n$-variate polynomials of total degree at most $d$. When the parameters $n$ and $d$ are clear, we will omit them and write $\Sigma\Pi(s)$, and $\Pi\Sigma$.

\paragraph{Debordering $\overline{\Sigma\Pi(s)}$.}~We will argue that $\overline{\Sigma\Pi(s)} = \Sigma\Pi(s)$. Obviously, $\Sigma\Pi(s) \subseteq \overline{\Sigma\Pi(s)}$. To see the other direction, let $f \in \overline{\Sigma\Pi(s)}$. By definition, there exists $g \in \Sigma\Pi(s)$, over $\C[\epsilon^{\pm 1}]$, such that $g = f + \epsilon \cdot S$, where $S \in \C[\epsilon][\x]$, i.e.~$f = \cf_{\epsilon^0}(g)$. Let $g(\x,\epsilon)= \sum_{i=1}^s\,c_{\e_i}\x^{\e_i}$, where $c_{\e_i} \in \C[\epsilon^{\pm 1}]$. Therefore, comparing $\epsilon$-degree, it is easy to see that each $c_{\e_i} \in \C[\epsilon]$, and hence,
\[
f\;=\;\cf_{\epsilon^0}(g)\;=\;\sum_{i=1}^s\,\cf_{\epsilon^0}(c_{\e_i}) \cdot \x^{\e_i}\;.
\]
In particular, this means that~$\sp(f) \le s$ implying $f \in \Sigma\Pi(s)$.

\paragraph{Debordering $\overline{\Pi\Sigma}$.}~We will argue that $\overline{\Pi\Sigma} = \Pi\Sigma$. Obviously, $\Pi\Sigma \subseteq \overline{\Pi\Sigma}$. To prove the other direction, let $f \in \overline{\Pi\Sigma}$.  By definition, there exists $g \in \Pi\Sigma$, such that $g = f + \epsilon \cdot S$, where $S \in \C[\epsilon][\x]$, i.e.~$f = \cf_{\epsilon^0}(g)$, and further assume that $f$ is a nonzero polynomial. Let $g = \prod_{i=1}^d\,\ell_i$, where $\ell_i \in \C[\epsilon^{\pm 1}][\x]$ are linear polynomials. Assume that $\val_{\epsilon}(\ell_i) = a_i$. Therefore, we can write $\ell_i = \epsilon^{a_i} \cdot \left(\sum_{j=0}^M \ell_{i,j} \epsilon^j\right)$, for some positive integer $M$, where $\ell_{i,j} \in \C[\x]$ are linear polynomials (not necessarily nonzero), and further by assumption, $\ell_{i,0}$ is nonzero. By definition of $g$, and the assumption on the nonzeroness of $f$, we have $0 =\val_{\epsilon}(g) = a_1 + \cdots + a_n$. Hence,
\[
f \;\simeq\;g \;\simeq\; \prod_{i=1}^d\,\left(\sum_{j=0}^M \ell_{i,j} \epsilon^j\right) \;\simeq\; \prod_{i=1}^d \ell_{i,0} \;\in\; \Pi\Sigma\;.
\]
In fact, a similar proof as above shows a much more general theorem for a class~$\calC$: 
\[\overline{\Pi\calC} \;\subseteq\; \Pi \overline{\calC}\;.\] 
Further, if $\calC$ is {\em closed} under approximation, i.e.~$\overline{\calC} = \calC$, then $\overline{\Pi\calC} = \Pi \calC$, since the following chain of containments holds: 
\[
\overline{\Pi\calC} \;\subseteq\; \Pi \overline{\calC}\;\subseteq \; \Pi \calC\;\subseteq\;\overline{\Pi \calC}\;.\]

\subsection{Debordering Border Waring Rank} \label{sec:waring}

Given a homogeneous polynomial $f$ of degree $d$ over $\C$, its \emph{Waring rank} $\WR(f)$ is defined as the smallest number $k$ such that the following holds:
\[f \;=\; \sum_{i = 1}^k (a_{i1} x_1 + \cdots + a_{in}x_n)^{d},\]
where $a_{ij} \in \C$. Saxena~\cite{Saxena08} introduced depth-3 diagonal circuits. They are denoted by $\Sigma\wedge\Sigma$, and they compute polynomials of the form 
\begin{equation} 
\label{eq:depth-3-diag}
f(\x) \;=\; \sum_{i=1}^k (a_{i0} + a_{i1} x_1 + \cdots + a_{in}x_n)^{d_i}\;,
\end{equation}
where $a_{ij} \in \C$. Let $d = \max d_i$. Then, the Waring rank $\WR(f)$ is the minimal top fanin of a homogeneous $\Sigma \wedge \Sigma$ circuit computing $f$.

Let $\VW(k,n,d)$ denote the set of homogeneous $n$-variate polynomials of degree $d$, with Waring rank at most $k$. Similarly, one can define $\SES(k,n,d)$. We will omit $k,n,d$, when they are polynomially related, and simply write $\VW$. 

In the case of quadratic forms (polynomials of degree $2$), Waring rank is equivalent to the rank of the symmetric matrix associated with a quadratic form;
hence Waring rank can be regarded as a generalization of the rank of a symmetric matrix.
Unlike the case of matrices, when $d \geq 3$, Waring rank is in general not lower semicontinuous \footnote{A function $f$ is lower semicontinuous at $a$ if $\underset{x \rightarrow a}{\lim \inf}~f(x) \ge f(a)$.}, that is, a limit of a family of polynomials with low Waring rank can have higher Waring rank.
The simplest example is given by the polynomial $x^{d - 1} y$, which has Waring rank $d$ (this is a classical result~\cite{Oldenburger}), but can be presented as a limit
\[
    x^{d - 1} y \;=\; \lim_{\epsilon\to 0}\; \frac{1}{d \eps} \left[(x + \epsilon y)^d - x^d\right]
\]
of a family of Waring rank $2$ polynomials (note that we work over $\bbC$, so this expression can be rearranged into a sum of two powers by moving constants inside the parentheses).
The \emph{border Waring rank} is a semicontinuous variation of Waring rank defined as follows:
the border Waring rank of~$f$, denoted $\bwr(f)$, is the smallest $r$ such that $f$ can be written as a limit of a sequence of polynomials of Waring rank at most $r$.
We have $\bwr(x^{d - 1} y) = 2$ and $\WR(x^{d - 1} y) = d$.
There exist examples of polynomials of degree $d$ with $\WR(f)/\bwr(f) = d - o(1)$ (Zuiddam~\cite{ZUIDDAM201733} gives such examples for tensor rank, but a similar example also works for Waring rank).

One can ask how powerful is the border of $\VW$? It turns out that the border of $\VW$ is not too powerful:

\medskip
\begin{lemma}[{\cite{Forbes16,blaser2020complexity}}] \label{bordervw-in-vbp}
$\approxbar{\VW} \subsetneq \VBP$.  
\end{lemma}
As an obvious corollary, we get $\approxbar{\SES} \subsetneq \VBP$.
\begin{proof}[Proof sketch]
Let $\bwr(f)= s$. One can argue that $f$ can also be computed by an $\ARO$ of width $O(snd)$. The key ingredient for the lemma is the {\em duality trick}. 
\begin{lemma}[Duality trick \cite{Saxena08}] \label{lem:duality}
        The polynomial $f = (x_1 + \hdots +x_n)^d$ can be written as
   \[
           f \;=\;\sum_{i \in [t]}\, f_{i1}(x_1) \cdots f_{in}(x_n),
      \]
\noindent where $t = O(nd)$, and $f_{ij}$ is a univariate polynomial of degree at most $d$. 
\end{lemma}    
\noindent By assumption, $\sum_{i=1}^s\, \ell_i^d\,=\,g \,=\, f + \eps \cdot S$, where $\ell_i \in \C[\eps^{\pm 1}][\x]$ are homogeneous linear forms. Using~\cref{lem:duality} on each $\ell_i^d$, we get that $g$ can be computed by a small $\ARO$. Since, $\approxbar{\ARO} = \ARO$~\cref{lem:aro}, we have $\approxbar{\VW} \subseteq \VBP$.

On the other hand, by simple partial derivatives, one can conclude that $\bwr(x_1 \cdots x_n) \geq \binom{n}{\lfloor n/2 \rfloor}$, which is exponentially large as a function of $n$. This shows that $\approxbar{\VW} \subsetneq \VBP$ is strict.
\end{proof}

It is easy to show the following identity holds (often known as {\em Fischer's formula}):
\[
x_1 \cdots x_n \;=\; \frac{1}{n! 2^{n-1}} \cdot \sum_{s_2, \cdots, s_n \in \{\pm 1\}} \left(\prod_{i=2}^n s_i\right) (x_1 + s_2 x_2 + \cdots + s_n x_n)^n\;.
\]
However, it is not known whether this is the best bound possible for the border Waring rank of a monomial.

\medskip
\begin{question}
$\bwr(x_1 \cdots x_n) = 2^{n-1}$.
\end{question}
It is known that $\WR(x_1 \cdots x_n)=2^{n-1}$ \cite{CARLINI20125}. The same result on cactus rank (a scheme-theoretic version of Waring rank) is proved in \cite{ranestad2011rank}. There are (at least) two incorrect/incomplete proofs available online of the same result for border rank: the
early versions of \cite{Oed16}, and the first version of \cite{CGO:19}. A discussion on the gaps in the proofs is
available in the first version of \cite[Section 6.1]{BB-apolarityarxiv}.

\paragraph{Fixed-parameter debordering of Waring rank.} Another debordering result connects the border Waring rank of a polynomial with its usual Waring rank, giving an upper bound which is polynomial in the degree, but exponential in the border Waring rank parameter.
This means that polynomial families with fixed (or even logarithmic) border Waring rank are in $\VW$.

\medskip
\begin{theorem}[\cite{DBLP:conf/stacs/DuttaGIJL24}]
    If $f$ is a homogeneous polynomial of degree $d$ and border Waring rank $\bwr(f) = r$, then $\WR(f) \leq \binom{2r - 2}{r - 1} \cdot d$.
\end{theorem}
In particular, the above theorem proves that when $r$ is constant, the Waring rank is at most $O(d)$. Before this work, explicit debordering for $r \le 5$ was known~\cite{landsberg2010ranks,ballico2018ranks}. 
\begin{proof}[Proof sketch]
    As the first step, we show that a polynomial with $\bwr(f) \leq r$ can be transformed into a polynomial in $r$ variables by a linear substitution.
    This is based on the fact that the space of first order partial derivatives of $f$ has dimension at most $r$, and a polynomial does not depend on variables with respect to which the derivative is zero.

    The case $d < r - 1$ is trivial, as in this case $\dim \bbC[\bold{x}]_d < \binom{2r - 2}{r - 1}$ and there exists a basis consisting of powers of linear forms, so every polynomial has Waring rank at most $\binom{2r - 2}{r - 1}$.

    To handle the nontrivial case $d \geq r - 1$ we transform the border rank decomposition into a \emph{generalized additive decomposition}:
    \begin{equation}\label{eq:-gad}
         f \;=\; \sum_{k = 1}^m \ell_k^{d - r_k + 1} g_k\;,
    \end{equation}
    with $r_1 + \dots + r_m = r$ where $\ell_k \in \bbC[\bold{x}]_1$, $g_k \in \bbC[\bold{x}]_{r_k - 1}$, and moreover $\bwr(\ell_k^{d - r_k + 1} g_k) \leq r_k$.
    In this decomposition again the Waring ranks of $g_k$ can be bounded from above using the trivial bound $\binom{2 r_k - 2}{r_k - 1}$, which implies the bound of $\binom{2 r_k - 2}{r_k - 1} \cdot d$ on the Waring ranks of summands, and 
    \[
    \left(\sum_{k = 1}^m \binom{2 r_k - 2}{r_k - 1}\right) d \;\leq\; \binom{2 r - 2}{r - 1} \cdot d\;
    \]
    on the total Waring rank of $f$.

    To obtain the generalized additive decomposition, we introduce an intermediate step: we separate the border rank decomposition
    \[ f = \lim_{\eps \to 0} \sum_{k = 1}^r \ell_k^d \]
    where $\ell_k \in \bbC[\eps^{\pm 1}][\bold{x}]_1$,
    into several \emph{local border rank decompositions} of the form,
    \[
        f_i \;=\; \lim_{\eps \to 0} \sum_{k = 1}^{r_i} \left(\eps^{q_k} \gamma_k \ell + \sum_{q = q_k + 1}^{q'_k} \eps^{q_k}\ell_{kq} \right)^d\;.
    \]
    That is, decompositions in which the linear form in the lower degree of $\eps$ is the same up to scaling.
    To prove that such separation is possible, that is, the parts having the same tend to a limit as $\eps \to 0$ independently, we use the following two facts.
    
    \begin{claim}
        The linear span of the $r_k$ powers $\left(\eps^{q_k} \gamma_k \ell + \sum_{q = q_k + 1}^{q'_k} \eps^{q_k}\ell_{kq} \right)^d$ tends to a subspace of $\ell^{d - r_k + 1}\bbC[\bold{x}]_{r_k - 1}$, as $\eps \to 0$
    \end{claim}
    This claim is proven by induction on degree using partial derivative methods.
    \begin{claim}
        The sum of spaces of polynomials of the form $\ell_k^{d - r_k + 1}\bbC[\bold{x}]_{r_k - 1}$ is direct if $\ell_k$ are distinct and $d \geq \sum_k r_k - 1$.
    \end{claim}
    This is a classical fact for bivariate polynomials, and the general case can be reduced to the bivariate case.
    
    Taken together, the claims say that if we group summands of the border rank decomposition by the linear forms in the lowest degree of $\eps$, we get subexpressions of the form $\eps^{p_k} \ell_k^{d - r_k + 1} g_k + \dots$, and the main terms of different subexpression cannot cancel each other, so $p_k \geq 0$ and we obtain a generalized additive decomposition.
    
    Alternatively, the existence of the generalized additive decomposition can be proven using algebro-geometric methods from~\cite{BBM-ranks,BuczBucz:SecantVarsHighDegVeroneseReembeddingsCataMatAndGorSchemes,BB-wild} involving $0$-dimensional schemes. For details, we refer to~\cite{DBLP:conf/stacs/DuttaGIJL24}.
\end{proof}

Shpilka~\cite{shpilka2025improved} improved the dependence of the debordered rank on $r$ by introducing a refined form of local decomposition, inspired by a {\em diagonalization trick}. The central idea is that, after an appropriate linear transformation and perturbation, the variable $x_i$ can be eliminated from $g_1, \cdots, g_{i-1}$ in \cref{eq:-gad}. With a suitable choice of parameters, this leads to the following result.
\medskip
\begin{theorem}[\cite{shpilka2025improved}]
    If $f$ is a homogeneous polynomial of degree $d$ and border Waring rank $\bwr(f) = r$, then $\WR(f) \leq r^{10 \sqrt{r}} \cdot d$
\end{theorem}

As a corollary, we get a polynomial upper bound for families with polylogarithmic border Waring rank.

\medskip
\begin{corollary}
Let $f_n$ be a polynomial family, with $\deg(f_n)$ being polynomially bounded, and further $\bwr(f_n) = O\left(\left(\frac{\log n}{\log \log n}\right)^2\right)$, then $(f_n) \in \VW$.
\end{corollary}

We remark that the current best Waring and Border Waring bounds for determinant and permanent are {\em different}. Using partial derivative methods, one can show that both $\bwr(\det_n)$ and $\bwr(\per_n)$ are lower bounded by $\binom{2n}{n}$. On the other hand, $\bwr(\per_n) \leq \WR(\per_n) \le 4^n$~\cite{DBLP:journals/ejc/Glynn10}, an almost tight upper bound. However, the best Waring rank and border Waring rank upper bound for $\det_n$ are still $2^{O(n \log n)}$~\cite{DBLP:journals/cpc/HoustonGJ24}. If both lower bound for border Waring rank and upper bound for Waring rank of the determinant are asymptotically tight, then the best debordering result we can hope is the following conjecture.

\medskip
\begin{question}
If $f$ is a homogeneous polynomial of degree $d$ and border Waring rank $\bwr(f) = r$, then $\WR(f) \leq r^{O(\log r)} \cdot \poly(d)$. 
\end{question}
A stronger version of the above conjecture is to prove a polynomial upper bound i.e. $\WR(f) \le \poly(rd)$.
An even stronger conjecture $\WR(f) \le (r - 1)(d - 1) + 1$ was proposed by Ballico and Bernardi in~\cite{ballicobernardi2012}.
It is known to hold for $r \leq 5$~\cite{landsberg2010ranks,ballico2018ranks}.
We remark that for $r \le 5$ the $\eps$-degree of each linear form in the border rank decomposition inside each linear form for is at most $r-1$.
Christian Ikenmeyer, in private communication, conjectured that this holds true for any $r$, which would imply Ballico--Bernardi conjecture.

\subsection{Debordering Bounded Depth-3 Circuits}
\label{sec:depth3}
A depth-3 circuit with top gate `$+$' denoted~$\Sigma\Pi\Sigma$ computes a polynomial of the form 
\begin{equation}\label{eq:depth-3}
f(\x) \;=\; \ell_{1,1} \cdots \ell_{1,d_1} \;+\; \cdots + \ell_{k,1}\cdots \ell_{k,d_k}\;,  
\end{equation}
where $\ell_{i,j}$ are linear polynomials in $\C[\x]$; and $d_i$ are some parameters. The {\em top fanin} is $k$. We use~$\SPS(k,n,d)$ to denote the set of depth-3 circuits of the form~\cref{eq:depth-3} where $\deg(f) \le d := \max_{i} d_i$. When $n$ and $d$ are polynomially related, we will often omit them and simply write $\SPS(k)$. In algebraic geometry, when \cref{eq:depth-3} is homogeneous, $k$ is called the {\em Chow rank}.

A depth-3 circuit with top gate `$\times$' denoted~$\Pi\Sigma\Pi$ computes a polynomial of the form 
\begin{equation*}
f(\x) \;=\; \prod_{i=1}^k g_i\;,  
\end{equation*}
where $g_i$ are sparse polynomials. We use~$\Pi\Sigma\Pi(s,k)$ to denote the set of depth-3 circuits with top gate being a product gate with fanin $k$ such that each $\sp(g_i) \le s$, and $\deg(f) \le d$. Since, $\overline{\Sigma\Pi(s)} = \Sigma\Pi(s)$,  from the discussion in~\cref{sec:depth-2}, we can conclude that $\overline{\Pi\Sigma\Pi(s,k)} = \Pi\Sigma\Pi(s,k)$. Therefore, we will focus on understanding $\overline{\SPS(k)}$.

We also remark that if one could establish a strong debordering result such as $\overline{\SPS(s)} \subseteq \VP$, for $s = \poly(nd)$, then by known depth-reduction results~\cite{Gupta16}, it would follow that any polynomial $f \in \approxbar{\VP}$ can also be computed by a circuit of size $\exp(\sqrt{s} \cdot \log s)$. In this context, understanding the structure of $\approxbar{\SPS}$ becomes both significant and intriguing.

\paragraph{Universality of $\overline{\SPS(2)}$.}~Surprisingly, Kumar~\cite{kum20} showed that $\overline{\SPS(2)}$ is {\em universal}: For any $n$-variate $d$-degree polynomial $f \in \C[\x]$, there exists a $D$, depending on $n$ and $d$, such that $f \in \overline{\SPS(2,n,D)}$. The proof also works for nonhomogeneous polynomials, but for simplicity, we assume $f$ to be homogeneous. We present a proof sketch of this fact via defining \emph{Kumar complexity} (implicitly defined in~\cite{kum20}, and explicitly in~\cite{dutta2025geometric}). 

The Kumar complexity of $f$, denoted $\Kc(f)$, is the {\em smallest}
$s$ such that there exist a constant $\alpha \in \C$ and homogeneous linear polynomials $\ell_1, \cdots, \ell_s$ with the property that
\begin{equation}\label{eq:Kc}
f \;=\; \alpha\big(\prod_{i=1}^s (1+\ell_i)-1\big)\;.
\end{equation}
For instance, given a linear form $\ell$, we see that $\Kc(\ell^d)=d$, because $\ell^d= \prod_{j=1}^d (1+\omega^j \ell)-1$, where $\omega$ is a primitive $d$-th root of unity. However, not all polynomials have finite Kumar complexity: for example, it is easy to see that $x_1 \cdots x_n$ cannot be expressed as in \cref{eq:Kc}. The \emph{border Kumar complexity} of $f$, denoted $\underline{\Kc}(f)$, is the smallest $s$ such that  
\begin{equation}\label{eq:borderKc}
f \;=\; \lim_{\epsilon\to 0}\, \alpha(\eps) \cdot \big(\prod_{i=1}^s (1+\ell_i(\eps))-1\big),   
\end{equation}
for $\alpha(\eps) \in \C[\eps^{\pm 1}]$, and linear forms $\ell_i \in \C[\eps^{\pm 1}][\x]_1$. Let~$\val_{\eps}(\alpha)= M$. Then, one can assume that $\alpha = \gamma \cdot \eps^M$, for some $\gamma \in \C$. 

\medskip
\begin{observation}\label{obs:1}
$\underline{\Kc}(f) =s \implies f \in \overline{\SPS(2,n,s)}$    
\end{observation}
Implicitly, Kumar~\cite{kum20} showed the following (which was explicitly proved in~\cite{dutta2025geometric}).

\medskip
\begin{proposition}[\cite{kum20,dutta2025geometric}]
\label{pro:introkumar}
For all homogeneous $f$ we have
$\underline{\Kc}(f) \leq \deg(f) \cdot \bwr(f)$.
\end{proposition}
\begin{proof}
Let $\bwr(f) = r$ and let $\ell_1 \cdots \ell_r$ be linear forms depending rationally on $\eps$ such that 
$f \simeq \sum_{i=1}^r \ell_i^d$. 
Then one verifies that
\[
f \;\simeq\; -e_{d}(-\omega^0\ell_1,-\omega^1\ell_1,\ldots,-\omega^{d-1}\ell_1,\ldots \ldots, -\omega^0\ell_r,-\omega^1\ell_r,\ldots,-\omega^{d-1}\ell_r)
\]
and for all $0<i<d$ we have
\[
e_{i}(-\omega^0\ell_1,-\omega^1\ell_1,\ldots,-\omega^{d-1}\ell_1,\ldots \ldots, -\omega^0\ell_r,-\omega^1\ell_1,\ldots,-\omega^{d-1}\ell_r) \;=\; 0.
\]
Choose $N$ large enough so that for all $d < i \leq dr$ we have that 
\[\eps^{-Nd} \cdot e_i(-\eps^N\omega^0\ell_1,\ldots, -\eps^N\omega^{d-1}\ell_r) \;\simeq\; 0\;.\]
We obtain
$f \;\simeq\; -\eps^{-Nd}\big(\big((1-\eps^N\omega^0\ell_1) \cdots (1-\eps^N\omega^{d-1}\ell_r)\big)-1\big)$. Therefore
$\underline{\Kc}(f) \leq rd$.
\end{proof}
Since, for any homogeneous $f \in \C[\x]$, $\bwr(f)$ is finite, \cref{pro:introkumar} proves the universality of $\overline{\SPS(2)}$ circuits.  

\paragraph{More on Border Kumar complexity.}~Assume $\deg(f)=d$. If $f = \ell_1\cdots\ell_d$ is a product of homogeneous linear forms $\ell_i$, then $\underline{\Kc}(f) = d$, since 
$f \simeq \eps^d\big(\big(\prod_{i=1}^d (1+\eps^{-1}\ell_i)\big)-1\big)$. Interestingly, \cite{dutta2025geometric} showed a converse theorem to \cite{kum20}, that either~$\bwr(f) \le \underline{\Kc}(f)$, or $f$ is a product of linear forms. More formally, they showed the following.

\medskip
\begin{proposition}[{\cite[Theorem 2.7]{dutta2025geometric}}] \label{prop:mainthmmega}
If $f$ is not a product of linear forms, then $\bwr(f) \le \underline{\Kc}(f)$.     
\end{proposition}
\begin{proof}[Proof sketch]
We quickly sketch the proof of the above proposition. In~\cref{eq:borderKc}, if $\alpha = \gamma \cdot \eps^M$, for some $M \ge 1$, then one can show that  $f  \simeq\gamma \eps^{M}\prod_{i=1}^s(1+\ell_i)$ implying $f$ must be a product of linear forms. 

If $M=0$, then $f \simeq \gamma \big(\prod_{i=1}^s(1+\ell_i)-1\big)$. One can verify that if even one of the $\ell_i$ diverges (i.e.~it has $1/\eps$ term), then the $j$-th homogeneous part of $f_\eps$ diverges, where $j$ is the number of diverging $\ell_i$.
Hence all $\ell_i$ converge, and we can set $\eps$ to zero.
Now, since $f$ is homogeneous, each homogeneous degree $i$ part of $f_\eps$ vanishes, $i <d$.
In other words, $e_i(\ELL)=0$ for all $1 \leq i<d$, where $\ELL=(\ell_1,\ldots,\ell_s)$. 
Therefore, the Newton identity (see~\cref{sec:prelimI}): $p_d=(-1)^{d-1}\cdot d \cdot e_d + \sum_{i=1}^{d-1}(-1)^{d+i-1} e_{d-i}\cdot p_i$ gives that
$e_d(\ELL)$ and $p_d(\ELL)$ are same up to multiplication by a scalar.
Hence
$\WR(f)\leq s$. 

If $M < 0$, then one can deduce that for each $i$ we have $\ell_i=\epsilon\ell'_i$ with $\ell'_i\in\C[\eps][\x]_1$. Let $f_{\eps,j}$ denote the homogeneous degree $j$ part of $f_\epsilon := \gamma \eps^{M}\prod_{i=1}^m(1+\epsilon\cdot \ell'_i)$, where by assumption $f \simeq f_{\eps}$. Since $f$ is homogeneous of degree~$d$, for $0 \leq j < d$ we have $f_{\eps,j}\simeq 0$. 

By expanding the product, observe that for all $0 < j < d$ we have 
\[
0 \;\simeq\; f_{\eps,j} \;=\; \gamma\eps^{M}e_j(\eps\ell'_1,\ldots,\eps\ell'_m) \;=\; \gamma\eps^{M+j}e_j(\ell'_1,\ldots,\ell'_m)\;.
\]
By induction, using the Newton's identities (\cref{sec:prelimI}), one can show that for all $1 \leq j < d$, we have $ \eps^{M+j}p_j(\ELL') \simeq 0$, where $\ELL':= (\ell'_1, \cdots, \ell'_m)$. We can use Newton's identities again in the same way
to conclude that $\eps^{M+d}p_{d}(\ELL')\simeq (-1)^{d-1} \cdot d \cdot \eps^{M+d}e_{d}(\ELL')$:
\[
\eps^{M+d}p_{d}(\ELL') 
\;=\;(-1)^{d-1} \cdot d \cdot \eps^{M+d}e_{d}(\ELL') + \sum_{i=1}^{d-1} (-1)^{d-1+i} \underbrace{\eps^{M+d-i} e_{d-i}(\ELL)}_{\simeq 0}\cdot  \underbrace{\eps^{-M}}_{\simeq 0} \cdot
\underbrace{\eps^{M+i} p_i(\ELL')}_{\simeq 0}.
\]
We are done now, since 
\[
f \;\simeq\; f_{\eps,d} \;=\; \gamma\eps^{M+d} e_d(\ELL') \;\simeq\; \gamma\eps^{M+d} \cdot \frac 1 d \cdot (-1)^{d-1} p_d(\ELL')\;,
\]
and
hence $\underline{\WR}(f)\leq s$. This finishes the proposition.
\end{proof}

\subsubsection{Debordering \texorpdfstring{$\overline{\SPS(2)}$}{}}
By definition, $\underline{\Kc}(f) =s \implies f \in \overline{\SPS(2,n,s)}$, and by the discussion above, it seems that understanding $\bwr(f)$ is `almost' good enough to understand $\underline{\Kc}(f)$. But is it sufficient to understand $\overline{\SPS(2,n,s)}$? Unfortunately, the answer is no due to the following.

Consider the polynomial~$f:= x_1 \cdots x_d + y_1 \cdots y_d$. Trivially, $f \in \overline{\SPS(2,2d,d)}$. On the other hand, using simple partial derivatives, one can show that $\bwr(f) = \exp(d)$. By~\cref{prop:mainthmmega}, $\underline{\Kc}(f) \ge \bwr(f) = \exp(d)$. 

Therefore, we return to the following question:

\begin{center}
 {\em How powerful are~~$\overline{\SPS(2)}$ circuits}?~~\footnote{$\SPS(2)$ circuits are not universal: the polynomial $x_1x_2 + x_3x_4 + x_5x_6$ cannot be expressed in this model} 
\end{center}

In~\cite{dutta2022}, Dutta, Dwivedi, and Saxena proved that these circuits are not very powerful, by showing $\overline{\SPS(k)} \subseteq \mathsf{VBP}$, for any constant $k \ge 2$. Formally, they showed the following.

\medskip
\begin{theorem}[Debordering bounded depth-3 circuits~{\cite{dutta2022}}] \label{thm:k=2}
If an $n$-variate $d$-degree polynomial $f$ can be approximated by a $\SPS(k)$ circuit of size $s$, then it can be computed by an ABP of size $(snd)^{\exp(k)}$.
\end{theorem}
The proof is quite complicated and uses a technique called \textsf{DiDIL}. We will sketch a detailed proof for $k=2$ in this section, and how to generalize it to general $k$ in~\cref{sec:proof-general-k}.

\begin{proof}[Proof sketch of~\cref{thm:k=2}]
Let us fix the basic notation:
\begin{equation}
g \;:=\; T_1 \;+\;T_2 \;=\; f + \epsilon \cdot S\;,
\end{equation} 
where the polynomials~$T_1, T_2 \in \C[\epsilon^{\pm 1}][\x]$,  and each of them is a product of linear polynomials $\Pi\Sigma$ over $\C[\epsilon^{\pm 1}]$, and $S \in \C[\epsilon][\x]$. Suppose, $\val_{\epsilon}(T_i) = - a_i$, i.e.~$T_i = \epsilon^{-a_i} \cdot \ell_{i,1} \cdots \ell_{i,s}$, where $a_i \in \Z$, and each $\ell_{i,j} \in \C[\epsilon][\x]$ are linear polynomials (in $\x$) such that each $\ell_{i,j,0} := \ell_{i,j} \rvert_{\epsilon =0}$ is nonzero. 

One can assume that $a_1 = a_2 > 0$: If one of them is $\le 0$, then for the limit to exist, each $a_i$ has to be nonpositive, implying $f \in \SPS(2) \subseteq \VBP$. And, if $a_1, a_2 > 0$, but $a_1 \ne a_2$, then clearly $\val_{\epsilon} (T_1 + T_2) = \min (-a_1, -a_2) < 0$, a contradiction. Therefore, we proceed with $a:= a_1 = a_2 > 0$. 

Let us define a homomorphism $\Phi$ as follows: 
\begin{equation}\label{eq:phi-map}
 \Phi: \C[\epsilon^{\pm 1}][\x] \to \C[\epsilon^{\pm 1}][\x, z]\,,\;\;\text{such that}\;\;x_i \mapsto z \cdot x_i + \alpha_i\;,   
\end{equation}
where $\alpha_i$ are {\em randomly chosen} from $\C$. Essentially, $\a$ ensures that $\ell_{i,j,0}(\a) \neq 0$, for $i \in [2], j \in [s]$. We will argue that $\Phi(f)$ has a $\poly(snd)$-size ABP, which would imply the same for $f$.

Let~$\Phi(T_i) = \epsilon^{-a} \cdot \tilde{T}_i$, where $\tilde{T}_i := \Phi(\ell_{i,1}) \cdots \Phi(\ell_{i,s}) \in \C[\eps][\x]$. Dividing both sides by $\tilde{T}_2$ and subsequently differentiating with respect to $z$, we get  
\begin{equation*}
 \Phi(f)/\tilde{T}_{2} \,+\, \epsilon \cdot \Phi(S)/\tilde{T}_{2} \,=\, \epsilon^{-a} \,+ \,\Phi(T_{1})/\tilde{T}_{2}    
\end{equation*}
\begin{equation}
\implies
\partial_z\,\left(\Phi(f)/\tilde{T}_2\right) + \epsilon \cdot \partial_z\,\left(\Phi(S)/\tilde{T}_2\right) \,=\, \partial_z\, (\Phi(T_1)/\tilde{T}_2) 
\label{eq:didi-step-k2}    
\end{equation}

This has reduced the number of summands on the right-hand side to $1$, unfortunately, the right-hand surviving summand has become more complicated now. Further, it seems that we have no control over the coefficient structure of the $\epsilon^0$-term. 

Let $\cf_{\epsilon^0}(\tilde{T}_i) =: t_i$. Observe that $t_i \in \C[\x,z]$ is a polynomial that is a product of linear polynomials, in particular, by simple interpolation, one can deduce that each $t_{i,j}$ can be computed by a $\poly(sn)$-size ABP, where $t_i := \sum_{j=0}^s t_{i,j} z^i$. Further, $t_i\rvert_{z=0}$ is a nonzero constant, which is ensured by the choice of $\a$. Hence, it is not hard to conclude that 
\begin{equation}\label{def:f1}
f_1 \;:=\; \partial_z\, (\Phi(T_1)/\tilde{T}_2)   \simeq \partial_{z}(\Phi(f)/t_{2})\;.
\end{equation}
Moreover, $f_1 \in \F(\x)[[z]]$. This also establishes that $\val_{\eps}(\partial_z\, (\Phi(T_1)/\tilde{T}_2))=0$.  Here is an important claim.

\medskip
\begin{claim}\label{claim:k=2-debordering}
Each~$\cf_{z^i}(f_1)$, for $0 \le i < d$, can be computed by a ratio of two $\poly(snd)$-size ABPs.     
\end{claim}

Let us first argue why \autoref{claim:k=2-debordering} is sufficient to prove that $\Phi(f)$ can be computed by a polynomial-size ABP. Let us assume that $f_1 = \sum_{i \ge 0} C_i z^i$, where $C_i \in \C(\x)$, and by~\autoref{claim:k=2-debordering}, each $C_i$, for $0 \le i < d$ can be computed by a ratio of two polynomial-size ABPs. Then,  by definite integration, we have
\begin{equation} \label{eq:pfidea-thm1-finaleq}
   \Phi(f)/t_2 - \left(\Phi(f)/t_2\right)\rvert_{z=0} \;=\; \sum_{i \ge 1} (C_i/i) \cdot z^i\;. 
\end{equation}
What is $\Phi(f)/t_2 \rvert_{z=0}$? As $\Phi(f)/t_2 \in \F(\x)[[z]]$, clearly~$\Phi(f)/t_2\,\rvert_{z=0} \in \F(\x)$. But in fact, by assumption $\Phi(T_1)$ and $\tilde{T}_2$, evaluated at $z=0$ are non-zero elements in $\C[\eps^{\pm 1}]$. Considering the $\eps^0$--term in~\cref{eq:didi-step-k2}, we get:
\begin{equation} \label{eq:pfidea-thm1-z1=0-lim}
    \Phi(f)/t_2 \,\rvert_{z=0} \;\simeq\;\left(\Phi(T_1)/\tilde{T}_2\,\rvert_{z=0} + \epsilon^{-a}\right) \;\simeq\; c\;.
\end{equation}
for some $c \in \C$. Therefore, \cref{eq:pfidea-thm1-finaleq} gives us that $\Phi(f) = \left(\sum_{i \ge 0} C'_i \cdot z^i\right) \cdot t_2$, where $C_i':= C_i/i$, for $i \ge 1$ and $C_0 = c$. In particular,
\[
\cf_{z^r}(\Phi(f))\;=\; \left(\sum_{i \ge 0} C'_i \cdot z^i\right) \cdot \left(\sum_{j \ge 0} t_{2,j} \cdot z^j\right)\;=\;\sum_{i+j=r} C'_i \cdot t_{2,j}\;.
\]
Since both the addition and multiplication of two ABPs incur only an additive blow-up in the size, clearly $\cf_{z^r}(\Phi(f))$, for each $0 \le r \le d$ can be written as a ratio of two $\poly(snd)$-size ABPs. However, $\Phi(f) \in \C[\x,z]$, implying $\cf_{z^r}(\Phi(f)) \in \C[\x]$. Therefore, one can use the standard division elimination trick by Strassen~\cite{strassen1973vermeidung}, to conclude that each coefficient can be computed by a $\poly(snd)$-size ABP. 

This concludes that $\Phi(f)$, as well as $f$ can be computed by a $\poly(snd)$-size ABP. Therefore, from now on, we will only focus on proving~\autoref{claim:k=2-debordering}.

\begin{proof}[Proof sketch of~\autoref{claim:k=2-debordering}]

As argued above, we want to understand the expression $\pderiv{ \Phi(T_1) / \tilde{T}_2}{z}$. Here, we use \emph{logarithmic derivative}, i.e.~the $\dlog$ operator which has many useful properties; see~\cref{sec:prelimI}. Recall the notations:  $\Phi(T_i) = \epsilon^{-a_i} \cdot \tilde{T}_i$, where $\tilde{T}_i := \Phi(\ell_{i,1}) \cdots \Phi(\ell_{i,s})$. Assume that~$\Phi(\ell_{i,j}) = c_{i,j} + z \cdot \tilde{\ell}_{i,j}$, for some linear form $\tilde{\ell}_{i,j} \in \C[\epsilon][\x]_1$. Then, the expression $\pderiv{ \Phi(T_1) / \tilde{T}_2}{z}$ can be re-written as
\begin{equation*}
 \pderiv{ \Phi(T_1) / \tilde{T}_2}{z}\;  =\; \epsilon^{-a} \cdot \partial_{z}(\tilde{T}_1/\tilde{T}_2)    
\end{equation*}
\begin{equation}\label{eq:crucial-k-2}
\implies \epsilon^{-a} \cdot (\tilde{T}_1/\tilde{T}_2) \cdot \dlog \left(\tilde{T}_1/\tilde{T}_2\right) \\
    \;  =\;  \epsilon^{-a} \cdot \left(\tilde{T}_1/\tilde{T}_2\right) \cdot \left(\dlog(\tilde{T}_1) - \dlog (\tilde{T}_2)\right)\;.     
\end{equation}


Since the $\dlog$ operator distributes the product terms (see~\cref{sec:prelimI}), by the discussion in~\cref{sec:prelimI} and \cref{eq:dlog-linear-mod}, we get that
\begin{displaymath}\label{dlog-ses}
\dlog (\tilde{T}_i) \;=\; \sum_{j=1}^s\, \dlog(\Phi(\ell_{i,j}))\;=\;\sum_{j=1}^s\,\left(\frac{\tilde{\ell}_{i,j}}{c_{i,j} + z \cdot \tilde{\ell}_{i,j}}\right) \,=\, \sum_{j \ge 0} P_{i,j}(\x,\epsilon) \cdot z^j\;.
\end{displaymath}
In the above expression, each $P_{i,j}$ can be computed by a $\SES$ circuit of size $O(snj)$, over $\C(\eps)$. Define $Q_j := \epsilon^{-a} \cdot (P_{1,j} - P_{2,j})$. Similarly, each $Q_j$ can be computed by a $\SES$ circuit of size $O(snj)$, over $\C[[\epsilon]]$. Further, by definition, $\val_{\epsilon}(\tilde{T}_1/\tilde{T}_2)=0$, and  $\tilde{T_1}/\tilde{T}_2 \simeq t_1/t_2$. Therefore, looking at~\cref{eq:crucial-k-2}, we get 
\begin{equation}\label{eq:blah-2}
\cf_{\epsilon^0}\left(\pderiv{\Phi(T_1) / \tilde{T}_2}{z}\right)\; =\;\cf_{\epsilon^0}\left(\left(\tilde{T}_1/\tilde{T}_2\right) \cdot \left(\sum_{j \ge 0} Q_j z^j \right)\right)    = (t_1/t_2) \cdot \left(\cf_{\epsilon^0} \left( \sum_{j \ge 0} Q_j z^j \right)\right)\;.   
\end{equation}
We have already argued that $\cf_{\epsilon^0}(\tilde{T}_1/\tilde{T}_2) = t_{1}/t_2$, where $t_i = \sum_{j=0}^s t_{i,j}z^j$, and each $t_{i,j} \in \C[\x]$ can be computed by a $\poly(sn)$-size ABP. 

\cref{eq:blah-2} shows that $\val_{\eps}(Q_j) \ge 0$, and further $Q_{j,0}:= \cf_{\eps}(Q_j)$ can be computed by a $\poly(snj)$-size $\overline{\SES}$. This further implies that each of them can be computed by a $\poly(snj)$-size ABP (see~\cref{bordervw-in-vbp}). Unfolding \cref{def:f1} and \cref{eq:blah-2}, we get
\begin{equation}\label{eq:main-abp-k=2}
f_1\;=\;\cf_{\epsilon^0}\left(\partial_z\left(\frac{\Phi(T_1)}{\tilde{T}_2}\right)\right)\; =\;\frac{\sum_{j=0}^s t_{1,j}z^j}{\sum_{j=0}^s t_{2,j} z^j} \cdot \left(\sum_{j \ge 0} Q_{j,0} z^j\right) = \sum_{i \ge 0} f_{1,i} z^i\;.
\end{equation}
Since each $t_{i,j}$ and $Q_{j,0}$ can be computed by polynomial-size ABPs, by simple power series expansion, we get that each $f_{1,j}$ can also be computed by a ratio of two polynomial-size ABPs, proving~\autoref{claim:k=2-debordering}, as desired.
\end{proof}
This finishes a detailed proof sketch of $\approxbar{\SPS(2)} \subseteq \VBP$.
\end{proof}
\begin{remark}
If one can improve the debordering for $\overline{\VW}$, and show that $\overline{\VW} \subseteq \mathsf{VF}$, then \autoref{claim:k=2-debordering} shows that $\cf_{z^i}(f_1)$ can be written as a ratio of two polynomial-size formulas, improving the current debordering result to $\overline{\SPS(2)} \subseteq \mathsf{VF}$. 

Further, from~\autoref{obs:1} and \cref{pro:introkumar}-\ref{prop:mainthmmega}, it is {\em necessary} that $\overline{\SPS(2)} \subseteq \mathsf{VF}$ implies $\overline{\VW} \subseteq \mathsf{VF}$. Therefore, we can conclude the following interesting phenomenon:
\[
\overline{\SPS(2)} \subseteq \mathsf{VF}\;\iff\;\overline{\VW} \subseteq \mathsf{VF}\;.\]
\end{remark}
\subsubsection{Debordering \texorpdfstring{$\approxbar{\SPS(k)}$.}{}} \label{sec:proof-general-k}

We build our argument by starting with the base case $k = 2$ and then extending it to all $k \ge 3$ using induction. But rather than working directly in this inductive framework, we introduce a more powerful (and convenient) model: a depth-5 circuit class called $ \GenSES(k,s) := \SPSSES{ [k] }$; they compute elements of the form 
\[
\sum_{i=1}^k (U_i/V_i) \cdot (P_i/Q_i)\,,
\]
where $U_i, V_i \in \Pi\Sigma$, and $P_i, Q_i \in \SES$, and the circuit (with division allowed) has size $s$. Of course, it trivially subsumes $\SPSfanin{ [k] }{}$.

\textbf{1. Apply $\Phi$ and repeat the divide-and-derive steps.}
We begin with a polynomial $f \in \overline{\SPS(k)}$. We apply the map $\Phi$ (see~\cref{eq:phi-map}), and then perform the divide-and-derive step - similar to what is done in \cref{eq:didi-step-k2} - a total of $k-1$ times. After these steps, we obtain a polynomial $f_{k-1}$ that can be expressed as a ratio of two algebraic branching programs (ABPs), each of polynomial-size; this is similar to~\autoref{claim:k=2-debordering}.

\textbf{2. Why this stays within our bloated model.}
In the base case $k = 2$, Equation~\textup{\ref{eq:main-abp-k=2}} tells us that 
\[
f_1 \;\simeq\; \left(\overline{\Pi\Sigma}/\overline{\Pi\Sigma}\right) \cdot \overline{\SES}\,
\]
where the $\Pi\Sigma$ terms correspond to some polynomials $\tilde{T}_i$. The crucial insight—highlighted in \cref{dlog-ses} is that the coefficients of $z^i$ in $\dlog(\Pi\Sigma)$ have polynomial-size $\SES$ representations over $\C[\epsilon^{\pm 1}]$. Since the same holds for $\dlog(\SES)$, where the coefficients can be written as a ratio of polynomial-size $\SES$ circuits, it follows that $\GenSES(k, s)$ is {\em closed} under the \textsf{DiDIL} process. This closure is key: it ensures that the entire transformation stays within a controlled model, allowing us to establish an upper bound on $\overline{\GenSES(k, \cdot)}$.

\textbf{3. Final step: substitute $z = 0$.}
In the case $k = 2$, we analyzed the size of $\Phi(f)$ by setting $z = 0$ and isolating the $\epsilon^0$-coefficient (see \cref{eq:pfidea-thm1-z1=0-lim}).
Doing the same for the general case yields  
\[
\overline{\GenSES(k,\cdot)}\rvert_{z=0} \;\simeq\; \sum_{i \in [k]}\, \overline{c_i \cdot (P_i/Q_i)}\rvert_{z=0} \;\simeq \; \overline{\SES}/\overline{\SES}\;.\]
In the above, where each $c_i \in \mathbb{C}(\epsilon)$. This is because $\Pi\Sigma\rvert_{z = 0}$ lies in $\C[\epsilon^{\pm 1}]$, by the way $\Phi$ is defined. This structure is preserved across all inductive steps, so that $(\Pi\Sigma)/(\Pi\Sigma)\rvert_{z = 0} \in \C(\epsilon)$. Moreover, since $\approxbar{\SES}$ is {\em closed} under both addition and multiplication, the overall expression remains in the form of an $\approxbar{\SES}/\approxbar{\SES}$ circuit, with only a multiplicative blow-up in size.

\textbf{4. Wrapping up via interpolation.}
Finally, since $\approxbar{\SES} \subseteq \VBP$, by~\cref{bordervw-in-vbp}, we can use the same interpolation-based argument as in the base case (see \autoref{claim:k=2-debordering}) to complete the proof for $k \ge 3$. Since, each step incurs a multiplicative blowup in size, the final size becomes $s^{\exp(k)}$, i.e.~the proof yields polynomial-size upper bound when $k$ is constant, yielding $\approxbar{\SPS(k)} \subseteq \VBP$. For more details, we refer to~\cite{dutta2022,dutta2022tale}.

The following question remains open.

\medskip
\begin{question}
Is $\overline{\SPS(\log \log n)} \subseteq \mathsf{VBP}$?
\end{question}

\subsubsection{Exponential-hierarchy for Border Bounded Depth-3 Circuits} 

As discussed above, \cref{thm:k=2} shows that $\overline{\SPS(k)} \subseteq \mathsf{VBP}$. How tight is the debordering result? In~\cite{dutta2022separated} Dutta and Saxena proved that any $\overline{\SPS(k)}$ circuit computing an $n \times n$ symbolic determinant requires $\exp(n)$ size. In fact, they proved a far stronger result.

\medskip
\begin{theorem}[{\cite[Theorem 2]{dutta2022separated}}] \label{thm:k=2lb}
For any $k \ge 2$, the generalized inner product polynomial $P_{k+1,d} := \sum_{i=1}^{k+1} \Pi_{j=1}^d x_{(i-1)d+j}$  requires $\exp(d)$-size $\overline{\SPS(k)}$-circuits.    
\end{theorem}
Since, $P_{k+1,d} \in \SPS(k+1)$, this shows an exponential gap between $\overline{\SPS(k+1)}$ and $\overline{\SPS(k)}$, for any $k \ge 2$. Below, we will sketch it for $k=2$. This also uses \textsf{DiDIL}, but in a more refined way. 

\begin{proof}[Proof sketch of \cref{thm:k=2lb}]
Suppose,
\begin{equation}\label{eq:main-fanin2}
g \;:=\; T_1 \;+\;T_2 \;=\; P_{3,d} + \epsilon \cdot S\;,
\end{equation} 
where the polynomials~$T_1, T_2 \in \F[\epsilon^{\pm 1}][\x]$,  and each of them is a product of linear polynomials $\Pi\Sigma$ and they have size at most $s$ over $\F[\epsilon^{\pm 1}]$, and $S \in \F[\epsilon][\x]$. Suppose, $T_i = \epsilon^{-a_i} \cdot \ell_{i,1} \cdots \ell_{i,s}$, where $a_i \in \Z_{\ge 0}$, and each $\ell_{i,j} \in \F[\epsilon][\x]$ are linear polynomials (in $\x$) such that $\ell_{i,j} \rvert_{\epsilon =0} \ne 0$. Now, one of the three things can happen.
\begin{enumerate}
\setlength\itemsep{.1mm}
    \item[1](Easy case).~{\em Both} $T_i$ have at least one linear factor, say $\ell_{1,1}$ and $\ell_{2,1}$ whose $\epsilon$-free term is a homogeneous linear form over $\F$;
    \item[2](Intermediate case).~{\em Exactly} one of $T_i$, say wlog, $T_1$,  has at least one factor, say $\ell_{1,1}$ whose $\epsilon$-free term is a homogeneous  linear form;
    \item[3](Hard case).~{\em None} of the factors of $T_i$, has $\epsilon$-free term as a homogeneous linear form.
\end{enumerate}
The first two cases can be ruled out via direct arguments, while the third requires a more involved analysis, where we ultimately prove an exponential lower bound.

For the first case, we reduce modulo the ideal $\langle \ell_{1,1}, \ell_{2,1} \rangle$ in equation~\cref{eq:main-fanin2}. Note that under this reduction, $g \equiv 0$. It is not hard to argue that after this reduction, we obtain a relation of the form 
\[
P_{3,d}(\ell_1, \cdots, \ell_{3d}) \;=\; 0\;,
\]
for linear forms $\ell_1, \ldots, \ell_{3d}$ with $\rank(\ell_1, \ldots, \ell_{3d}) \in {3d-1, 3d-2}$. This is easily seen to be impossible, ruling out the first case. It is easy to show that this can never happen. 

In the second case, we reduce modulo the ideal $\langle \ell_{1,1} \rangle$. Then,
\[
T_2 \bmod \langle \ell_{1,1} \rangle = g \bmod \langle \ell_{1,1} \rangle = P_{3,d}(\ell_1, \cdots, \ell_n) + \epsilon \cdot S'\;,
\]
where $\rank(\ell_1, \cdots, \ell_n) =n-1$. The constant term (coefficient of $\epsilon^0$) on the left is a product of non-homogeneous linear forms, while on the right it is the homogeneous polynomial $P_{3,d}(\ell_1, \ldots, \ell_n)$, a contradiction.

In the third case, we introduce the notion of the {\em all-non-homogeneous} property: a term $T_i$ is said to satisfy this property if, for every linear form $\ell_{i,j}$ appearing in $T_i$, its constant-term projection $\ell_{i,j} \rvert_{\epsilon=0}$ is a nonzero non-homogeneous linear polynomial. When all $T_i$ in the expression for $g$ satisfy this, we say that $P_{3,d}$ is computed by an all-non-homogeneous $\overline{\SPS(2)}$ circuit.

This setting is more subtle and requires a technical analysis. In this case, we show that any such representation must have size at least $\exp(d)$, thus proving an exponential lower bound. 

The two primary claims leading to the lower bound for case III are as follows.
\begin{claim}\label{claim:main-claim-fanin}
If $P_{3,d}$ is computed by an all-non-homogeneous $\overline{\SPS(2)}$ circuit of size $s$, then $P_{3,d}$ can also be computed by a $\overline{\SES}$ circuit of size~$\poly(s)$.
\end{claim}

\begin{claim}\label{claim:lb-k=2}
If a $\overline{\SES}$ circuit computes the polynomial $P_{3,d}$, then its size must be at least $2^{\Omega(d)}$.   
\end{claim}
It is straightforward to see that \autoref{claim:main-claim-fanin}, together with \autoref{claim:lb-k=2}, implies the desired lower bound $s \ge 2^{\Omega(d)}$, thereby completing the proof for the case $k=2$.

The proof of \autoref{claim:lb-k=2} follows from standard arguments based on partial derivatives and is relatively straightforward. Hence, we focus our efforts on proving \autoref{claim:main-claim-fanin}.

Apply a simple variable-scaling map $\Phi: x_i \mapsto z \cdot x_i$, to~\cref{eq:main-fanin2}; note that this is simpler than \cref{eq:phi-map}, and {\em does not require} any shift, unlike the debordering proof. Note that $\Phi(P_{3,d}) = z^d \cdot P_{3,d}$, and $\Phi(T_i) = \epsilon^{-a_i} \cdot \Phi(\ell_{i,1}) \cdots \Phi(\ell_{i,s})$. Now, for $i \in [2]$, let $\tilde{T}_i := \Phi(\ell_{i,1}) \cdots \Phi(\ell_{i,s})$. Divide and derive like before to get:
\begin{equation}\label{eq:1}
\partial_z\, \left(\Phi(P_{3,d} + \epsilon \cdot S)/\tilde{T}_2\right) = \partial_z\, (\Phi(T_1)/\tilde{T}_2)\;.    
\end{equation}

Since we are in the third case (all-non-homogeneous), we know that $\ell_{i,j} = c_{i,j} + \tilde{\ell}_{i,j}$, where each $\tilde{\ell}_{i,j} \in \F[\epsilon][\x]$ is a homogeneous linear polynomial, further $c_{i,j}\rvert_{\epsilon = 0} \ne 0$. Trivially, $\Phi(\ell_{i,j})= c_{i,j} + z \cdot \tilde{\ell}_{i,j}$. In fact $1/\tilde{T}_2 = c + \epsilon \cdot R(\x,\epsilon,z)$, where $0 \ne c \in \C$, and $R \in \C[[\epsilon,z]][\x]$. Now, a simple calculation shows that 
\begin{equation}\label{eq:2}
\cf_{\epsilon^0z^{d-1}}\left(\partial_z\,\left(\Phi(P_{3,d} + \epsilon \cdot S)/\tilde{T}_2\right)\right) = P_{3,d}/c\;.    
\end{equation}

On the other hand, using the $\dlog$-trick, and power series expansion, we get that 
\begin{equation*}
 \cf_{\epsilon^0z^{d-1}}\left(\pderiv{\Phi(T_1) / \tilde{T}_2}{z}\right) =\cf_{z^{d-1}}\left(\cf_{\epsilon^0}\left(\tilde{T}_1/\tilde{T}_2\right)\right) \cdot \left(\cf_{\epsilon^0} \left( \sum_{j \ge 0} Q_j z^j \right)\right)\;.   
\end{equation*}

In the above, $Q_j := \epsilon^{-a_1} \cdot (P_{1,j} - P_{2,j})$,  where as argued in the debordering proof (after \cref{eq:crucial-k-2}),  $\dlog (\tilde{T}_i) = \sum_{j \ge 0} P_{i,j}(\x,\epsilon) \cdot z^j\;$, and  each $Q_j$ has a $\SES$ expression of size $O(snj)$, over $\C[[\epsilon]]$. 

In fact, the above shows that the minimum $z$ power in the term $\epsilon^0$ in $\pderiv{\Phi(T_1) / \tilde{T}_2}{z}$ is $d-1$. Further it is easy to check that $\val_{z}\left(\cf_{\epsilon^0}(\tilde{T}_1/\tilde{T}_2)\right)=0$. Hence, looking at~\cref{eq:blah-2}, it must happen that $\val_{\epsilon}(Q_{j}) \ge 1$, for all $0 \le j \le d-2$. Therefore, there exists  some constant $c" \in \C$ such that 
\[
\cf_{\epsilon^0z^{d-1}}\left(\pderiv{\Phi(T_1) / \tilde{T}_2}{z}\right)\;=\; \left(\frac{\prod_{j=1}^s c_{1,j,0}}{\prod_{j=1}^s c_{2,j,0}}\right) \cdot \cf_{\epsilon^0} \left(Q_{d-1}\right)\;=\;c" \cdot \cf_{\epsilon^0}(Q_{d-1})\;.\]
Since we already argued that $Q_{d-1}$ can be computed by a $\SES$ circuit of size $O(snd)$, over $\C[[\epsilon]]$, we conclude that $\cf_{\epsilon^0z^{d-1}}\left(\pderiv{\Phi(T_1) / \tilde{T}_2}{z}\right)$ can be computed by a $\overline{\SES}$ circuit of size $O(snd)$. Combining equations~\cref{eq:1}–\ref{eq:2} with the final observation above, we obtain \autoref{claim:main-claim-fanin}. This also finishes the detailed sketch of $k=2$.
\end{proof}

A similar argument extends to the case $k \ge 3$; we refer the reader to~\cite{dutta2022separated} for a detailed treatment. Due to a multiplicative blow-up in size at each step of the reduction, this approach yields a lower bound of the form $\exp(d^{1/\exp(k)})$. Thus, we obtain an exponential lower bound as long as $k$ is constant. We conclude this section with the following open question.

\medskip
\begin{question}
Prove an exponential lower bound for $\overline{\SPS(o(n))}$ circuits. 
\end{question}

\subsection{Border Complexity of Symbolic Determinant under Rank One Restriction}

Symbolic determinant is known to be a complete polynomial for $\VBP$. More precisely, for any $f \in \C[\x]$, there exist some $m$ and $m \times m$ matrices $A_0, \cdots, A_n$ such that $f = \det(A_0 + \sum_{i \in [n]} A_i x_i)$. The class of our interest is the symbolic
determinant under rank one restriction: Consider the class of polynomials of form where
for each $1 \le i \le n$, the rank of $A_i$ is $k$, for some parameter $k$. One can define the class $\VBP_{[k]}$ based on such restrictions:
\begin{equation}\label{eq:vbp-1}
    \VBP_{[k]} \;:=\; \{ (f_n)_n \mid f_n = \det(A_0 + \sum_{i=1}^n A_i x_i), A_i \in \C^{\poly(n) \times \poly(n)}, \rank(A_i) =k\}\;.
\end{equation}
The class~$\VBP_{[1]}$  has been studied extensively in contexts of polynomial identity testing, combinatorial optimization, and matrix completion (see, for example, \cite{edmonds1979matroid,lovasz1989singular,murota1993mixed}). This class admits a deterministic polynomial-time identity testing algorithm in the white-box setting~\cite{lovasz1989singular}, and a deterministic quasipolynomial-time algorithm in the black-box setting~\cite{gurjar2017linear}. It coincides with the class of polynomial families computed as the determinant of symbolic matrices in which each variable appears {\em at most once} - commonly referred to as read-once determinants. Surprisingly, Chatterjee, Ghosh, Gurjar and Raj~\cite{chatterjee2023border} showed that~$\overline{\VBP_{[1]}} = \VBP_{[1]}$. More formally, they proved the following.

\medskip
\begin{theorem}[{\cite{chatterjee2023border}}]\label{theorem:rank1abp}
Given $A_0, A_1, A_2, \ldots, A_n \in \C[\eps^{\pm 1}]^{r\times r}$ such that for each $1\leq i\leq n$, $\rank(A_i) = 1$ over $\C[\epsilon^{\pm 1}]$. Let $f \simeq \det(A_0 + \sum_{i=1}^n A_ix_i)$. Then, there exists $B_0, B_1, B_2, \ldots, B_n$ in  ${\C^{(n+r)\times (n+r)}}$
such that $f = \det(B_0 + \sum_{i=1}^n B_ix_i)$ and $\rank(B_i) = 1$ over $\F$ for each $i\in [n]$. 
\end{theorem}
We will discuss the geometric perspective as well as the proof idea of \cref{theorem:rank1abp} below. 

\paragraph*{An algebraic geometry perspective on \cref{theorem:rank1abp}.}
Consider the simpler case when $A_0=0$. Now, suppose $A_1, A_2, \dots, A_n$ are $m \times m$ matrices of rank $1$.
Let us write $A_i = \vect{u}^i \cdot {\vect{v}^i}^T$ for some vectors $\vect{u}^i, \vect{v}^i \in \C^{m}$ and define matrices $U, V \in \C^{m \times n}$ whose
$i$th columns are $\vect{u}^i$ and $\vect{v}^i$, respectively. 
%
%
It can be verified that  
\[\det\left(\sum_i A_ix_i\right) \;=\;
\sum_{S} \det(U_S) \det(V_S) \prod_{j \in S} x_j,\]
where the sum is over all 
size-$m$ subsets $S$ of $[n]$ and $U_S$ (or $V_S$) denotes the submatrix of $U$ (or $V$) obtained by taking columns with indices in the set $S$. 
The result of \cite{chatterjee2023border} shows that the image of the map 
\[(\C^{m \times n})^2 \to \C^{\binom{n}{m}}, 
\quad (U,V) \mapsto (\det(U_S)\times \det(V_S))_S\]
is {\em Zariski closed}. The following map is a closely related one and has been well-studied in algebraic geometry, which gives the Pl\"{u}cker coordinates of elements in the Grassmannian variety. 
\[\C^{m \times n} \to \C^{n \choose m}, 
t\quad  U \mapsto (\det(U_S))_S\;.\]
The image of this map is known to be a {\em closed set}. 
In other words, \cref{theorem:rank1abp} implies that the set obtained by taking coordinatewise products of pairs of points in the Grassmannian is closed. This property is quite special, as it does not hold for arbitrary varieties—indeed, there are simple examples where the coordinatewise product of pairs of points from a variety {\em fails} to be closed. 

%
%
%
%
For any $k\leq n$, let us define the map
$\phi_k: \C^{n^2}\xrightarrow{} \C^{\binom{n}{k}}$ as
\[\phi_k(A) = (\det(A_I))_{I \in \binom{[n]}{k}}\]
where $\binom{[n]}{k}$ is the set of all size-$k$ subsets of $[n]$.
\cref{theorem:rank1abp} also implies that the image of $\phi_k$ on $n\times n$ rank-$k$ matrices is closed.

\medskip
\begin{corollary}
For any $n>0$ and $k\leq n$, the image of the size $k$ principal minor map on $n\times n$ matrices with rank at most $k$ is closed in $\C^{\binom{n}{k}}$.
\end{corollary}
%

\begin{proof}[Proof idea of \cref{theorem:rank1abp}.]
We will prove the simpler case, i.e.~when $A_0=0$. As discussed above, the goal is to show that the image of the following map is closed under limits:
\[(U,V) \mapsto (\det(U_S)\times \det(V_S))_S\;,\]
where the product is taken over all size-$m$ subsets $S \subseteq [n]$.
To prove this, we start with two matrices $U,V \in \C[\epsilon^{\pm 1}]^{m \times n}$ and aim to construct approximating matrices $\widehat{U}, \widehat{V} \in \C^{m \times n}$ such that for each size-$m$ subset $S\subseteq [n]$, we have
\[ \left( \det(U_S) \det(V_S) \right) \;\simeq  \;\det(\widehat{U}_S) \det(\widehat{V}_S). \]

Naturally, this can only be expected when the limit exists for each such $S$. Note that one cannot simply 
apply the limit operation on the matrix entries. Further, clearly, $\lim_{\eps \to 0} f$
exists if and only if $\val(f) \geq 0$.
So, we can assume that~$\val(\det(U_S) \det(V_S)) \geq 0$ for every $S$. Equivalently,
\[\min_S \{ \val(\det(U_S) \det(V_S)) \} = \min_S \{ \val(\det(U_S)) + \val( \det(V_S)) \}  =0.\] 
Note that in the limit, only those subsets $S$ that achieve this minimum will contribute nonzero terms.

A challenge arises because the valuation function does not distribute over sums in the usual way: i.e.,
\[\min_S \{ \val(\det(U_S)) + \val( \det(V_S)) \} \ne \min_S \{ \val(\det(U_S))\} + \min_S\{ \val( \det(V_S)) \}.\]
Nonetheless, for the valuation function, the failure of distributivity is limited. This is due to a combinatorial property resembling an exchange axiom:  for any two distinct $S,T\subseteq[n]$ of size $m$ and any $j\in T\setminus S$, there exists a $k\in S\setminus T$ such that \[\val(\det(U_S))+\val(\det(U_T))\;\geq\; \val(\det(U_{S-k+j}))+\val(\det(U_{T-j+k}))\;.
\]
%
This submodularity-like behavior underlies the theory of valuated matroids, introduced by Dress and Wenzel~\cite{dress1990valuated}. Going further, Murota~\cite{murota1996valuated} established a splitting theorem for valuated matroids:
the minimum of a sum of two such functions can be decomposed as a pair of independent minima, corrected by a linear term. Specifically, there exists a vector $\vect z \in \Z^{n}$ such that
\[
\min_S \{ \val(\det(U_S)) + \val( \det(V_S))\}\;=\; \min_S \{ \val(\det(U_S)) + \sum_{i \in S} {\vect z}_i\} 
+ \min_S\{ \val( \det(V_S))- \sum_{i \in S} {\vect z}_i \}.
\]
This decomposition is powerful because the correction term is linear and hence easy to handle. The problem now separates into two independent ones, involving only $U$ and $V$ respectively.
That is, given any two  matrices $U,V \in \C[\epsilon^{\pm 1}]^{m \times n}$, construct 
matrices $\widehat{U}, \widehat{V} \in \C^{m \times n}$ such that for each size-$m$ subset $S\subseteq [n]$, we have
\[ \det(U_S)\;\simeq\;  \det(\widehat{U}_S)
\text{ and } 
 \det(V_S) \;\simeq\; \det(\widehat{V}_S)\;. \]
 The problem now becomes tractable essentially because the image of the map 
  $U \mapsto (\det(U_S))_S$ is known to be closed. Putting it all together, we obtain:
  \[\det(\sum_{i=1}^n A_ix_i) \,\simeq\,\;\sum_{\substack{S \subseteq [n] \\ |S|=m}} \;\det(U_S)\det(V_S) \x_S \,\simeq\, \sum_{\substack{S \subseteq [n] \\ |S|=m}}\; \det(U'_S) \det(V'_S) \x_S \,=\, \det(\sum_{i=1}^n B_i x_i)\;.
  \]
\end{proof}
We leave this section by asking the following open question.

\medskip
\begin{question}
Is $\overline{\VBP_{[2]}} = \VBP_{[2]}$?
\end{question}

\subsection{Debordering boder of Width-2 ABPs}
For any positive integer $k \in \N$, the class $\VPk{k}$ contains the families of polynomials computable by width-$k$ ABPs of polynomially bounded size. 
Ben-Or and Cleve \cite{cleve1988computing} showed that $\VPk{k}=\VPe$ for all $k \geq 3$.
Later, Allender and Wang \cite{allender2016power} showed that width-2 ABPs cannot compute even simple polynomials such as $x_1x_2 + \cdots + x_{15}x_{16}$, so in particular we have a strict inclusion $\VPk{2} \subsetneq \VPk{3}$. Surprisingly, Bringmann, Ikenmeyer and Zuiddam~\cite{bringmann2018algebraic} showed the following.

\medskip
\begin{theorem}[{\cite{bringmann2018algebraic}}]\label{eq:vp2bar}
$\approxbar{\VPk{2}}\;=\;\approxbar{\VPe}$.
\end{theorem}
This result holds over any field of characteristic $\ne 2$. Interestingly, as a direct corollary of \cref{eq:vp2bar} and the result of Allender and Wang, the inclusion $\VPk{2} \subsetneq \overline{\VPk{2}}$ is {\em strict}.

\paragraph{The characteristic issue.}~The proof in~\cite{bringmann2018algebraic} used that the field characteristic is not $2$, since they used this simple identity: $x\cdot y = (\frac{x+y}{2})^2 - (\frac{x-y}{2})^2$. Therefore, it was left open whether $\approxbar{\VPk{2}}$ is even complete over $\F_2$. Later \cite{dutta2024power} proved the universality of border of width-2 ABPs even when $\text{char}(\F)=2$. Formally, they proved the following.

\medskip
\begin{theorem}[{\cite{dutta2024power}}]\label{thm:width-2-f2}
Any degree $d$ polynomial $f$, with the number of monomials $m$, can be approximated by $O(m2^d)$-size width-2 ABPs.  
\end{theorem}
This proof is {\em independent} of the characteristic of the field. Below, we will sketch both the proofs. 

\begin{proof}[Proof sketch of~\cref{eq:vp2bar}]
To facilitate understanding of the proofs and associated figures, recall that an algebraic branching program (ABP) naturally corresponds to an iterated matrix product, assuming a fixed numbering of the vertices in each layer (starting from 1). Specifically, for any two consecutive layers $i$ and $i+1$, define a matrix $M_i$ whose $(v, w)$-entry is the label of the edge from vertex $v$ in layer $i$ to vertex $w$ in layer $i+1$, or zero if no such edge exists. Then, the value computed by the ABP equals the matrix product $M_s \cdots M_2 M_1$. 

Additionally, for a polynomial $f \in \C[\eps^{\pm 1}][\x]$, define the matrix
\[
Q(f) \;:=\; \begin{pmatrix}
f & 1\\
1 & 0
\end{pmatrix}\;.
\]
A {\em primitive Q-matrix} is any matrix~$Q(\ell)$, where $\ell$ is a linear form over $\C[\eps^{\pm 1}]$.
For a $2\times 2$ matrix $M$ with entries in $\C[\eps^{\pm 1}][\x]$,  we use the shorthand notation $M + \Oh(\varepsilon^k)$ for
$M + \Big(\begin{smallmatrix}\Oh(\varepsilon^k)&\Oh(\varepsilon^k)\\\Oh(\varepsilon^k)&\Oh(\varepsilon^k)\end{smallmatrix}\Big)$,
where $\Oh(\varepsilon^k)$ denotes the set $\eps^k\, \FF[\eps, \x]$.
As a product of matrices, the ABP construction in the proof will be of the form $(\begin{smallmatrix}1 & 0\end{smallmatrix}) M_s \cdots M_2 M_1 (\begin{smallmatrix}1 \\ 0\end{smallmatrix})$ where the $M_i$ are primitive Q-matrices $Q(f)$ for which $f$ is either a constant from $\C[\varepsilon^{\pm 1}]$ or a variable. 

\paragraph{The measure $\mu_k$.}~For $k \in \N$, let us define the measure~$\mu_k(f) :=(s,t)$, where $s$ is the number of matrices such that $Q(f) + \Oh(\eps^k)$ can be written as $(\begin{smallmatrix}1 & 0\end{smallmatrix}) M_s \cdots M_2 M_1 (\begin{smallmatrix}1 \\ 0\end{smallmatrix})$, and $t$ is the highest $\eps$-error-degree. 

The following two claims constitute the main technical contributions and are sufficient to establish the main result. Their proofs are illustrated (primarily) through the figures presented below.

\medskip
\begin{claim}\label{cl:addition}
Given $f, g \in \C[\x]$, such that $\mu_k(f) = (s_1,t_1), \mu_k(g) = (s_2, t_2)$, we have $\mu_k(f+g) =  (s_1 + s_2 +1, t_1 + t_2)$
\end{claim}
\begin{claim}\label{cl:squaring}
Given $f \in \C[\x]$ such that $\mu_3(f) = (s,t)$, we have $\mu_1(\pm f^2) = (O(s), O(t))$.
\end{claim}
\begin{proof}[Proof sketch of~\autoref{cl:addition}]
Follows from the identity (also see~\cref{fig:addition}): 
\[(Q(f) + \Oh(\eps^k)) \cdot Q(0) \cdot (Q(g) + \Oh(\eps^k)) \;=\; Q(f+g) + \Oh(\eps^k)\;.\]  
\end{proof}
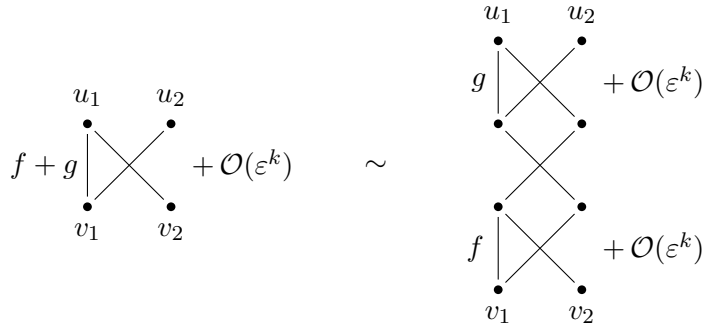
\begin{figure}[h]
\centering
\begin{minipage}{7em}
\begin {tikzpicture}
\node[state] (11) [label={$u_1$}]{};
\node[state] (12) [right=of 11, label={$u_2$}] {};
\node[state] (21) [below=of 11, label={-90:$v_1$}] {};
\node[state] (22) [below=of 12, label={-90:$v_2$}] {};
\path (11) edge (22);
\path (12) edge (21);
\path (11) edge node[left, xshift=-1pt] {$f + g$} (21);
\path (12) edge[draw=none] node[right] {$\hspace{0.5em}+\, \Oh(\eps^k)$} (22);
\end{tikzpicture}
\end{minipage}
\hspace{5em}$\sim$\hspace{2em}
\begin{minipage}{10em}
\begin {tikzpicture}
\node[state] (11) [label={$u_1$}]{};
\node[state] (12) [right=of 11, label={$u_2$}] {};
\node[state] (21) [below=of 11] {};
\node[state] (22) [below=of 12] {};
\node[state] (31) [below=of 21] {};
\node[state] (32) [below=of 22] {};
\node[state] (41) [below=of 31, label={-90:$v_1$}] {};
\node[state] (42) [below=of 32, label={-90:$v_2$}] {};
\path (11) edge (22);
\path (12) edge (21);
\path (12) edge[draw=none] node[right] {$\hspace{0.5em}+\, \Oh(\eps^k)$} (22);
\path (11) edge node[left, xshift=-1pt] {$g$} (21);
\path (21) edge (32);
\path (22) edge (31);
\path (31) edge (42);
\path (32) edge (41);
\path (32) edge[draw=none] node[right] {$\hspace{0.5em}+\, \Oh(\eps^k)$} (42);
\path (31) edge node[left, xshift=-1pt] {$f$} (41);
\end{tikzpicture}
\end{minipage}
\caption{Addition construction for \autoref {cl:addition}}\label{fig:addition}
\end{figure}
\begin{proof}[Proof sketch of~\autoref{cl:squaring}]
Let us define matrices $A,B,C$ as follows:
\[
A: = Q(-\eps^{-1}) \cdot Q(\eps)\cdot Q(-\eps^{-1}), B: = Q(1) \cdot Q(-1)\cdot Q(1) \cdot Q(\eps^{2})\,,
\]
\[
C := Q(-\eps^{-1}) \cdot Q(\eps - 1)\cdot Q(1) \cdot Q(\eps^{-1} -1)\;.
\]
Then one can check that 
\[
A \cdot (Q(f) + \Oh(\eps^3)) \cdot B \cdot (Q(f) + \Oh(\eps^3)) \cdot C = Q(-f^2) + \Oh(\eps)\;.
\]
One can check~\cref{fig:squaring}-\ref{fig:squaringsub} for the pictorial construction.
\end{proof}
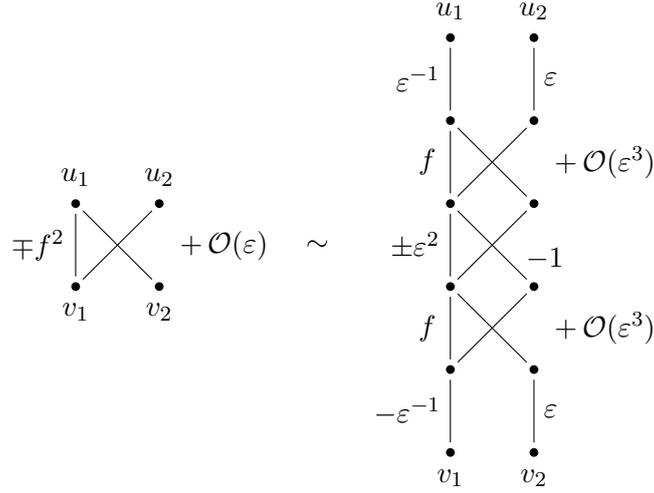
\begin{figure}[h!]
\centering
\begin{minipage}{7em}
\begin {tikzpicture}
\node[state] (11) [label={$u_1$}]{};
\node[state] (12) [right=of 11, label={$u_2$}] {};
\node[state] (21) [below=of 11, label={-90:$v_1$}] {};
\node[state] (22) [below=of 12, label={-90:$v_2$}] {};
\path (11) edge (22);
\path (12) edge (21);
\path (11) edge node[left, xshift=-1pt] {$\mp f^2$} (21);
\path (12) edge[draw=none] node[right] {$\hspace{0.5em}+\, \Oh(\eps)$} (22);
\end{tikzpicture}
\end{minipage}
\hspace{3em}$\sim$\hspace{1em}
\begin{minipage}{10em}
\begin {tikzpicture}
\node[state] (01) [label={$u_1$}]{};
\node[state] (02) [right=of 01, label={$u_2$}] {};
\node[state] (11) [below=of 01]{};
\node[state] (12) [right=of 11] {};
\node[state] (21) [below=of 11] {};
\node[state] (22) [below=of 12] {};
\node[state] (31) [below=of 21] {};
\node[state] (32) [below=of 22] {};
\node[state] (41) [below=of 31] {};
\node[state] (42) [below=of 32] {};
\node[state] (51) [below=of 41, label={-90:$v_1$}] {};
\node[state] (52) [below=of 42, label={-90:$v_2$}] {};
\path (01) edge node[left, xshift=-1pt] {$\eps^{-1}$} (11);
\path (02) edge node[right, xshift=1pt] {$\eps$} (12);
\path (11) edge (22);
\path (12) edge (21);
\path (12) edge[draw=none] node[right] {$\hspace{0.5em}+\, \Oh(\eps^3)$} (22);
\path (11) edge node[left, xshift=-1pt] {$f$} (21);
\path (21) edge node[right, xshift=10pt, yshift=-5pt] {$-1$} (32);
\path (22) edge (31);
\path (21) edge node[left, xshift=-1pt] {$\pm\eps^2$} (31);
\path (31) edge (42);
\path (32) edge (41);
\path (32) edge[draw=none] node[right] {$\hspace{0.5em}+\, \Oh(\eps^3)$} (42);
\path (31) edge node[left, xshift=-1pt] {$f$} (41);
\path (41) edge node[left, xshift=-1pt] {$-\eps^{-1}$} (51);
\path (42) edge node[right, xshift=1pt] {$\eps$} (52);
\end{tikzpicture}
\end{minipage}
\caption{Squaring construction for \autoref{cl:squaring}}\label{fig:squaring}
\end{figure}
\begin{figure}[h!]
\centering
\begin{minipage}{10em}
\begin {tikzpicture}
\node[state] (01) [label={$u_1$}]{};
\node[state] (02) [right=of 01, label={$u_2$}] {};
\node[state] (11) [below=of 01, label={-90:$v_1$}]{};
\node[state] (12) [right=of 11, label={-90:$v_2$}] {};
\node[state] (21) [below=of 11, yshift=-3em, label={$u_1$}] {};
\node[state] (22) [right=of 21, label={$u_2$}] {};
\node[state] (31) [below=of 21] {};
\node[state] (32) [below=of 22] {};
\node[state] (41) [below=of 31] {};
\node[state] (42) [below=of 32] {};
\node[state] (51) [below=of 41] {};
\node[state] (52) [below=of 42] {};
\node[state] (61) [below=of 51, label={-90:$v_1$}] {};
\node[state] (62) [below=of 52, label={-90:$v_2$}] {};
\path (01) edge node[left, xshift=-1pt] {$\eps^{-1}$} (11);
\path (02) edge node[right, xshift=1pt] {$\eps$} (12);
\path (12) edge[draw=none] node[rotate=90] {$\sim$} (21);
\path (21) edge (32);
\path (22) edge (31);
\path (21) edge node[left, xshift=-1pt] {$\eps^{-1} - 1$} (31);
\path (31) edge (42);
\path (32) edge (41);
\path (31) edge node[left, xshift=-1pt] {$1$} (41);
\path (41) edge (52);
\path (42) edge (51);
\path (41) edge node[left, xshift=-1pt] {$\eps-1$} (51);
\path (51) edge node[left, xshift=-1pt] {$-\eps^{-1}$} (61);
\path (51) edge (62);
\path (52) edge (61);
\end{tikzpicture}
\end{minipage}
\hspace{1em}
\begin{minipage}{10em}
\begin {tikzpicture}
\node[state] (01) [label={$u_1$}]{};
\node[state] (02) [right=of 01, label={$u_2$}] {};
\node[state] (11) [below=of 01, label={-90:$v_1$}]{};
\node[state] (12) [right=of 11, label={-90:$v_2$}] {};
\node[state] (21) [below=of 11, yshift=-3em, label={$u_1$}] {};
\node[state] (22) [right=of 21, label={$u_2$}] {};
\node[state] (31) [below=of 21] {};
\node[state] (32) [below=of 22] {};
\node[state] (41) [below=of 31] {};
\node[state] (42) [below=of 32] {};
\node[state] (51) [below=of 41] {};
\node[state] (52) [below=of 42] {};
\node[state] (61) [below=of 51, label={-90:$v_1$}] {};
\node[state] (62) [below=of 52, label={-90:$v_2$}] {};
\path (01) edge node[left, xshift=-1pt] {$\pm\eps^{2}$} (11);
\path (01) edge node[right, xshift=10pt, yshift=-5pt] {$-1$} (12);
\path (02) edge (11);
\path (12) edge[draw=none] node[rotate=90] {$\sim$} (21);
\path (21) edge (32);
\path (22) edge (31);
\path (21) edge node[left, xshift=-1pt] {$\pm\eps^2$} (31);
\path (31) edge (42);
\path (32) edge (41);
\path (31) edge node[left, xshift=-1pt] {$1$} (41);
\path (41) edge (52);
\path (42) edge (51);
\path (41) edge node[left, xshift=-1pt] {$-1$} (51);
\path (51) edge node[left, xshift=-1pt] {$1$} (61);
\path (51) edge (62);
\path (52) edge (61);
\end{tikzpicture}
\end{minipage}
\begin{minipage}{10em}
\begin {tikzpicture}
\node[state] (01) [label={$u_1$}]{};
\node[state] (02) [right=of 01, label={$u_2$}] {};
\node[state] (11) [below=of 01, label={-90:$v_1$}]{};
\node[state] (12) [right=of 11, label={-90:$v_2$}] {};
\node[state] (21) [below=of 11, yshift=-3em, label={$u_1$}] {};
\node[state] (22) [right=of 21, label={$u_2$}] {};
\node[state] (31) [below=of 21] {};
\node[state] (32) [below=of 22] {};
\node[state] (41) [below=of 31] {};
\node[state] (42) [below=of 32] {};
\node[state] (51) [below=of 41, label={-90:$v_1$}] {};
\node[state] (52) [below=of 42, label={-90:$v_2$}] {};
\path (01) edge node[left, xshift=-1pt] {$-\eps^{-1}$} (11);
\path (02) edge node[right, xshift=1pt] {$\eps$} (12);
\path (12) edge[draw=none] node[rotate=90] {$\sim$} (21);
\path (21) edge (32);
\path (22) edge (31);
\path (21) edge node[left, xshift=-1pt] {$-\eps^{-1}$} (31);
\path (31) edge (42);
\path (32) edge (41);
\path (31) edge node[left, xshift=-1pt] {$\eps$} (41);
\path (41) edge (52);
\path (42) edge (51);
\path (41) edge node[left, xshift=-1pt] {$-\eps^{-1}$} (51);
\end{tikzpicture}
\end{minipage}
\caption{Squaring construction subroutines for $C$, $B$, and $A$ for \autoref{cl:squaring}}\label{fig:squaringsub}
\end{figure}
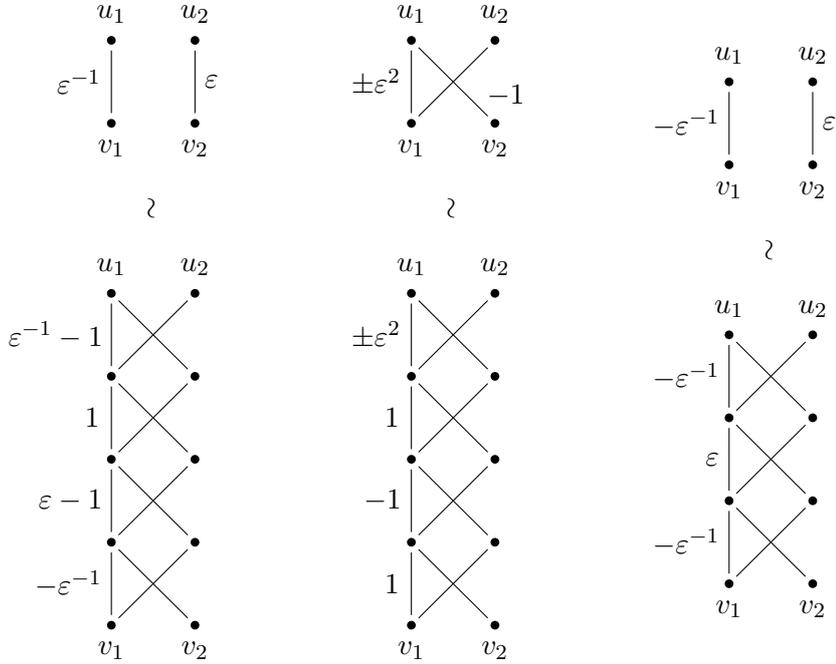
\paragraph{Wrapping up the proof.}~Using the identity: $f \cdot g = \left(\frac{f+g}{2}\right)^2 - \left(\frac{f-g}{2}\right)^2$, one can show that if $\mu_3(f) = (s_1,t_1)$ and $\mu_3(g) = (s_2,t_2)$, then 
\[
\mu_1(f\cdot g) \;=\; (O(s_1+s_2), O(t_1+t_2))\;.
\]
This allows us to induct on the depth of a formula: at each addition or multiplication gate, we apply \autoref{cl:addition} and the above multiplication trick to obtain the following proposition: 

If a depth-$\Delta$ formula of fanin-2 can compute $f$, then there exists some $c_1, c_2 \in \N$, such that
\[
\mu_1(f) \;=\; (c_1^{\Delta},c_2^{\Delta})\,.
\]
Now, given any $f \in \VPe$, we invoke the classical depth-reduction result of Brent~\cite{brent1974parallel} (see also~\cite[Lemma 5.5]{saptharishi2015survey}), states: If a family $(f_n)$ has polynomially bounded formula size, then there are formulas computing~$f_n$ that have size $\poly(n)$ and depth $\Delta = O(\log n)$.

Applying the proposition above, we conclude that $\VPe \subseteq \overline{\VPk{2}}$. Since it is immediate that $\VPk{2} \subseteq \VPe$ implies $\overline{\VPk{2}} \subseteq \overline{\VPe}$, we obtain the desired inclusion stated in~\cref{eq:vp2bar}.
\end{proof}

\begin{proof}[Proof sketch of~\cite{dutta2024power}.]
As seen earlier, the key building block in the proof of~\cite{bringmann2018algebraic} is the matrix $Q(f)$. That proof crucially relies on the identity: $fg = (\frac{1}{2}(f+g))^2 - (\frac{1}{2}(f-g))^2$, which fails over fields of characteristic 2. Consequently, the argument in~\cite{bringmann2018algebraic} does not extend to such fields.

The work of~\cite{dutta2024power} circumvents this barrier by avoiding the need to compute the product of arbitrary polynomials. Instead, to establish universality, it suffices to show how to compute $Q(fx)$ from $Q(f)$, where $x$ is a variable. The key advantage here is that $x$ is not an arbitrary polynomial, allowing us to explicitly use $2 \times 2$ matrices involving only constants and the variable $x$ in the computation of $Q(fx)$. In contrast, computing $Q(fg)$ directly would require access to both $f$ and $g$—which are typically available only as $Q$ matrices or through inductive constructions. This shift is captured in the central technical lemma from~\cite{dutta2024power}:

\medskip
\begin{claim}[{\cite[Lemma 12]{dutta2024power}}] \label{cl:char=2}
If $\mu_2(f) = (s,t)$, then $\mu_1(fx) = (2s+4, O(t))$.
\end{claim}

The proof of~\autoref{cl:char=2}] follows from the identity:
\[
Q(fx) + \Oh(\eps) \;=\;\begin{pmatrix}
\frac{1}{\epsilon} & 0 \\
0 & 1 
\end{pmatrix} \cdot (Q(f) + \Oh(\eps^2)) \cdot \begin{pmatrix}
\epsilon & 1\\
0 & 1 
\end{pmatrix}  \begin{pmatrix}
\frac{1}{\epsilon} & x \\
-1 & 1
\end{pmatrix}  (Q(f) + \Oh(\eps^2))  \cdot
\begin{pmatrix}
1 & 0 \\
1 & -\epsilon 
\end{pmatrix}\;.
\]  
Although \autoref{cl:char=2} does not enable the multiplication of two arbitrary polynomials, it is nonetheless sufficient to establish universality. Let $f$ be a polynomial with $m$ monomials. For any monomial $\x^{\e}$ appearing in $f$, repeated application of \autoref{cl:char=2} yields a sequence of $\mathcal{O}(2^{\deg(\x^{\e})})$ matrices that approximately computes $Q(\x^{\e})$. Therefore, by linearity, one can approximately compute $Q(f)$ using a sequence of $\mathcal{O}(m \cdot 2^{\deg(f)})$ matrices.

This completes the universality argument in~\cite{dutta2024power}.
\end{proof}
\begin{question}
Is $\overline{\VPe}=\overline{\VBP_2}$, over fields of characteristics $2$?
\end{question}

\subsection{Debordering Border Depth-4 Circuits and Beyond}
Looking at the success story of debordering border bounded depth-3 circuits, one can 
analogously look at the depth-4 model. A depth-$4$ circuit $\Sigma\Pi\Sigma\Pi$ computes a polynomial of the form
 \begin{equation}\label{eq:depth-4}
 f(\x)\;:=\;T_1 \;+\;\cdots\;+\; T_k\;,~~~\textrm{and}~~~T_i\;=\;\prod_{i=1}^d\,f_{i,j}\,,
 \end{equation}
where $f_{i,j}$ are $s$-sparse polynomials. We use~$\SPSP(k,n,d,s)$ to denote the set of all such depth-4 circuits. Further, when $\deg(f_{i,j}) \le \delta$, we denote the set by $\SPSP_{\delta}(k,n,d,s)$. Note that $s$ can be at most $\binom{n+\delta}{\delta}$. The size of the circuit is defined to be $s':=knds$. We will often write $\SPSP(k)$, when there is no restriction on the degree of $f_{i,j}$, and $\SPSP_{\delta}(k)$, when the bottom product fanin is bounded by $\delta$; in both cases the size is polynomially bounded. 

We introduce two more models: $\Sigma\wedge\Sigma\Pi_{\delta}$, and $\Sigma\Pi\Sigma\wedge$. The first computes polynomials of the form $\sum_{i\in [k]} g_i^{e_i}$, where $\deg(g_i) \le \delta$, while the second computes polynomials of the form $\sum_{i \in [k]} \prod_{j=1}^d f_{i,j}$, where each $f_{i,j} = \sum_{t \in [n]} f_{i,j,t}(x_t)$, can be written as a sum of univariate polynomials. We will denote them by $\Sigma\wedge\Sigma\Pi_{\delta}(k)$, and $\Sigma\Pi\Sigma\wedge(k)$.

\cref{thm:k=2} shows that $\approxbar{\SPSP_{1}(k)} \subseteq \VBP$, for any constant $k \in \N$. What happens when $\delta =2$? One can use the \textsf{DiDIL} technique as before, to conclude the following.

\medskip
\begin{theorem}[Implicit in~{\cite{dutta2022}}]
Let $k, \delta \in \N$. Then, 
\[
\approxbar{\Sigma\wedge\Sigma\Pi_{\delta}(k)} \;\subseteq\; \VBP \;\implies\; \approxbar{\SPSP_{\delta}(k)} \;\subseteq\; \VBP\;.
\]
\end{theorem}
We omit the proof here, as it closely parallels the case of $\delta = 1$. The key idea remains the same: the $\dlog$ operator transforms a product gate into a $\wedge$ gate, while potentially increasing the top fan-in to an unbounded value.

Unfortunately, the first containment is not known even when $\delta=2$, because unlike $\delta=1$, the model $\Sigma\wedge\Sigma\Pi_{\delta}(k)$ is not contained in $\ARO$. However, it is known that the class $\Sigma\wedge\Sigma\wedge(k)$ is contained in $\ARO$, even for polynomially bounded $k$. Circuits in $\Sigma\wedge\Sigma\wedge(k)$ compute polynomials of the form $\sum_{i \in [k]} g_i^{e_i}$, where $g_i = \sum_{j \in [n]} g_{i,j}(x_j)$. This containment can be shown using \cref{lem:duality}. Once this is established, the \textsf{DiDIL} technique can be used to derive the following result.

\medskip
\begin{theorem}[\cite{dutta2022}]
$\approxbar{\Sigma\Pi\Sigma\wedge(k)} \subseteq \VBP$.   
\end{theorem}
We leave this section with a couple of open questions.

\medskip
\begin{question}
Is $\approxbar{\Sigma\wedge\Sigma\Pi_{\delta}(k)} \;\subseteq\; \VBP$?   
\end{question}
We also introduce the following depth-5 model $\Sigma\wedge\Sigma\wedge\Sigma(k)$, which computed polynomials of the form $\sum_{i \in [k]} g_i^{e_i}$, where each $g_i$ can be computed by a polynomial-size $\SES$ circuit. We ask the following question.
\medskip
\begin{question}
Is $\overline{\Sigma\wedge\Sigma\wedge\Sigma(k)} \subseteq \VBP$?
\end{question}
We also ask whether we can extend the hierarchy theorem to bounded (top \& bottom fanin) depth-$4$ circuits. In particular,

\medskip
\begin{question}
Let $\delta \in \N$. Is $\overline{\SPSP_{\delta}(1)} \;\subsetneq\; \overline{\SPSP_{\delta}(2)} \;\subsetneq\; \overline{\SPSP_{\delta}(3)} \cdots$, where the respective gaps are exponential?   
\end{question}
Clearly, $\delta=1$ holds from~\cref{thm:k=2lb}, and it is unclear what happens even when $\delta=2$.


\subsection{Demystifying Border of \texorpdfstring{$3\times 3$}{} Determinants}

Consider the $3\times 3$ symbolic determinant
\[
\det_3 := \begin{pmatrix}
    x_1 & x_2 & x_3 \\
    x_4 & x_5 & x_6 \\
    x_7 & x_8 & x_9
  \end{pmatrix}  \in \C[x_1,\cdots,x_9]\;.
\]
We consider it as a homogeneous form of degree~$3$ on the space~$\bC^{3\times 3}$ of~$3\times 3$ matrices, denoted~$W$.
Let~$\C[W]_3$ denote the $165$-dimensional space of all homogeneous forms of degree~$3$ on~$W$.
The group~$G:=\GL(W)$ acts on~$\C[W]_3$ by right composition. 

For a nonzero~$f \in \C[W]_3$,
let~$\Omega(P)$ denote the (projective) orbit of~$P$, namely the set of all~$[P \circ a] \in \bP(\bC[W]_3)$, with~$a\in \GL(W)$.
The \emph{boundary} of the orbit of~$P$, denoted~$\partial\Omega(P)$, is~$\overline{ \Omega(P) } \setminus \Omega(P)$,
where~$\overline{\Omega(P)}$, denoted also~$\overline\Omega(P)$, is the Zariski closure of the orbit in~$\bP(\bC[W]_3)$. Understanding the boundary of such orbit closures is a key goal in geometric complexity theory (GCT), where one studies how structured polynomials (like $\det_n$ or $\per_n$) degenerate under linear transformations. 

In~\cite{huttenhain2016boundary} H\"uttenhain and Lairez characterized $\partial \Omega(\det_3)$ completely.

\medskip
\begin{theorem}[{\cite{huttenhain2016boundary}}]\label{thm:det3}
 The boundary of the orbit closure of $\det_3$, namely $\partial\Omega(\det_3)$ has exactly two irreducible components. These components are given by the closures of the orbits of the following two polynomials:   
 \begin{itemize}
  \item[(1)] A trace-zero symbolic matrix determinant:
  \[
  P_1 := \det
  \begin{pmatrix}
  x_1 & x_2 & x_3 \\
  x_4 & x_5 & x_6 \\
  x_7 & x_8 & -x_1 - x_5
  \end{pmatrix},
  \]
  which corresponds to a degeneration where the symbolic matrix is constrained to have trace zero.
  \item[(2)] A special quadric-in-cubic form:
  \[
  P_2 := x_4 x_1^2 + x_5 x_2^2 + x_6 x_3^2 + x_7 x_1 x_2 + x_8 x_2 x_3 + x_9 x_1 x_3,
  \]
  representing a structured degeneration where variables appear as coefficients of quadratic forms.
\end{itemize}
\end{theorem}
It is straightforward to verify that $\dim(\Omega(\det_3)) = 64$, whereas $\dim(\Omega(P_1)) = \dim(\Omega(P_2)) = 63$. In what follows, we will only demonstrate that both $P_1$ and $P_2$ lie in the boundary and that their orbit closures form irreducible components. For the full proof that these are in fact the \emph{only} components of the boundary, along with further details, we refer the reader to~\cite{huttenhain2016boundary} and to H\"uttenhain's beautiful PhD thesis~\cite{huttenhain2017geometric}.

\medskip
\begin{proof}[Proof sketch of \cref{thm:det3}]
Define the rational map $\varphi : \mathbb{P}(\mathrm{End}(W)) \dashrightarrow \mathbb{P}(\Sym^3(W^*))$, via
\[
\varphi : [a] \mapsto [\det_3 \circ a]\;.
\]
This image on the open subset of invertible $a$ is the orbit $G \cdot \det_3$. Let also~$Z$ be the irreducible hypersurface of~$\bP(\End(W))$
\[ Z \;:=\; \left\{ [a]\in \bP(\End(W))\;\mid\; \text{det}(a) = 0 \right\}. \]
By definition, $\Omega(\detdrei) = \varphi(\bP(\End(W)) \setminus Z)$.
Let~$\varphi(Z)$ denote the image of the set of points of~$Z$ where~$\varphi$ is defined. The following claim proves (1) of \cref{thm:det3}.

\begin{claim}\label{cl:claim-1-det3}
$\overline{\varphi(Z)}$ is an irreducible component of~$\partial \Omega(\detdrei)$, and furthermore~$\overline{\varphi(Z)} = \overline\Omega(P_1)$.    
\end{claim}
\begin{proof}[Proof sketch of \autoref{cl:claim-1-det3}]
Consider the function~$\nu : \bC[W]_3\to \N$ which associates to~$P$ the
  dimension of the linear subspace of~$\bC[W]_2$ spanned by the partial
  derivatives~$\frac{\partial P}{\partial x_1}$,\,\dots\,, $\frac{\partial
  P}{\partial x_9}$. The function~$\nu$ is invariant under the action
  of~$\GL(W)$. 
  Because every form in~$\varphi(Z)$
  can be written as a polynomial in at most~$8$ linear forms, $\nu(P)\leq 8$ for all~$P\in\varphi(Z)$.
  On the other hand, $\nu(\det_3)=9$ and so~$\nu(P) = 9$ for any~$P\in\Omega(\detdrei)$.
  This shows that~$\varphi(Z)\cap\Omega(\detdrei) = \varnothing \implies \overline{\varphi(Z)} \subset \partial\Omega(\detdrei)$.
  Moreover~$\overline{\varphi(Z)}$ is irreducible because~$Z$ is.
  
  Clearly~$P_1 \in \varphi(Z)$ and  further one can show that $\Omega(P_1)$ has dimension~$63$.
  Since
  \[ \overline\Omega(P_1) \;\subset\; \overline{\varphi(Z)} \;\subset\; \partial\Omega(\detdrei), \]
  they all three have dimension~$63$ and $\overline\Omega(P_1) = \overline{\varphi(Z)}$
  because the latter is irreducible. This gives a component of~$\partial \Omega(\detdrei)$.
\end{proof}
The following claim shows (2) of \cref{thm:det3}. 
\begin{claim}\label{cl:claim-2-det3}
  The orbit closure~$\overline\Omega(P_2)$ is an irreducible component of~$\partial \Omega(\detdrei)$ and is distinct from~$\overline\Omega(P_1)$.
\end{claim}

\begin{proof}[Proof of \autoref{cl:claim-2-det3}]
 We first prove that~$[P_2] \in \partial\Omega(\detdrei)$.
  Let
  \[ A = \begin{pmatrix} 0 & x_1 & -x_2 \\ -x_1 & 0 & x_3 \\ x_2 & - x_3 & 0.  \end{pmatrix}
    \text{ and }
    S = \begin{pmatrix} 2x_6 &  x_8 & x_9 \\ x_8 & 2x_5 & x_7 \\ x_9 & x_7 & 2x_4 \end{pmatrix}.
    \]
 Since~$\det(A) = 0$, by Jacobi's formula, it is not hard to argue that 
  the projective class of the polynomial~$\det(A+\eps S)$
  tends to~$[\mop{Tr}( \mop{adj}(A) S )]$ when~$\eps\to 0$, and by construction, this limit is a point in~$\overline \Omega(\detdrei)$.
  Besides, for $u = (x_3,x_2,x_1)$, we have
  \[ \mop{Tr}( \mop{adj}(A) S ) \;=\; u S u^T = 2P_2\;\implies [P_2] \in \overline\Omega(\detdrei)\;.\] 
  Yet~$[P_2]$ is not in~$\Omega(\detdrei)$, because its orbit has dimension~$63$, whereas the orbit of every point of~$\Omega(\detdrei)$
  is~$\Omega(\detdrei)$ itself.
  Therefore~$[P_2]$ is in the boundary~$\partial\Omega(\detdrei)$.
  Since~$\Omega(P_2)$ has dimension~$63$, this gives a compoment of~$\partial\Omega(\detdrei)$.
  It remains to show that~$[P_2]$ is not in~$\Omega(P_1)$, and indeed~$\nu(P_2) = 9$ whereas~$\nu(P_1) =8$, where~$\nu$ is the function introduced in the proof of \autoref{cl:claim-1-det3}.
\end{proof}
\let\qed\relax
\end{proof}
We leave this section by asking the following question.

\medskip
\begin{question}
Characterize $\partial \Omega(\det_4)$.
\end{question}
\section{Acknowledgements}
The authors thank the organizers of the {\em Workshop on Recent Trends in Computer Algebra} (\href{https://rtca2023.github.io/}{RTCA 2023}) at the Institut Henri Poincaré, Paris, for inviting the first author to give a survey talk on debordering and subsequently for inviting us to write this survey. They also thank the organizers and participants of the 8th Workshop on Algebraic Complexity Theory (\href{https://qi.rub.de/events/wact25/}{WACT 2025}) at Ruhr University Bochum for the engaging discussions, which helped clarify key aspects of the survey that warranted further emphasis.

P.D.~gratefully acknowledges support from the Jane Street Research Fellowship at the Simons Institute for the Theory of Computing, University of California, Berkeley. He is also supported by SUG from the College of Computing and Data Science, Nanyang Technological University Singapore, titled ``Debordering and Derandomization in Algebraic Complexity".   

V.L.~acknowledges support by the European Research Council (ERC Grant Agreement No. 101040907 ``SYMOPTIC''). Views and opinions expressed are however, those of the author(s) only and do not necessarily reflect those of the European Union or the European Research Council Executive Agency. Neither the European Union nor the granting authority can be held responsible for them.

The authors thank Fulvio Gesmundo, Christian Ikenmeyer, and Nitin Saxena for many helpful discussions.

\bibliographystyle{alpha}
\bibliography{references.bib}
\end{document}